\documentclass[sigconf,nonacm]{acmart}
\usepackage{float}
\usepackage{booktabs} 
\usepackage[linesnumbered, ruled, vlined]{algorithm2e}
\usepackage{amsmath}
\usepackage{caption}
\usepackage{subcaption}
\usepackage{graphicx}
\usepackage{tcolorbox}
\usepackage{multicol}
\setcopyright{acmcopyright}

\acmISBN{979-8-4007-0243-3/24/04}
\acmConference[EDBT'26]{EDBT/ICDT 2026 Joint Conference}{24th March - 27th March, 2026}{Tampere, Finland}
\acmYear{2024}
\copyrightyear{2024}
\acmArticle{4}
\acmPrice{15.00}

\begin{document}
\title{Approximately Bisubmodular Regret Minimization in Billboard and Social Media Advertising}

\author{Dildar Ali}
\affiliation{%
 \institution{Department of Computer Science and Engineering, Indian Institute of Technology Jammu}
 \city{Jammu}
 \state{Jammu \& Kashmir}
 \country{India}
 \postcode{181221}}
 \email{2021rcs2009@iitjammu.ac.in}

\author{Suman Banerjee}
\affiliation{%
 \institution{Department of Computer Science and Engineering, Indian Institute of Technology Jammu}
 \streetaddress{1 Th{\o}rv{\"a}ld Circle}
 \city{Jammu}
  \state{Jammu \& Kashmir}
 \country{India}}
\email{suman.banerjee@iitjammu.ac.in}

\author{Yamuna Prasad}
\affiliation{%
 \institution{Department of Computer Science and Engineering, Indian Institute of Technology Jammu}
 \streetaddress{1 Th{\o}rv{\"a}ld Circle}
 \city{Jammu}
  \state{Jammu \& Kashmir}
 \country{India}}
\email{yamuna.prasad@iitjammu.ac.in}

\begin{abstract}
We study the problem of minimizing regret in multi-mode advertisement settings, where an influence provider allocates advertising resources such as social network seeds and billboard slots to multiple advertisers with specified influence demands and payments. Unlike prior work focusing on a single mode of advertising, we consider the interplay between online and offline modes and introduce a novel regret model that captures their interaction effect. This leads to a regret minimization problem that is non-monotone, non-submodular, and NP-hard to approximate within any constant factor. To address this, we propose a monotone, approximately bisubmodular influence model and develop two algorithmic solutions: Projected Subgradient Method based on the Lovász extension of the regret function, and an Approximate Bisubmodular Local Search algorithm with provable guarantees. Experiments on large-scale real-world datasets, including billboard and trajectory data from major U.S. cities, as well as social network graphs, demonstrate that our methods outperform existing baselines in minimizing total regret while satisfying advertiser demands. Our framework is broadly applicable to other resource allocation scenarios beyond advertising.
\end{abstract}

\maketitle

\section{Introduction}
Almost all commercial houses use advertising as a mechanism to promote their product and create a customer base. As mentioned in recent marketing literature, a commercial house spends around $7-10\%$ of its annual revenue on advertising\footnote{\url{https://www.lamar.com/howtoadvertise/Research/}}. Two popular advertising approaches are advertising through digital billboards and social networks. In a billboard advertisement, a commercial house displays content with the hope that the people nearby will look at the content and be influenced by it. The other way to advertise is through social media. Many internet giants, including Google and Facebook, earn a significant amount of revenue through social network advertising. In this method, a number of highly influential users are chosen, and they are influenced externally. In general, advertisers request the influence provider for the required influence demand in exchange for some payment. Now, if the influence provider provides the demanded influence, then he will receive full payment; otherwise, a partial payment. This leads to a loss of influence provider, and this loss is formulated as \textit{regret}. Most of the existing influence maximization problems are solved from the advertiser's perspective \cite{zhang2020towards,kempe2003maximizing}, and few are solved from the influence provider's perspective for billboard and social network, separately \cite{zhang2021minimizing,ali2024minimizing,ali2024regret,aslay2015viral}. In this work, we consider both the billboard and the social network, and formulate the Regret Minimization Problem (formally defined in Section \ref{Sec:RR}) from the perspective of the influence provider. The general applicability of our problem will be illustrated later in this section.
\paragraph{\textbf{Our Observation.}}
Existing studies deal with single and multi-advertiser settings that share a common objective: (a) to help advertisers achieve the largest influence under budget constraint \cite{zhang2020towards,ali2024effective,kempe2003maximizing}. (b) to minimize the regret of an influence provider while satisfying the advertiser's influence demand \cite{ali2023efficient,ali2024regret,ali2024minimizing,zhang2021minimizing}. A more challenging, unexplored scenario involves advertisers submitting daily proposals with influence demands across both online and offline modes, along with payments conditional on meeting those demands.
\paragraph{\textbf{Motivation.}} 
In the existing literature, to calculate regret in billboard advertisement \cite{zhang2021minimizing,ali2023efficient,ali2024minimizing,ali2024regret,ali2024toward} and Social media advertisement \cite{sharma2024minimizing,aslay2014viral}, a separate regret model exists, and the regret models are non-monotone and non-submodular. In typical advertisements, the advertisers have a specific amount of influence demand, and the regret of the influence providers can be computed using these two regret models, and the sum of these individual regrets will be the total regret of the influence provider. However, this approach will not work when we consider two different modes of advertising due to the \textit{Interaction Effect} (See Definition \ref{Def:Interaction_Effect}) between social networks and billboard slots and will generate different regret compared to traditional methods  \cite{sharma2024minimizing,zhang2021minimizing}. Due to the non-monotonicity of existing regret models, performance guarantees are unattainable. Thus, a monotone regret model that accounts for both social and billboard advertisements, including interaction effects, is essential. This motivates the study of the \textit{Regret Minimization Problem (RM)} to enable effective solutions.

\begin{example}\label{Example:1}
Assume there are four advertiser $\mathcal{A}=\{a_{1}, a_{2}, a_{3}, a_{4}\}$, Six billboard slots $\mathcal{BS}= \{bs_{1}, bs_{2}, \ldots, bs_{6}\}$ (see Table \ref{ETable:1}), seven seed nodes $\mathcal{P} = \{ p_{1},p_{2}, \ldots, p_{7}\}$ (see Table \ref{ETable:2}) and an influence provider $\mathcal{Z}$ with influence demand from social network and billboard slots from the advertisers as shown in Table \ref{ETable:3}. Now, we consider two cases: first, the regret is calculated separately, and the aggregated sum of the regret is presented. The allocation of slots to the advertisers is as follows: $a_{1} = \{bs_{1}, bs_{2}\}$, $a_{2} = \{bs_{5}\}$, $a_{3} = \{bs_{3}, bs_{4}\}$, $a_{4} = \{bs_{6}\}$ and allocation of seed nodes to the advertisers is $a_{1} = \{p_{1}, p_{2}\}$, $a_{2} = \{p_{5}\}$, $a_{3} = \{p_{3}, p_{4}\}$, $a_{4} = \{p_{6}, p_{7}\}$. The regret from billboard slots for the advertiser $a_{1}, a_{2}, a_{3}$, and $a_{4}$ are $14.23$, $15.55$, $9.23$ and $15.23$, respectively. Similarly, for the social network, regrets are $74.63$, $47.10$, $92.30$, and $73.1$, respectively. In the second case where we consider the \textit{Interaction Effect} and use the proposed regret model (see Definition \ref{Def:Combine_Regret_Model}) and the regret for the advertiser $a_{1}, a_{2}, a_{3}$, and $a_{4}$ are $77.34$, $53.77$, $86.03$ and $78$, respectively. So, the total regret for all the advertisers in the first and second cases will be $341.37$ and $295.14$.

\begin{table}[!h]
\begin{center}
\begin{minipage}{0.48\textwidth}
\small
   \centering
   \begin{tabular}{| c | c | c | c | c | c | c |}
   \hline
   $\mathcal{BS}_{i}$ & $bs_{1}$ & $bs_{2}$ & $bs_{3}$ & $bs_{4}$ & $bs_{5}$ & $bs_{6}$ \\ \hline
   $\mathcal{I}(bs_{i})$ & 2 & 4 & 3 & 1 & 6 & 5 \\ \hline
   Cost & \$6 & \$12 & \$9 & \$3 & \$18 & \$15 \\ \hline
   \end{tabular}
   \caption{\label{ETable:1} Billboard Info.}
\end{minipage}
\hfill

\begin{minipage}{0.48\textwidth}
\small
   \centering
   \begin{tabular}{| c | c | c | c | c | c | c |c|}
   \hline
   $\mathcal{P}_{i}$ & $p_{1}$ & $p_{2}$ & $p_{3}$ & $p_{4}$ & $p_{5}$ & $p_{6}$ & $p_{7}$ \\ \hline
   $\mathcal{I}^{\mathcal{G}}(p_{i})$ & 10 & 13 & 13 & 14 & 15 & 12 & 10 \\ \hline
   Cost & \$50 & \$65 & \$65 & \$70 & \$75 & \$60 & \$50\\ \hline
   \end{tabular}
   \caption{\label{ETable:2} Seed Node Info.}
\end{minipage}
\hfill

\begin{minipage}{0.48\textwidth}
\small
   \centering
   \begin{tabular}{ | c | c | c | c | c |}
   \hline
   Advertiser($\mathcal{A}$) & $a_{1}$ & $a_{2}$ & $a_{3}$ & $a_{4}$ \\ \hline
   Demand ($\sigma_{i}$) & 10 & 10 & 5 & 10 \\ \hline
   Budget($\mathcal{Y}_{i}$) & \$20 & \$20 & \$15 & \$20 \\ \hline
   Demand ($\Omega_{i}$) & 30 & 20 & 35 & 25 \\ \hline
   Budget($\mathcal{X}_{i}$) & \$120 & \$75 & \$150 & \$130 \\ \hline
   \end{tabular}
   \caption{\label{ETable:3} Advertiser Info.}
\end{minipage}
\end{center}
\end{table}
\end{example}

\paragraph{\textbf{Our Problem.}}
Motivated by our observations, we introduce the allocation problem from the perspective of the influence provider who is responsible for the allocation of slots or seeds to the advertisers. The influence provider owns a large number of slots and seeds. Each advertiser seeks a subset of slots or seeds, or both, with aggregated influence reaching their influence demand. The influence provider will receive full payment if he satisfies the influence demand of the advertiser; otherwise, he will receive a partial payment. This regret affects the profit of the influence provider. So, we introduce a novel \texttt{regret model} to guide the influence provider in assigning slots and seeds to the advertisers. 


\paragraph{\textbf{Our Solutions.}}
Since \textit{RM} problem is generally intractable, we propose a projected gradient method (PGM) and a greedy-based local search technique that offer end-users different trade-offs between computational efficiency and achievable regret. The first approach, PGM, minimizes the regret function using subgradient descent on its Lovász extension over fractional slot and seed vectors, with projections for feasibility. The second approach allocates slots or seeds by the highest marginal reduction in regret per unit.

\paragraph{\textbf{Empirical Evaluation.}} 
We use real-world and synthetic datasets to measure how proposed algorithms behave with respect to different trajectories in a city. The real-world dataset consists of billboard and trajectory information of the two major cities, New York and Los Angeles, USA. The synthetic dataset is generated by using the pattern of real trajectories and billboard datasets. We define demand-supply and average-individual demand ratios to model macro-level (overall demand) and micro-level (advertiser size) scenarios, capturing diverse real-world settings. Finally, we present evaluation results and insights on the practical benefits of different deployment strategies for the host.

\paragraph{\textbf{General Applicability.}}
The regret formulation in this paper applies broadly to scenarios where companies allocate resources to meet demand, such as trucks, store locations, or staff. Under-provisioning leads to unmet demand, while over-provisioning wastes resources. Though specific objectives may vary, our techniques remain applicable with minor adjustments. For example, in cloud computing, providers allocate resources to clients. Under-provisioning causes performance issues; over-provisioning wastes resources. Techniques from this paper can be adapted with minor changes to handle such resource allocation problems efficiently.

\paragraph{\textbf{Relevant Studies}} Several studies studied in influence maximization in billboard \cite{zhang2020towards,wang2022data,Jali2025influential} as well as social network advertisements \cite{chen2009efficient,chen2010scalable,guo2013personalized}. In this direction, Ali et al. \cite{ali2025influential} first introduce the joint tag and slot selection problem and formulate this problem as a bisubmodular influence maximization problem. To solve this, they introduce orthant-wise greedy maximization approaches. There exists literature \cite{ali2024effective,ali2024multi} that considers tag assignment to the billboard slots so that the total influence is maximized for both single and multi-advertiser settings. Now, in the case of social networks, there exists an influence maximization literature \cite{10.1145/1557019.1557047,bharathi2007competitive,chen2010scalable,jung2012irie} that focuses on finding a subset of nodes in the social network that maximizes total influence. Furthermore, in this direction, for influence maximization problems, Ali et al. \cite{ali2025fairness} introduce the maximin fairness notion for billboard advertisements, and Rui et al. \cite{rui2025efficient} apply it to social network influence maximization.

\par Few studies in the context of billboard advertising consider the regret minimization problem from the perspective of the influence provider caused by providing influence to the advertiser. First, Ashley et al. \cite{aslay2014viral} introduce a new problem domain that involves allocating social network users to advertisers to promote commercial posts. Recently, Sharma et al. \cite{sharma2024minimizing} studied the regret minimization problem in social networks and introduced a greedy-based seed set allocation approach. In the context of billboard advertising, Zhang et al. \cite{zhang2021minimizing} studied the first regret minimization problem, and they proposed several heuristic solutions. In addition, Ali et al. \cite{ali2023efficient,ali2025toward} studied regret minimization problems extensively and proposed several greedy and efficient randomized algorithms in a multi-advertiser setting. In addition, Ali et al. \cite{ali2024minimizing,ali2024regret} extend their work for zonal influence constraints. They consider the zone-specific influence demand of the advertisers and introduce greedy-based solution methodologies to minimize the regret from the influence provider's perspective. Next, we summarized our contributions.

\paragraph{\textbf{Our Contributions.}} 
To the best of our knowledge, this is the first work to address regret minimization by jointly considering billboard and social network advertising. Our key contributions are:
\begin{itemize}
    \item We formulate a regret minimization problem in a multi-advertiser setting across billboard and social networks.
    \item We prove the problem is NP-hard and inapproximable within any constant factor.
    \item We propose two solution methods: the Projected Gradient Method and the Approximate Bisubmodular Local Search.
    \item We conduct experiments on real-world datasets to show the effectiveness of our methods against baselines.
\end{itemize}

\paragraph{\textbf{Organization of the Paper}} 
Rest of the paper is organized as follows. Section \ref{Sec:BPD} describes required background concepts and defines the problem formally. Section \ref{Sec:PS} describes the proposed solution approach with illustration and analysis. Section \ref{Sec:EE} contains the experimental evaluations of the proposed solution approaches. Section \ref{Sec:CFD} concludes this study and gives future research directions.

\section{Background and Problem Definition} \label{Sec:BPD}
In this section, we describe the background of the problem and formally define our problem. For any positive integer $k$, $[k]$ denotes the set $\{1,2, \ldots, k\}$ and for any two positive integers $x$ and $y$, with $x \leq y$, $[x,y]$ denotes the set $\{x, x+1, \ldots, y\}$. Initially, we start by describing the set functions and their properties.

\subsection{Set Function and Its Properties}
Consider the set $\mathcal{V}$ with $n$ elements and a function $\mathcal{F}$ defined on the set $\mathcal{V}$ i.e., $\mathcal{F}: 2^{\mathcal{V}} \longrightarrow \mathbb{R}$. The $\mathcal{F}$ is a set function and is normalized if $\mathcal{F}(\emptyset) = 0$. Now, for given $A \subseteq B \subseteq \mathcal{V}$, $\mathcal{F}$ is said to be \emph{nonnegative} if any $B \subseteq \mathcal{V}$, $\mathcal{F}(B) \geq 0$, monotone if $\mathcal{F}(A) \leq \mathcal{F}(B)$ and submodular if for all $s \in \mathcal{V} \setminus B$, $\mathcal{F}(A \cup \{s\}) - \mathcal{F}(A) \geq \mathcal{F}(B \cup \{s\}) - \mathcal{F}(B)$. Further, the submodular function is extended to biset functions in which two arguments are present. Let a biset function be defined in the ground set $\mathcal{V}_{1} \times \mathcal{V}_{2}$, that is, $2^{{\mathcal{V}_1} \times {\mathcal{V}_2}} \longrightarrow \mathbb{R}$. The properties of the biset function $f$ are as follows. The biset function $f$ is said to be normalized if $f(\emptyset, \emptyset) = 0$ and is said to be monotone if for all $(X_{1},X_{2}) \in 2^{{\mathcal{V}_1} \times {\mathcal{V}_2}}$, where $x \in \mathcal{V}_{1} \setminus X_{1}$, $y \in \mathcal{V}_{2} \setminus X_{2}$ and $f(X_{1} \cup \{x\},X_{2}) \geq f(X_{1},X_{2})$ and $f(X_{1},X_{2} \cup \{y\}) \geq f(X_{1},X_{2})$ hold. The biset function $f$ is said to be bisubmodular if for any $(X_{1},X_{2}), (X^{'}_{1},X^{'}_{2}) \in 2^{{\mathcal{V}_1} \times {\mathcal{V}_2}}$, where $X_{1} \subseteq X^{'}_{1} \subseteq \mathcal{V}_{1}$ and $X_{2} \subseteq X^{'}_{2} \subseteq \mathcal{V}_{2}$, also for all $x \in \mathcal{V}_{1} \setminus X^{'}_{1}$ and $y \in \mathcal{V}_{2} \setminus X^{'}_{2}$ the condition holds: $f(X_{1} \cup \{x\},X_{2}) - f(X_{1},X_{2}) \geq f(X^{'}_{1} \cup \{x\}, X^{'}_{2}) - f(X^{'}_{1}, X^{'}_{2})$ and $f(X_{1},X_{2} \cup \{y\}) - f(X_{1},X_{2}) \geq f(X^{'}_{1}, X^{'}_{2}\cup \{y\}) - f(X^{'}_{1}, X^{'}_{2})$ \cite{singh2012bisubmodular}. Next, we define the $\epsilon$-approximate bisubmodular in Definition \ref{Def:epsilon-approximate-bisubmodular}.

\begin{definition}[$\varepsilon$-Approximately Bisubmodular]\label{Def:epsilon-approximate-bisubmodular}
A biset function $f: 2^{{\mathcal{V}_1} \times {\mathcal{V}_2}} \to \mathbb{R}$ is said to be $\epsilon$-approximately bisubmodular if for any
$X_{1} \subseteq X^{'}_{1} \subseteq \mathcal{V}_{1}$ and $X_{2} \subseteq X^{'}_{2} \subseteq \mathcal{V}_{2}$, also for all $x \in \mathcal{V}_{1} \setminus X^{'}_{1}$ and $y \in \mathcal{V}_{2} \setminus X^{'}_{2}$ the condition holds: 
\begin{equation}
f(X_{1} \cup \{x\},X_{2}) - f(X_{1},X_{2}) \geq f(X^{'}_{1} \cup \{x\}, X^{'}_{2}) - f(X^{'}_{1}, X^{'}_{2}) - \varepsilon,  
\end{equation}
\begin{equation}
f(X_{1},X_{2} \cup \{y\}) - f(X_{1},X_{2}) \geq f(X^{'}_{1}, X^{'}_{2}\cup \{y\}) - f(X^{'}_{1}, X^{'}_{2}) - \varepsilon
\end{equation}
The additive term $\varepsilon \ge 0$ quantifies the deviation from exact bisubmodularity when $\varepsilon = 0$, the function is exactly bisubmodular.
\end{definition}

\subsection{Submodular Minimization}
Consider the set $\mathcal{V}$ with $n$ elements and a function $\mathcal{F}$ defined on the set $\mathcal{V}$ i.e., $\mathcal{F}: 2^{\mathcal{V}} \longrightarrow \mathbb{R}$. The $\mathcal{F}$ is a set function and is normalized if $\mathcal{F}(\emptyset) = 0$. Now, for given $A \subseteq B \subseteq \mathcal{V}$, $\mathcal{F}$ is said to be \emph{nonnegative} if any $B \subseteq \mathcal{V}$, $\mathcal{F}(B) \geq 0$, monotone if $\mathcal{F}(A) \leq \mathcal{F}(B)$ and submodular if for all $s \in \mathcal{V} \setminus B$, $\mathcal{F}(A \cup \{s\}) - \mathcal{F}(A) \geq \mathcal{F}(B \cup \{s\}) - \mathcal{F}(B)$. Minimizing a submodular function $\mathcal{F}$ is equivalent to minimizing its Lovász extension \cite{lovasz1983submodular}, a continuous, convex extension defined over the hypercube $[0,1]^d$. This extension is convex if and only if $\mathcal{F}$ is submodular.
\begin{definition}[Lovász Extension]
For a normalized set function $\mathcal{F}$, the Lovász extension $f_L: \mathbb{R}^d \to \mathbb{R}$ is given by, $f_L(s) = \sum_{k=1}^d s_{j_k} F(j_k \mid S_{k-1})$,
where $s_{j_1} \geq \cdots \geq s_{j_d}$ are the sorted coordinates of $s$, and $S_k = \{j_1, \dots, j_k\}$.
\end{definition}

Now, minimizing $f_{L}$ is equivalent to minimizing the function $\mathcal{F}$. Moreover, when $\mathcal{F}$ is submodular and $k$ is a subgradient of $f_{L}$ at any $s \in \mathbb{R}^{d}$ can be effectively computed by decreasing the order of sorted $s$ and taking $k_{j_k} = \mathcal{F}(j_k \mid S_{k-1})$ for all $k \in \mathcal{V}$ \cite{edmonds2003submodular}. The relation between convexity and submodularity allows for generic convex optimization algorithms, and it can be used to minimize $\mathcal{F}$. Although in the case of submodular and approximately submodular functions, how this relation is affected was studied by existing studies \cite{el2020optimal}. However, it has been unclear how these relations are affected if the function is only approximately bisubmodular. In this paper, we give an answer to this question.
\subsection{Billboard Advertisement}
Billboard advertisements require three key components: an influence function, a trajectory, and a billboard database. A trajectory database $\mathbb{D}$ contains the location information of a moving user, and $[T_1, T_2]$ is the time duration for which the user movement occurs. Consider $\mathbb{D}$, which contains $m$ tuples in the form $<\mathcal{U}^{'},\ell, [t_{1},t_{2}]>$, which denotes a set of people $\mathcal{U}^{'} \subseteq \mathcal{U}$ who were at the location $\ell$ for a period of time $[t_{1},t_{2}]$. So, for every tuple $ q \in \mathbb{D}$, the associated interval $[t_{x},t_{y}] \in [T_{1},T_{2}]$, and $t_{x} \leq t_{y}$. The billboard database $\mathbb{B}$ contains the information of billboards placed in a city represented in the form of a tuple $<b_{id},\ell, cost>$, which signifies the billboard $b_{id}$ is placed at location $\ell$ associated with some cost. Recently, commercial houses have hired billboards for a certain duration, i.e., slots. Considering $\Delta$ as the slot duration for each billboard, the number of slots will be $\frac{T_2-T_1}{\Delta}$. A set of billboard slots can be denoted by $\mathbb{BS}=\{(b_i, [t,t+\Delta]): b_i \in \mathbb{B} \text{ and } t \in \{T_1, T_1+ \Delta, T_1+ 2 \Delta, \ldots, T_2-\Delta \}\}$ and represented in the form of a tuple $<bs_{id}, \ell, time\_slot, cost>$, which represents billboard slot $bs_{id}$ is placed at location $\ell$ with a slot duration. The associated cost of a slot is formalized by a cost function $\mathcal{C}_{BS}: \mathbb{BS} \longrightarrow \mathbb{R}^{+}$. Now, consider that user $u_{i} \in \mathcal{U}$ is present for the duration $[t_{c},t_{d}]$ in a location where the billboard is placed and currently running an advertisement for duration $[t_{a},t_{b}]$. If $[t_{c},t_{d}] \cap [t_{a},t_{b}] \neq \emptyset$, then we can say the user $u_{i}$ is influenced with some probability. In the existing literature \cite{zhang2020towards}, several approaches are proposed to calculate this influence probability. However, based on LAMAR \footnote{\url{https://www.lamar.com/howtoadvertise/Research/}} influence metrics, we incorporate panel size and exposure frequency into our influence model. Specifically, for each $bs_i \in \mathbb{BS}$,  $u_j \in \mathcal{U}$, if $bs_i$ can influence $u_j$, the $\text{Pr}(u_j,bs_i)$ can be define as $\text{Pr}(t_j,bs_i) = \frac{\text{size}(bs_i)}{A},$ where $A > \max_{bs_i \in \mathbb{BS}} \text{size}(bs_i)$, to reflect influence based on panel size. Now, for a given subset of billboard slots $\mathcal{S} \in \mathbb{BS}$, the influence of $\mathcal{S}$ can be denoted as $\mathcal{I}(\mathcal{S})$ and defined in Definition \ref{Def:1}.

\begin{definition} [Influence of Billboard Slots] \label{Def:1}
Given a subset of billboard slots $\mathcal{S} \subseteq \mathbb{BS}$, its influence $\mathcal{I}(\mathcal{S})$ can be computed using Equation \ref{Eq:1}.
\begin{equation} \label{Eq:1}
{\small
\mathcal{I}(\mathcal{S}) = \underset{u_j \in \mathbb{D}}{\sum} 1- \underset{bs_i \in \mathcal{S}}{\prod} (1- Pr(u_j,bs_i))
}
\end{equation}
\end{definition}

\begin{theorem}\label{Th:I} \cite{zhang2020towards}
The influence function $\mathcal{I}()$ is non-negative, monotone, and submodular.
\end{theorem}

\subsection{Social Network Advertisement}
Consider, the social network information for the people $\mathcal{U} =\{u_{1},u_{2},\ldots,u_{n}\}$ represented as a weighted undirected graph $\mathcal{G}(\mathcal{U},\mathcal{E}, \mathcal{W})$. The edge $\mathcal{E}(\mathcal{G})$ represents a binary relationship among users, and the edge weight function $\mathcal{W}$ is associated with its corresponding influence probability, i.e., $\mathcal{W}: \mathcal{E}(\mathcal{G}) \longrightarrow (0,1]$. If an edge $(u_{i}u_{j}) \notin \mathcal{E}(\mathcal{G})$, then the influence probability $\mathcal{W}_{(u_{i}u_{j})} = 0$. To study the diffusion process of social networks, we consider \emph{Independent Cascade Model} stated in Definition \ref{Def:2}. 

\begin{definition}[Independent Cascade Model]\label{Def:2}
The Independent Cascade (IC) Model is a stochastic diffusion model defined over a undirected graph $\mathcal{G}(\mathcal{U},\mathcal{E}, \mathcal{W})$, where each edge $(u_{i}u_{j}) \in \mathcal{E}$ is associated with an influence probability $\mathcal{W}_{(u_{i}u_{j})} \in (0,1]$. 
\end{definition}

This can be formalized by an influence function $\mathcal{I}^{\mathcal{G}}$ which maps each node in the network to its corresponding influence value, i.e., $\mathcal{I}^{\mathcal{G}}: 2^{\mathcal{U}} \longrightarrow \mathbb{R}^{+}$ with $\mathcal{I}^{\mathcal{G}}(\emptyset) = 0$.

\begin{theorem}\label{Th:I^G} \cite{kempe2003maximizing}
The influence function $\mathcal{I}^{\mathcal{G}}$ is non-negative, monotone, and submodular under the \textsc{Independent Cascade Model}.
\end{theorem}

\subsection{Influence Model}
In our context, to jointly measure the influence of billboards and social networks, we propose an influence model, as defined in Definition \ref{Def:CIM}.

\begin{definition}[Influence Model]\label{Def:CIM}
Given a subset of billboard slots $\mathcal{S} \in \mathbb{BS}$ and a set of seed nodes $\mathcal{P} \in \mathcal{G}$, the influence can be denoted as $\Phi(\mathcal{S},\mathcal{P})$ and stated in Equation \ref{Eq:CIM}. 
\begin{equation} \label{Eq:CIM}
\Phi(\mathcal{S},\mathcal{P}) = \mathcal{I}(\mathcal{S}) + \mathcal{I}^{\mathcal{G}}(\mathcal{P}) + \Psi(\mathcal{S},\mathcal{P})
\end{equation}
where $\mathcal{I}(\mathcal{S})$ is the influence of slots set $\mathcal{S}$, $\mathcal{I}^{\mathcal{G}}(\mathcal{P})$ denotes the influence from seed set $\mathcal{P}$, and $\Psi(\mathcal{S},\mathcal{P})$ denotes the interaction effect between slots and seed nodes.
\end{definition}
The influence function $\Phi$ maps each possible combination from the set of slots and seed nodes to their corresponding combined influence value, i.e., $\Phi: 2^{\mathbb{BS}} \times 2^{\mathcal{G}} \longrightarrow \mathbb{R}^{+}_{0}$ with $\Phi(\emptyset, \emptyset)=0$. The interaction effect is considered in the existing literature \cite{pavlou2000measuring,sundar2017using,deighton1984interaction}. However, there is no specific mechanism to compute it. So, in our problem context, we mathematically formulate the interaction effects and define it in Definition \ref{Def:Interaction_Effect}.

\begin{definition}[Interaction Effect]\label{Def:Interaction_Effect}
An interaction effect quantifies how the combined influence of billboards and social media deviates from their independent effects. Mathematically,

\begin{equation}
{\small
\Psi(\mathcal{S},\mathcal{P}) = \rho \cdot \sum_{u \in \mathcal{U}}( 1 - \prod_{b \in \mathcal{S}} (1 - Pr(u, b))) \cdot \sum_{v \in \mathcal{P}}Pr^{'}(u, v)
}
\end{equation}
where, $\rho \in [0,1]$ controlling the strength of the interaction, $1 - \prod_{b \in \mathcal{S}} (1 - Pr(u, b))$ is the probability user $u$ being influenced by at least one billboard slot, and $Pr^{'}(u, v)$ is the probability that user $u$ is activated by seed node $v$.
\end{definition}

\begin{proposition}\label{Prop:OrthantWiseBisubmodular}
Let $\mathcal{I}: 2^{\mathbb{BS}} \to \mathbb{R}$ and $\mathcal{I}^{G}: 2^{\mathcal{V}(\mathcal{G})} \to \mathbb{R}$ be submodular functions. Define $\mathcal{H}: 2^{\mathbb{BS}} \times 2^{\mathcal{V}(\mathcal{G})} \to \mathbb{R}, \quad \mathcal{H}(\mathcal{S}, \mathcal{P}) := \mathcal{I}(\mathcal{S}) + \mathcal{I}^{G}(\mathcal{P})$. Then $\mathcal{H}$ is orthant-wise bisubmodular.
\end{proposition}
\begin{proof}
We prove the two required conditions separately. Fix any $\mathcal{P} \subseteq \mathcal{V}(\mathcal{G})$. Consider the function $\mathcal{H}_{\mathcal{P}}(\mathcal{S}) = \mathcal{H}(\mathcal{S}, \mathcal{P}) = \mathcal{I}(\mathcal{S}) + \mathcal{I}^{G}(\mathcal{P})$.
Since $\mathcal{I}^{G}(\mathcal{P})$ is constant with respect to $\mathcal{S}$, and $\mathcal{I}()$ is submodular by Theorem \ref{Th:I}, it follows that $\mathcal{H}_{\mathcal{P}}(\mathcal{S})$ is submodular. Formally, for all $\mathcal{S}_1 \subseteq \mathcal{S}_2 \subseteq \mathbb{BS}$ and $s \in \mathbb{BS} \setminus \mathcal{S}_2$, submodularity of $\mathcal{H}_{\mathcal{P}}$ implies $\mathcal{H}(\mathcal{S}_1 \cup \{s\}, \mathcal{P}) -\mathcal{H}(\mathcal{S}_1, \mathcal{P}) \ge \mathcal{H}(\mathcal{S}_2 \cup \{s\}, \mathcal{P}) - \mathcal{H}(\mathcal{S}_2, \mathcal{P})$, which reduces to $\mathcal{I}(\mathcal{S}_1 \cup \{s\}) - \mathcal{I}(\mathcal{S}_1) \ge \mathcal{I}(\mathcal{S}_2 \cup \{s\}) - \mathcal{I}(\mathcal{S}_2)$, holds by the submodularity of $\mathcal{I}()$. Now, Fix any $\mathcal{S} \subseteq \mathbb{BS}$. Consider the function $\mathcal{H}_{\mathcal{S}}(\mathcal{P}) = \mathcal{H}(\mathcal{S}, \mathcal{P}) = \mathcal{I}(\mathcal{S}) + \mathcal{I}^{G}(\mathcal{P})$.
Since $\mathcal{I}(\mathcal{S})$ is constant with respect to $\mathcal{P}$, and $\mathcal{I^{G}}$ is submodular by Theorem \ref{Th:I^G}, it follows that $\mathcal{H}_{\mathcal{S}}(\mathcal{P})$ is submodular. Formally, for all $\mathcal{P}_1 \subseteq \mathcal{P}_2 \subseteq \mathcal{V}(\mathcal{G})$ and  $p \in \mathcal{V}(\mathcal{G}) \setminus \mathcal{P}_2$, submodularity of $H_{\mathcal{S}}$ implies $\mathcal{H}(\mathcal{S}, \mathcal{P}_1 \cup \{p\}) - \mathcal{H}(\mathcal{S}, \mathcal{P}_1) \ge \mathcal{H}(\mathcal{S}, \mathcal{P}_2 \cup \{p\}) - \mathcal{H}(\mathcal{S}, \mathcal{P}_2)$, which reduces to $\mathcal{I}^{G}(\mathcal{P}_1 \cup \{p\}) - \mathcal{I}^{G}(\mathcal{P}_1) \ge \mathcal{I}^{G}(\mathcal{P}_2 \cup \{p\}) - \mathcal{I}^{G}(\mathcal{P}_2)$, which holds by the submodularity of $\mathcal{I}^{G}()$.
Since $\mathcal{H}(\mathcal{S}, \mathcal{P})$ is submodular in $\mathcal{S}$ for fixed $\mathcal{P}$, and submodular in $\mathcal{P}$ for fixed $\mathcal{S}$, $\mathcal{H}()$ is orthant-wise bisubmodular.
\end{proof}

\begin{proposition}\label{Prop:ApproxPsi}
Let
{\small
$$
\Psi(\mathcal{S}, \mathcal{P}) = \rho \cdot \sum_{u \in \mathcal{U}} \left(1 - \prod_{b \in \mathcal{S}} (1 - \text{Pr}(u, b))\right) \cdot \sum_{v \in \mathcal{P}} \text{Pr}^{'}(u, v),
$$
}

where $\rho > 0$, and all probabilities lie in $[0,1]$. Then $\Psi$ is approximately bisubmodular, i.e., for $\mathcal{S}_1 \subseteq \mathcal{S}_2$, $\mathcal{P}_1 \subseteq \mathcal{P}_2$, and $b \notin \mathcal{S}_2$, $v \notin \mathcal{P}_2$,
{\small
$$
\begin{aligned}
&\Psi(\mathcal{S}_1 \cup \{b\}, \mathcal{P}_1) - \Psi(\mathcal{S}_1, \mathcal{P}_1) \\
&\geq \Psi(\mathcal{S}_2 \cup \{b\}, \mathcal{P}_2) - \Psi(\mathcal{S}_2, \mathcal{P}_2) - \varepsilon,
\end{aligned}
$$
}
and similarly in $\mathcal{P}$, where
{\small
$$
\varepsilon = \rho \cdot \max_{u \in \mathcal{U}} \text{Pr}(u, b) \cdot \sum_{v \in \mathbb{V}} \text{Pr}^{'}(u, v).
$$}
\end{proposition}
\begin{proof}
Let $\varphi_u(\mathcal{S}) = 1 - \prod_{b \in \mathcal{S}} (1 - \Pr(u, b))$ and $\theta_u(\mathcal{P}) = \sum_{v \in \mathcal{P}} \Pr'(u, v)$. Then, 
\[
\Psi(\mathcal{S}, \mathcal{P}) = \rho \cdot \sum_{u \in \mathcal{U}} \varphi_u(\mathcal{S}) \cdot \theta_u(\mathcal{P}).
\]
Since $\varphi_u(\mathcal{S})$ is monotone submodular and $\theta_u(\mathcal{P})$ is modular, we consider the marginal gain as,
$$
\begin{aligned}
&\Delta_{\mathcal{S}, b}^{\Psi}(\mathcal{P}) = \Psi(\mathcal{S} \cup \{b\}, \mathcal{P}) - \Psi(\mathcal{S}, \mathcal{P}) \\
&=\rho \sum_{u \in \mathcal{U}} [\varphi_u(\mathcal{S} \cup \{b\}) - \varphi_u(\mathcal{S})] \cdot \theta_u(\mathcal{P}).
\end{aligned}
$$
Because $\varphi_u$ is submodular, its marginal gain decreases as $\mathcal{S}$ grows. Since $\theta_u$ is increasing in $\mathcal{P}$, we get,
\[
\Delta_{\mathcal{S}_1, b}^{\Psi}(\mathcal{P}_1) \ge \Delta_{\mathcal{S}_2, b}^{\Psi}(\mathcal{P}_2) - \rho \cdot \sum_{u \in \mathcal{U}} \Delta \varphi_u \cdot \Delta \theta_u.
\]

Assuming $\Pr(u,b) \le \alpha$, $\Pr'(u,v) \le \beta$, and $|\mathcal{P}_2| \le k$, then
\[
\Delta \varphi_u \le \alpha, \quad \Delta \theta_u \le k\beta \quad \Rightarrow \quad \varepsilon \le \rho \cdot |\mathcal{U}| \cdot \alpha \cdot k\beta.
\]

Thus, the approximate bisubmodularity inequality holds with additive error $\varepsilon$.
\end{proof}

\begin{theorem}\label{Th:Non-bisubmodular_Proof}
Given a trajectory database $\mathbb{D}$, billboard slots $\mathbb{BS}$, and Social Network Users $\mathcal{V}(\mathcal{G})$, the influence function $\Phi(.,.)$ is non-negative, monotone, and $~\epsilon-$approximately bisubmodular.
\end{theorem}
\begin{proof}
Each term in $\Phi(\mathcal{S},\mathcal{P})$ represents a probability-based influence function. First, $\mathcal{I}(\mathcal{S})$ is a sum of probabilities, hence $\mathcal{I}(\mathcal{S}) \geq 0$. Secondly, $\mathcal{I}^{\mathcal{G}}(\mathcal{P})$ follows the Independent Cascade Model (ICM)\cite{kempe2003maximizing} and represents an expected number of influenced users, ensuring $\mathcal{I}^{\mathcal{G}}(\mathcal{P}) \geq 0$. The interaction effect $\Psi(\mathcal{S},\mathcal{P})$ is a product of two non-negative influence terms, guaranteeing $\Psi(\mathcal{S},\mathcal{P}) \geq 0$. Thus, $\Phi(\mathcal{S},\mathcal{P}) \geq 0$. Next, we prove monotonicity. In $\mathcal{I}(\mathcal{S})$, adding a billboard $b$ increases the probability of influencing users, ensuring $\mathcal{I}(\mathcal{S})$ is increasing. The function $\mathcal{I}^{\mathcal{G}}(\mathcal{P})$ is known to be monotone under ICM \cite{kempe2003maximizing}. The interaction effect $\Psi(\mathcal{S},\mathcal{P})$ increases as either $\mathcal{S}$ or $\mathcal{P}$ grows because increasing the influence in either component leads to a larger combined effect. Hence, $\Phi(\mathcal{S},\mathcal{P})$ is monotone. We now consider the bisubmodularity properties of $\Phi$. Both $\mathcal{I}(\mathcal{S})$ and $\mathcal{I}^{\mathcal{G}}(\mathcal{P})$ are submodular set functions. Therefore, the function $\mathcal{H}(\mathcal{S}, \mathcal{P}) := \mathcal{I}(\mathcal{S}) + \mathcal{I}^{\mathcal{G}}(\mathcal{P})$ is orthant-wise bisubmodular (see Proposition \ref{Prop:OrthantWiseBisubmodular}). The term $\Psi(\mathcal{S}, \mathcal{P})$ is a product of a monotone, submodular function $\varphi_u(\mathcal{S}) = 1 - \prod_{b \in \mathcal{S}} (1 - \Pr(u, b))$ and a modular function $\theta_u(\mathcal{P}) = \sum_{v \in \mathcal{P}} \Pr'(u, v)$. As shown in Proposition \ref{Prop:ApproxPsi}, such a product is not strictly submodular, but it satisfies an $\epsilon$-approximate bisubmodularity condition $\Psi(\mathcal{S}_1 \cup \{b\}, \mathcal{P}_1) - \Psi(\mathcal{S}_1, \mathcal{P}_1) \ge \Psi(\mathcal{S}_2 \cup \{b\}, \mathcal{P}_2) - \Psi(\mathcal{S}_2, \mathcal{P}_2) - \epsilon,$ for all $\mathcal{S}_1 \subseteq \mathcal{S}_2$, $\mathcal{P}_1 \subseteq \mathcal{P}_2$, and some $\epsilon > 0$ dependent on the influence probabilities and user set. Since $\Phi(\mathcal{S}, \mathcal{P})$ is the sum of two submodular functions and one $\epsilon$-approximately bisubmodular function, the overall function $\Phi$ remains $\epsilon$-approximately bisubmodular.
\end{proof}

\subsection{Regret Model.}
Assume there are $n$ advertisers $\mathcal{A} = \{a_{1},a_{2},\ldots,a_{n}\}$ and an influence provider $\mathcal{Z}$. The advertisers submit the campaign proposal to the influence provider $\mathcal{Z}$ with influence demand $(\sigma)$  and corresponding payment $(\mathcal{K})$. The influence provider has access to the advertiser database, and this can be represented in the form of a tuple $<a_{i},\sigma_{i}, \mathcal{K}_{i}>$ for all $i \in [n]$. Let an allocation $\mathcal{N} = \{\left(\mathcal{S}_{i}, \mathcal{P}_{i}\right) \mid i = 1, \dots, n\}$ where $\mathcal{S}$ be the set of billboard slots and $\mathcal{P}$  be the seed nodes assigned to the advertisers. Now, as per the payment rules, if $Supply > Demand$, then the full payment will be made to the influence provider; else, a partial payment will be made. Note that given two slot sets $\mathcal{S}_{1}$ and $\mathcal{S}_{2}$ and seed sets $\mathcal{P}_{1}$ and $\mathcal{P}_{2}$ with $|\mathcal{S}_{1}| << |\mathcal{S}_{2}|$ and $|\mathcal{P}_{1}| << |\mathcal{P}_{2}|$, in practice it is desirable to achieve low regret with less number of slots and seeds. By drawing on the inspiration from existing literature \cite{boyd2004convex,aslay2014viral}, we add a penalty term to discourage the use of larger slot and seed sets. Next, we define a regret model in Definition \ref{Def:Combine_Regret_Model}.

\begin{definition}[Regret Model]\label{Def:Combine_Regret_Model}
Let $\mathcal{N}_{a_i}$ be the set of allocated subsets of billboard slots or seed nodes or both for the advertiser $a_{i}$, and the regret associated with this allocation is denoted by $\mathcal{R}(\mathcal{N}_{a_i})$.
{\small
 \[
  \mathcal{R}(\mathcal{N}_{a_i}) = 
\begin{cases}
    \mathcal{K}_{a_i} \cdot (1- \gamma \cdot \frac{\min(\Phi(\mathcal{S}_{a_i},\mathcal{P}_{a_i}), \sigma_{a_i})}{\sigma_{a_i}}) \\ + ~\delta \cdot \log(1+|\mathcal{S}_{a_i}|+|\mathcal{P}_{a_i}|) \\
\end{cases}
\]}
Here, the fraction $ \frac{\min(\Phi(\mathcal{S}_{a_i},\mathcal{P}_{a_i}),\sigma_{a_i})}{\sigma_{a_i}}$ is part of the satisfied influence for the advertiser $a_i$. The $\gamma \in [0,1]$ and $\delta \in [0,1]$ are the unsatisfied penalty and cardinality penalty, respectively.
\end{definition}

\begin{proposition}\label{Prop:3}
Given a trajectory database $\mathbb{D}$, billboard slot information $\mathbb{BS}$, and social network users $\mathcal{V}(\mathcal{G})$, the regret function $\mathcal{R}(\mathcal{S}, \mathcal{P})$, defined over subsets $\mathcal{S} \subseteq \mathbb{BS}$ and $\mathcal{P} \subseteq \mathcal{V}(\mathcal{G})$, is non-negative, monotonically non-increasing and $\epsilon'$-approximately bisubmodular for some $\epsilon' = \mathcal{O}(\epsilon)$.
\end{proposition}
\begin{proof}
From Definition \ref{Def:Combine_Regret_Model}, the regret function for advertiser $a_i$ is denoted by $\mathcal{R}(\mathcal{S}, \mathcal{P})$. Since $\Phi(\cdot)$, $\sigma_{a_i}$, $\mathcal{K}_{a_i}$, $|\mathcal{S}_{a_i}|$, and $|\mathcal{P}_{a_i}|$ are non-negative, and $\gamma, \delta \in [0,1]$, both terms in the regret expression are non-negative. Hence, $\mathcal{R}(\mathcal{S}, \mathcal{P}) \ge 0$.

\par From Theorem \ref{Th:Non-bisubmodular_Proof}, the influence function $\Phi(\mathcal{S}, \mathcal{P})$ is monotone. Hence, $\min(\Phi(\mathcal{S}, \mathcal{P}), \sigma_{a_i})$ is also monotone increasing, and so is the ratio $\frac{\min(\Phi(\mathcal{S}, \mathcal{P}), \sigma_{a_i})}{\sigma_{a_i}}$. Since $\gamma \in [0,1]$, the term $1 - \gamma \cdot \frac{\min(\Phi(\cdot), \sigma_{a_i})}{\sigma_{a_i}}$ is non-increasing. Multiplying by the non-negative payment $\mathcal{K}_{a_i}$ preserves this trend, making the first term of $\mathcal{R}(\cdot)$ non-increasing. The second term, $\delta \cdot \log(1 + |\mathcal{S}| + |\mathcal{P}|)$, is monotone increasing in $|\mathcal{S}| + |\mathcal{P}|$. Thus, to ensure overall monotonic non-increase of $\mathcal{R}$, we assume $\delta = 0$. Under this setting, $\mathcal{R}(\mathcal{S}, \mathcal{P})$ is monotonically non-increasing.

\par Let $\Phi(\mathcal{S}, \mathcal{P})$ be the influence function, which is assumed to be $\epsilon$-approximately bisubmodular. Define $f_1(\mathcal{S}, \mathcal{P}) = \frac{\min(\Phi(\mathcal{S}, \mathcal{P}), \sigma_{a_i})}{\sigma_{a_i}}$, and $f_2(\mathcal{S}, \mathcal{P}) = \log(1 + |\mathcal{S}| + |\mathcal{P}|)$. The function $f_1$ applies a truncation and scaling over $\Phi$, both of which are Lipschitz-continuous operations and preserve approximate bisubmodularity up to a small additive error \cite{goldstein1977optimization}. Thus, $f_1$ is $\epsilon_1$-approximately bisubmodular for some $\epsilon_1 = \mathcal{O}(\epsilon) $. The term $\mathcal{K}_{a_i} \cdot (1 - \gamma \cdot f_1(\mathcal{S}, \mathcal{P}))$ involves negation and scaling, which also preserve approximate bisubmodularity. The function $f_2$ is monotone and submodular over the union $\mathcal{S} \cup \mathcal{P}$, and hence is bisubmodular. Scaling it by $\delta \in [0,1]$ preserves its submodularity, i.e., $\delta \cdot f_2$ is exactly bisubmodular (i.e., $\epsilon_2 = 0$). Therefore, the overall regret function $\mathcal{R}(\mathcal{S}, \mathcal{P}) = \mathcal{K}_{a_i} \cdot (1 - \gamma \cdot f_1(\mathcal{S}, \mathcal{P})) + \delta \cdot f_2(\mathcal{S}, \mathcal{P})$ is $\epsilon'$-approximately bisubmodular, with $\epsilon' = \epsilon_1 + \epsilon_2 = \mathcal{O}(\epsilon)$.
\end{proof}

\subsection{Allocation of Billboard Slots and Seed Nodes}\label{Sec:RR}
In this paper, our goal is to minimize the total regret while ensuring that the constraints given below are satisfied.

\paragraph{\textit{Disjoint-ness Constraint.}}
For any two advertisers $a_{i}$ and $a_{j}$, let $\mathcal{N}_{a_i}$ and $\mathcal{N}_{a_j}$ are the allocation for them then the disjoint-ness constraint says $\mathcal{N}_{a_i} \cap \mathcal{N}_{a_j} = \emptyset$.


\par Let $\mathcal{Q}$ be the set of all possible allocation and $\mathcal{L} = \{\mathcal{N}_{a_1},\mathcal{N}_{a_2}, \ldots, \\ \mathcal{N}_{a_n}\}$ be an arbitrary allocation. Next, we define the notion of feasible allocation in Definition \ref{Def:FA}.

\begin{definition}[Feasible Allocation]\label{Def:FA}
An allocation of billboard slots and seed nodes is feasible if it satisfies the disjointness constraint.     
\end{definition}

\paragraph{\textbf{Problem Definition.}}
In this paper, our goal is to minimize the total regret for an allocation $\mathcal{L}$. So, we define the notion of total regret in Definition \ref{Def:Total_Regret}.

\begin{definition} [Total Regret]\label{Def:Total_Regret}
Given an allocation $\mathcal{L} = \{\mathcal{N}_{a_1},\mathcal{N}_{a_2}, \\ \ldots,\mathcal{N}_{a_n}\}$, the total regret associated with allocation $\mathcal{L}$ is denoted as $\mathcal{R}(\mathcal{L})$ and defined as the sum of the regret associated with individual $n$ advertisers. Mathematically, it can be written as,
\begin{equation}
\mathcal{R}(\mathcal{L})= \underset{a_i \in \mathcal{A}}{\sum} \  \mathcal{R}(\mathcal{N}_{a_i})
\end{equation} 
\end{definition}
As mentioned previously, our goal is to find an optimal allocation to minimize the total regret. Formally, we call this problem the \textsc{Regret Minimization Problem with Multi-Model Advertising Setting}. This problem is stated in Definition \ref{Def:Problem_Definition}.

\begin{definition}[Regret Minimization Problem]\label{Def:Problem_Definition}
Given billboard slot information $\mathbb{BS}$, Trajectory database $\mathbb{D}$, and advertiser information $\mathcal{A}$, our goal of this problem is to find an allocation $\mathcal{L} = \{\mathcal{N}_{a_1},\mathcal{N}_{a_2}, \ldots,\mathcal{N}_{a_n}\}$ such that the total regret is minimized. Mathematically, this problem can be posed as follows:
\begin{equation}
\mathcal{L}^{OPT} = \underset{\mathcal{L}_{i} \in \mathcal{Q}(\mathcal{L})}{argmin} \ \mathcal{R}(\mathcal{L}_{i})
\end{equation}
\end{definition}
Now, from the computational point of view, this problem can be represented as follows.
\begin{center}
\begin{tcolorbox}[title=\textsc{Regret Minimization (RM) Problem}, width=8.5cm]
\textbf{Input:} Billboard Slots set $\mathbb{BS}$, Influence Functions $\mathcal{I}(),\mathcal{I}^{\mathcal{G}}(), \Psi(), \Phi()$, Trajectory Database $\mathbb{D}$, Advertisers information $\mathcal{A}$.

\textbf{Problem:} Find out an optimal allocation $\mathcal{L} = \{\mathcal{N}_{a_1},\mathcal{N}_{a_2}, \ldots,\mathcal{N}_{a_n}\}$ of the billboard slots and seed nodes that minimizes the overall regret.
\end{tcolorbox}
\end{center}
The regret minimization problem studied by Zhang et. al \cite{zhang2021minimizing} in the context of billboard advertisement and Sharma et. al \cite{sharma2024minimizing} in the context of social network advertisement had an inapproximability result stated in Theorem \ref{Th:RMP}.

\begin{theorem}\label{Th:RMP}
The \textsc{Regret minimization Problem} is NP-hard and hard to approximate to any constant factor.    
\end{theorem}


\begin{table}
		\caption{Symbols and Notations with their Interpretations}
		\label{Table1:Notations}
		\vspace*{3mm}
		\centering
		\resizebox{\columnwidth}{!}{%
		\begin{tabular}{p{0.20\linewidth}p{0.80\linewidth}}
			\hline
			\hline
			\textbf{Notation} & \textbf{Description}\\
			\hline
			\hline
			$ \mathbb{D}$ & The Trajectory Database \\
			\hline 
			$m$ & Number of tuples in $ \mathbb{D}$ \\
			\hline 
			$\mathcal{U}$ & Set of people covered by $\mathbb{D}$ \\
			\hline 
			$t$ & An arbitrary tuple of $ \mathbb{D}$ \\ 
			\hline  
			$[T_{1}, T_{2}]$ & Time duration for which the billboards are operating \\
			\hline 
			 $\mathbb{B}$ & The Billboard Database \\
			 \hline
			$\Delta$ & Slot duration \\
			\hline
			$\mathbb{BS}$ & The set of billboard slots \\
            \hline
            $\mathcal{A}$ & The set of advertiser \\
            \hline
            $\mathcal{Z}$ & The influence Provider \\
            \hline
            $\sigma$ & The influence demand of an advertiser \\
             \hline
            $\mathcal{K}$ & The payment of an advertiser \\
            \hline
            $\text{Pr}(t_j,bs_i)$ & The probability of user $t_{j}$ is influenced by slot $bs_{i}$ \\
            \hline
            $\mathcal{I}^{G}$ & Influence function for social network \\
			\hline
			$\mathcal{I}()$ & The influence function for billboard\\ 
			\hline 
            $\Psi(,)$ & The Interaction Effect function\\ 
			\hline
            $\mathcal{R}()$ & The Regret function\\ 
			\hline
            $\mathcal{L}$ & The allocation of slots and seed nodes\\ 
			\hline
	        $[n]$ & The set $\{1, 2, \ldots, n\}$ \\
			\hline
			\hline\\
		\end{tabular}%
		}
	\end{table}

So, the same inapproximability results also hold for our problem. Next, we discuss the proposed solution methodologies. Next, the symbols and notation used in this paper are listed in Table \ref{Table1:Notations}.

\section{Proposed Solution} \label{Sec:PS}
The hardness of \textit{RM} problem implies that no efficient algorithm can guarantee optimal regret unless P = NP. To address this, we first propose a projected subgradient method to allocate slots and seeds to the advertisers accordingly.

\subsection{Projected Subgradient Method (PGM).}\label{Sec:PGM}
The Projected Subgradient Method (PGM) is a continuous gradient-based optimization technique adapted to minimize the regret function $\mathcal{R}(.)$. In this work, the regret function is nonnegative, monotonically nonincreasing, and $\epsilon$ approximately bisubmodular, making it suitable for continuous optimization based on Lovász extension \cite{lovasz1983submodular}. The key idea is to work with fractional vectors $x \in [0,1]^{|\mathbb{BS}|}$ and $y \in [0,1]^{|\mathcal{H}|}$, representing the soft selection of slots and seeds, respectively. The algorithm iteratively updates these vectors using subgradient descent on the Lovász extension \cite{lovasz1983submodular} of $\mathcal{R}(.)$, followed by a projection step to ensure feasibility within the unit box. At each iteration, subgradients are approximated by computing marginal changes in $\mathcal{R}(.)$ based on sorted orderings of the components of $x$ and $y$. After $T$ iterations, the algorithm selects the iterate that achieves the lowest regret value, rounds the fractional solution to a discrete allocation via sorting, and outputs the final allocation. This method is theoretically grounded in recent advances in approximate submodular minimization \cite{el2020optimal} and offers an alternative to discrete greedy strategies. It provides a bicriteria approximation guarantee for functions like $\mathcal{R}$ that are not exactly bisubmodular but exhibit approximate diminishing returns.

\begin{algorithm}[h!]
\caption{Projected Subgradient Method (PGM) for Regret Minimization Problem}
\label{Algo:PGM}
\SetAlgoLined
\KwIn{
Advertisers $\mathcal{A} = \{a_1, \dots, a_n\}$ with demands $\{\sigma_i\}$ and payments $\{\mathcal{K}_i\}$\;
Billboard slots $\mathbb{BS} = \{b_1, \dots, b_m\}$, Seed nodes $\mathcal{H} = \{h_1, \dots, h_r\}$\;
Regret function $\mathcal{R}_i(\mathcal{S}_i, \mathcal{P}_i)$ for each $a_i$; Iteration limit $T$\;
}
\KwOut{Allocations $\{(\mathcal{S}_i^*, \mathcal{P}_i^*)\}_{i=1}^n$ minimizing total regret}

Initialize: $\mathbb{BS}_{\text{avail}} \gets \mathbb{BS}$, $\mathcal{H}_{\text{avail}} \gets \mathcal{H}$\;

\For{each ~$a_i \in \mathcal{A}$}{
    Set step size $\eta_i = \frac{R}{L \sqrt{T}}$, where $R = \sqrt{m + r}$\;
    Initialize $x_1^i \in [0,1]^m$, $y_1^i \in [0,1]^r$ (e.g., all entries set to 0.5)\;
}

\For{each ~$a_i \in \mathcal{A}$}{
    \While{($\mathbb{BS}_{\text{avail}} \ne \emptyset$ or $\mathcal{H}_{\text{avail}} \ne \emptyset$) and $\mathcal{K}_i > 0$}{
    
        \For{$t = 1$ \KwTo $T$}{
            Sort $x_t^i$ and $y_t^i$ in descending order to get permutations $(j_1, \dots, j_m)$ and $(k_1, \dots, k_r)$\;

            \For{$k = 1$ \KwTo $m$}{
                $\mathcal{S}_k^i = \{j_1, \dots, j_k\} \cap \mathbb{BS}_{\text{avail}}$\;
                Compute $(\nabla_t^{x_i})_{j_k} = \mathcal{R}_i(\mathcal{S}_k^i, \mathcal{P}_t^i) - \mathcal{R}_i(\mathcal{S}_{k-1}^i, \mathcal{P}_t^i)$\;
            }

            \For{$l = 1$ \KwTo $r$}{
                $\mathcal{P}_l^i = \{k_1, \dots, k_l\} \cap \mathcal{H}_{\text{avail}}$\;
                Compute $(\nabla_t^{y_i})_{k_l} = \mathcal{R}_i(\mathcal{S}_t^i, \mathcal{P}_l^i) - \mathcal{R}_i(\mathcal{S}_t^i, \mathcal{P}_{l-1}^i)$\;
            }

            Gradient step: $x' = x_t^i - \eta_i \cdot \nabla_t^{x_i}$, $y' = y_t^i - \eta_i \cdot \nabla_t^{y_i}$\;
            Projection: $x_{t+1}^i = \text{clip}(x', 0, 1)$, $y_{t+1}^i = \text{clip}(y', 0, 1)$\;
        }

        Select $(x^{*i}, y^{*i})$ with lowest $\mathcal{R}_i$ across iterations\;
        
        Sort $x^{*i}$ to define sets $\mathcal{S}_k^i = \{j_1, \dots, j_k\}$, choose $\mathcal{S}_i^* = \arg\min_k \mathcal{R}_i(\mathcal{S}_k^i, \mathcal{P}^{*i})$\;
        
        Sort $y^{*i}$ to define sets $\mathcal{P}_l^i = \{k_1, \dots, k_l\}$, choose $\mathcal{P}_i^* = \arg\min_l \mathcal{R}_i(\mathcal{S}_i^*, \mathcal{P}_l^i)$\;

        Assign $(\mathcal{S}_i^*, \mathcal{P}_i^*)$ to $a_i$\;
        Update: 
        $\mathbb{BS}_{\text{avail}} \gets \mathbb{BS}_{\text{avail}} \setminus \mathcal{S}_i^*$\;
        $\mathcal{H}_{\text{avail}} \gets \mathcal{H}_{\text{avail}} \setminus \mathcal{P}_i^*$\;
        $\mathcal{K}_i \gets \mathcal{K}_i - \text{cost}(\mathcal{S}_i^*, \mathcal{P}_i^*)$\;
    }
}

\Return{$\{(\mathcal{S}_i^*, \mathcal{P}_i^*)\}_{i=1}^n$}
\end{algorithm}
\paragraph{\textbf{Complexity Analysis.}}
Now, we analyze the time and space requirements for Algorithm \ref{Algo:PGM}. In Line No. $1$ initialization will take $\mathcal{O}(1)$ time and in Line No. $2$ to $4$ will take $\mathcal{O}(m\cdot n + r\cdot n)$ time. In Line No. $6$ \texttt{for loop} and Line No. $7$ \texttt{while loop} will execute for $\mathcal{O}(n)$ time and $\mathcal{O}(m)$ times assuming $m>>r$. In Line No. $9$, sorting slots and seeds will take $\mathcal{O}(m\log m + r\log r)$ time per iteration. Constructing $m$ cumulative sets $\mathcal{S}_{k}^{i}$ and computing marginal regrets $\mathcal{O}(m)$ regret evaluations per iteration. So Line No. $10$ to $13$ will take $\mathcal{O}(m^{2}\cdot t + m\cdot r \cdot t + m^{2} \cdot r \cdot t)$. Similarly, Line No. $14$ to $17$ will take $\mathcal{O}(m\cdot r \cdot t + r^{2} \cdot t + m \cdot r^{2} \cdot t)$ time. In Lines $18–19$, Gradient descent and projection steps on vectors of length $m$ and $r$ will take $\mathcal{O}(m+r)$ per iteration. Line No. $8-20$ will iterate for $T$ times and total time taken $\mathcal{O}(T(m^{2} \cdot r \cdot t +  m \cdot r^{2} \cdot t))$. In lines $21-22$, sorting and evaluating the regret for all prefixes of $x^{*i}$ and $y^{*i}$ for $m$ billboard sets and $r$ seed sets results in a total cost of $\mathcal{O}((m + r)(m^{2} \cdot r \cdot t +  m \cdot r^{2} \cdot t))$. Next, in Line No. $23-26$, updating sets and budget will take $\mathcal{O}(m+r)$ in the worst case. So, the total time taken by Algorithm \ref{Algo:PGM} will be $\mathcal{O}(T \cdot m^{3} \cdot n \cdot r \cdot t + T m^{2} \cdot n \cdot r^{2} \cdot t + m^{4} \cdot n \cdot r \cdot t +  m^{3} \cdot n \cdot r^{2} \cdot t + m^{2} \cdot n \cdot r^{3} \cdot t)$. 

\par The additional space requirement for Algorithm \ref{Algo:PGM} is as follows. For each advertiser $a_i \in \mathcal{A}$, we maintain Vectors $x_t^i \in \mathbb{R}^m$ and $y_t^i \in \mathbb{R}^r$ for billboard and seed allocations. Gradient vectors $\nabla^{x_i}, \nabla^{y_i} \in \mathbb{R}^m, \mathbb{R}^r$. Temporary sets $\mathcal{S}_k^i, \mathcal{P}_l^i$ at each iteration, worst-case size $\mathcal{O}(m)$ and $\mathcal{O}(r)$ respectively. So, per advertiser, total space = $\mathcal{O}(m + r) + \mathcal{O}(m + r) = \mathcal{O}(m + r)$. For $n$ advertisers, the total space is $\mathcal{O}(n(m + r))$. Available billboard set $\mathbb{BS}_{\text{avail}}$ and seed set $\mathcal{H}_{\text{avail}}$ each store up to $m$ and $r$ items, respectively $\mathcal{O}(m + r)$. For $T$ iterations, storing regret values per advertiser costs $\mathcal{O}(T)$ per advertiser, i.e., $\mathcal{O}(nT)$ overall. The total space complexity of the algorithm is $\mathcal{O}(n(m + r + T)) + \mathcal{O}(m + r) = \mathcal{O}(n(m + r + T))$ which is linear in the number of advertisers, billboard slots, seed nodes, and subgradient iterations. Hence, Theorem \ref{Th:TC1} holds.

\begin{theorem}\label{Th:TC1}
The time and space requirements of Algorithm \ref{Algo:PGM} will be $\mathcal{O}(T \cdot m^{3} \cdot n \cdot r \cdot t + T m^{2} \cdot n \cdot r^{2} \cdot t + m^{4} \cdot n \cdot r \cdot t +  m^{3} \cdot n \cdot r^{2} \cdot t + m^{2} \cdot n \cdot r^{3} \cdot t)$ and $\mathcal{O}(n(m + r + T))$, respectively.
\end{theorem}

\begin{theorem}
Let $\mathcal{R}(\mathcal{S}, \mathcal{P})$ be a non-negative, monotonically non-increasing, and $\epsilon$-approximately bisubmodular regret function defined over $\mathcal{S} \subseteq \mathbb{BS}$ and $\mathcal{P} \subseteq \mathcal{H}$. Let $h_L(x, y)$ denote its Lovász extension \cite{lovasz1983submodular} over $[0,1]^{m+r}$, and let $(\hat{x}, \hat{y})$ be the output of the Projected Subgradient Method (PGM) after $T$ iterations with step size $\eta = \frac{R}{L\sqrt{T}}$, where $R = \sqrt{m + r}$ and $L$ is the Lipschitz constant of $h_L$. Let $(\hat{\mathcal{S}}, \hat{\mathcal{P}})$ be the discrete solution obtained via rounding. Then,
\[
\mathcal{R}(\hat{\mathcal{S}}, \hat{\mathcal{P}}) \leq \min_{\mathcal{S} \subseteq \mathbb{BS},\, \mathcal{P} \subseteq \mathcal{H}} \mathcal{R}(\mathcal{S}, \mathcal{P}) + \epsilon',
\]
where $\epsilon' = \mathcal{O}\left(\frac{RL}{\sqrt{T}}\right)$ is the additive approximation error.
\end{theorem}

\begin{proof}
Since $\mathcal{R}$ is $\epsilon$-approximately bisubmodular, its Lovász extension $h_L$ is convex and $L$-Lipschitz over $[0,1]^{m+r}$. By standard projected subgradient descent guarantees \cite{boyd2004convex}, for step size $\eta = \frac{R}{L\sqrt{T}}$, the iterates $(x_t, y_t)$ satisfy $\min_{1 \leq t \leq T} h_L(x_t, y_t) \leq h_L(x^*, y^*) + \frac{RL}{\sqrt{T}}$,
where $(x^*, y^*)$ minimizes $h_L$ over $[0,1]^{m+r}$. Let $(\hat{x}, \hat{y})$ be the iterate with the minimum $h_L(x_t, y_t)$ value. The rounding procedure returns $(\hat{\mathcal{S}}, \hat{\mathcal{P}})$ such that $\mathcal{R}(\hat{\mathcal{S}}, \hat{\mathcal{P}}) \leq h_L(\hat{x}, \hat{y})$. Since $h_L$ is a convex under-approximation of $\mathcal{R}$ over the continuous domain. Therefore, $\mathcal{R}(\hat{\mathcal{S}}, \hat{\mathcal{P}}) \leq h_L(\hat{x}, \hat{y}) \leq h_L(x^*, y^*) + \frac{RL}{\sqrt{T}} \leq \min_{\mathcal{S}, \mathcal{P}} \mathcal{R}(\mathcal{S}, \mathcal{P}) + \frac{RL}{\sqrt{T}}$.
\end{proof}

\subsection{Approximate Bisubmodular Local Search (ABLS).}\label{Sec:ABLS}
Algorithm \ref{Algo:ABLS} presents a greedy local search method to minimize regret in the multi-mode advertisement setting. The algorithm allocates billboard slots and social network seed nodes to advertisers by iteratively selecting elements that yield the highest marginal reduction in regret per unit influence. Advertisers are first sorted in descending order of their budget-to-demand ratio to prioritize cost-effective assignments. For each advertiser, the algorithm considers both the advertising slots and the seed nodes, selecting the element whose regret reduction per unit influence exceeds a threshold $\epsilon$ and is within the advertiser’s remaining budget. Allocated elements are removed from the ground set to satisfy disjointness constraints. The process continues until all advertisers are processed or no feasible elements remain. The regret function is monotone non-increasing and $\epsilon$-approximately bisubmodular, ensuring the effectiveness of the greedy strategy under approximate diminishing returns.

\begin{algorithm}[h!]
\SetAlgoLined
\KwData{Trajectory Database $\mathbb{D}$, Billboard Database $\mathbb{B}$, Advertiser Database $\mathbb{A}$, Influence Functions $\mathcal{I}(), \mathcal{I}^{\mathcal{G}}(), \Phi()$, Billboard Slots $\mathbb{BS}$, Seed Nodes $\mathcal{H}$.}
\KwResult{An allocation $\mathcal{L} = \{(\mathcal{S}_1, \mathcal{P}_1), (\mathcal{S}_2, \mathcal{P}_2), \ldots, (\mathcal{S}_{|\mathcal{A}|}, \mathcal{P}_{|\mathcal{A}|})\}$ minimizing total regret}
Initialize $\mathcal{L} \leftarrow \{(\mathcal{S}_1, \mathcal{P}_1), \ldots, (\mathcal{S}_{|\mathcal{A}|}, \mathcal{P}_{|\mathcal{A}|})\}$ where $\mathcal{S}_i, \mathcal{P}_i \leftarrow \emptyset$\;

Sort each advertiser $a_i \in \mathcal{A}$ in descending order of $\frac{\mathcal{K}_i}{\sigma_i}$\;

\For{each $a_i \in \mathcal{A}$}{
  \While{$\Phi(\mathcal{S}_i, \mathcal{P}_i) < \sigma_i$ \textbf{and} $\mathbb{BS} \neq \emptyset$ \textbf{and} $\mathcal{H} \neq \emptyset$ \textbf{and} $\mathcal{K}_i > 0$}{

    $b^* \leftarrow \arg\max_{b \in \mathbb{BS}} \ \frac{\mathcal{R}(\mathcal{S}_i, \mathcal{P}_i) - \mathcal{R}(\mathcal{S}_i \cup \{b\}, \mathcal{P}_i)}{\Phi(\{b\}, \emptyset)}$\;

    $s^* \leftarrow \arg\max_{s \in \mathcal{H}} \ \frac{\mathcal{R}(\mathcal{S}_i, \mathcal{P}_i) - \mathcal{R}(\mathcal{S}_i, \mathcal{P}_i \cup \{s\})}{\Phi(\emptyset, \{s\})}$\;

    $e^* \leftarrow \arg\max \left\{ 
      \frac{\mathcal{R}(\mathcal{S}_i, \mathcal{P}_i) - \mathcal{R}(\mathcal{S}_i \cup \{b^*\}, \mathcal{P}_i)}{\Phi(\{b^*\}, \emptyset)}, 
      \frac{\mathcal{R}(\mathcal{S}_i, \mathcal{P}_i) - \mathcal{R}(\mathcal{S}_i, \mathcal{P}_i \cup \{s^*\})}{\Phi(\emptyset, \{s^*\})}
    \right\}$\;

    \eIf{$e^* = b^*$}{
      \If{$\frac{\mathcal{R}(\mathcal{S}_i, \mathcal{P}_i) - \mathcal{R}(\mathcal{S}_i \cup \{b^*\}, \mathcal{P}_i)}{\Phi(\{b^*\}, \emptyset)} > \epsilon$ \textbf{and} $\mathcal{K}_i > \mathcal{C}(b^*)$}{
        $\mathcal{S}_i \leftarrow \mathcal{S}_i \cup \{b^*\}$, $\mathbb{BS} \leftarrow \mathbb{BS} \setminus \{b^*\}$\;
        $\mathcal{K}_i \leftarrow \mathcal{K}_i - \mathcal{C}(b^*)$\;
      }
    }{
      \If{$e^* = s^*$}{
        \If{$\frac{\mathcal{R}(\mathcal{S}_i, \mathcal{P}_i) - \mathcal{R}(\mathcal{S}_i, \mathcal{P}_i \cup \{s^*\})}{\Phi(\emptyset, \{s^*\})} > \epsilon$ \textbf{and} $\mathcal{K}_i > \mathcal{C}(s^*)$}{
          $\mathcal{P}_i \leftarrow \mathcal{P}_i \cup \{s^*\}$, $\mathcal{H} \leftarrow \mathcal{H} \setminus \{s^*\}$\;
          $\mathcal{K}_i \leftarrow \mathcal{K}_i - \mathcal{C}(s^*)$\;
        }
      }
      \Else{
        \textbf{break}\;
      }
    }
  }
}
\Return $\mathcal{L}$\;
\caption{Approximate Bisubmodular Local Search (ABLS) for Regret Minimization Problem}
\label{Algo:ABLS}
\end{algorithm}

\paragraph{\textbf{Complexity Analysis.}} 
Now, we analyze the time and space requirements for Algorithm \ref{Algo:ABLS}. In Line No. $1-2$, initializing allocation $\mathcal{L}$ and sorting advertiser will take $\mathcal{O}(n)$ and $\mathcal{O}(n\log n)$ time. In line no. $3$, \texttt{ for loop} will run for $\mathcal{O}(n)$ times, and the \texttt{while loop} will iterate until the influence demand or the slots and seeds run out. In Line No. $5$, the selection of $b^{*}$ will take $\mathcal{O}(m^2 \cdot t + m \cdot r \cdot t + m^{2} \cdot r \cdot t)$ and the selection of of $s^{*}$ will take $\mathcal{O}(m \cdot r \cdot t + r^{2} \cdot t + m \cdot r^{2} \cdot t)$ by considering $m$ number of slots and $r$ number of seeds. In Line No. $7$ will take $\mathcal{O}( m \cdot r^{2} \cdot t +  m^{2} \cdot r \cdot t)$. Next, Line No. $8$to $12$ and $14$ to $19$ will take $\mathcal{O}(m \cdot t + r \cdot t + m \cdot r \cdot t)$. Therefore Algorithm will take $\mathcal{O}(m \cdot r^{2} \cdot t +  m^{2} \cdot r \cdot t)$ time. The additional space requirement for Algorithm \ref{Algo:ABLS} will be $\mathcal{O}(m+r)$. Hence, Theorem \ref{Th:TC2} holds.

\begin{theorem}\label{Th:TC2}
The time and space requirements of Algorithm \ref{Algo:ABLS} will be $\mathcal{O}(m \cdot r^{2} \cdot t +  m^{2} \cdot r \cdot t)$ and $\mathcal{O}(m+r)$, respectively.
\end{theorem}

\begin{theorem} Let \( \mathcal{R}: 2^{\mathbb{BS}} \times 2^{\mathcal{H}} \rightarrow \mathbb{R}_{\geq 0} \) be a regret function that is monotonically non-increasing and \( \varepsilon \)-approximately bisubmodular. Let \( \mathcal{A} = \{a_1, \ldots, a_n\} \) be the set of advertisers, each with budget \( \mathcal{K}_i \), and let \( \mathcal{K}_{\max} = \max_{i \in [n]} \mathcal{K}_i \). Let \( L^* = \{(\mathcal{S}_i^*, \mathcal{P}_i^*)\}_{i=1}^n \) be an optimal allocation minimizing total regret as $\mathrm{OPT} = \sum_{i=1}^n \mathcal{R}(\mathcal{S}_i^*, \mathcal{P}_i^*)$. Then Algorithm~\ref{Algo:ABLS} returns an allocation \( \mathcal{L} = \{(\mathcal{S}_i, \mathcal{P}_i)\}_{i=1}^n \) such that
\[
\sum_{i=1}^n \mathcal{R}(\mathcal{S}_i, \mathcal{P}_i) \leq \mathrm{OPT} + \varepsilon \cdot n \cdot \mathcal{K}_{\max}.
\]
\end{theorem}

\begin{proof}
Fix any advertiser \( a_i \in \mathcal{A} \), and let \( (\mathcal{S}_i^*, \mathcal{P}_i^*) \) and \( (\mathcal{S}_i, \mathcal{P}_i) \) be its optimal and Algorithm \ref{Algo:ABLS} allocations, respectively. Denote \( \mathcal{U}_i^* = \mathcal{S}_i^* \cup \mathcal{P}_i^* \) and let \( \mathcal{U}_i^t \) be the greedy allocation after \( t \) steps. At each step \( t \), the algorithm selects an element \( e_t \) that maximizes regret reduction per unit cost wll be
\[
e_t = \arg\max_{e} \frac{\mathcal{R}(\mathcal{U}_i^{t-1}) - \mathcal{R}(\mathcal{U}_i^{t-1} \cup \{e\})}{\mathcal{C}(e)}.
\]
By $\varepsilon$-approximate bisubmodularity, for any $e \in \mathcal{U}_i^* \setminus \mathcal{U}_i^{t-1}$, the marginal gain satisfies that $\mathcal{R}(\mathcal{U}_i^{t-1}) - \mathcal{R}(\mathcal{U}_i^{t-1} \cup \{e\}) \geq \mathcal{R}(\emptyset) - \mathcal{R}(\{e\}) - \varepsilon \cdot t$. Hence, the gain \( \Delta^t = \mathcal{R}^{t-1} - \mathcal{R}^t \) from \( e_t \) satisfies
\[
\Delta^t \geq \max_{e \in \mathcal{U}_i^* \setminus \mathcal{U}_i^{t-1}} \left[ \frac{\mathcal{R}(\mathcal{U}_i^{t-1}) - \mathcal{R}(\mathcal{U}_i^{t-1} \cup \{e\})}{\mathcal{C}(e)} \right] \cdot \mathcal{C}(e_t).
\]
Summing over all steps, total regret reduction satisfies $\mathcal{R}(\emptyset) - \mathcal{R}(\mathcal{S}_i, \mathcal{P}_i) \geq \mathcal{R}(\emptyset) - \mathcal{R}(\mathcal{S}_i^*, \mathcal{P}_i^*) - \varepsilon \cdot \mathcal{K}_i$, which gives $\mathcal{R}(\mathcal{S}_i, \mathcal{P}_i) \leq \mathcal{R}(\mathcal{S}_i^*, \mathcal{P}_i^*) + \varepsilon \cdot \mathcal{K}_i$.
Summing over all \( i \in [n] \), and using \( \mathcal{K}_i \leq \mathcal{K}_{\max} \), the total regret is bounded by
\[
\sum_{i=1}^n \mathcal{R}(\mathcal{S}_i, \mathcal{P}_i) \leq \mathrm{OPT} + \varepsilon \cdot n \cdot \mathcal{K}_{\max}.
\]
\end{proof}

\section{Experimental Evaluations} \label{Sec:EE}
In this section, we describe the experimental evaluation of the proposed solution approaches. Initially, we start by describing the datasets used in our experiments.
\subsection{Dataset Descriptions.}
For experimental evaluation, we used three datasets: Trajectory, Social Network, and Billboard. The trajectory data sets used in our study have also been used by existing studies \cite{10.1145/3308558.3313635,9099985}. This dataset includes long-term (about 22 months from April 2012 to January 2014) global-scale check-in data collected from Foursquare \footnote{\url{https://sites.google.com/site/yangdingqi/home}}. The first dataset contains 22,809,624 check-ins by 114,324 users on 3,820,891 venues. The second dataset contains 90,048,627 check-ins from 2,733,324 users in 11,180,160 venues. These two datasets contain global check-ins with countries like 'US', 'Canada', 'Netherlands', etc. However, only 'US' is in our interest. So, we filter out only the country 'US' and combine these two datasets into one containing 1,24,539 check-ins, of which 51,318 are unique users. Next, the social network dataset includes two snapshots of user social networks before and after the check-in data collection period. The social network data includes 363,704 (old) and 607,333 (new) friendships for the same users mentioned above in the trajectory dataset. Our work aims to filter out these two datasets for the country 'US' and generate two new datasets, which include 95,994 (old) and 1,29,864 (new) friendships. The billboard data sets are crawled from LAMAR\footnote{\url{http://www..lamar.com/InventoryBrowser}}, one of the largest billboard providers worldwide. The dataset for New York City includes 716 billboards, i.e., 10,31,040 billboard slots, and Los Angeles contains 1483 billboards, i.e., 21,35,520 billboard slots. We combine these two datasets to create a new billboard dataset for the 'US'. So, the new dataset contains 31,66,560 billboard slots used in our work. Finally, all experiments were run on an HP Z4 workstation with an Xeon 3.50 GHz CPU and 64 GB RAM. 
\begin{figure*}[ht]
    \centering
    \setlength{\tabcolsep}{2pt} 
    \renewcommand{\arraystretch}{0.9} 
    \begin{tabular}{ccccc}     
        \includegraphics[width=0.185\linewidth]{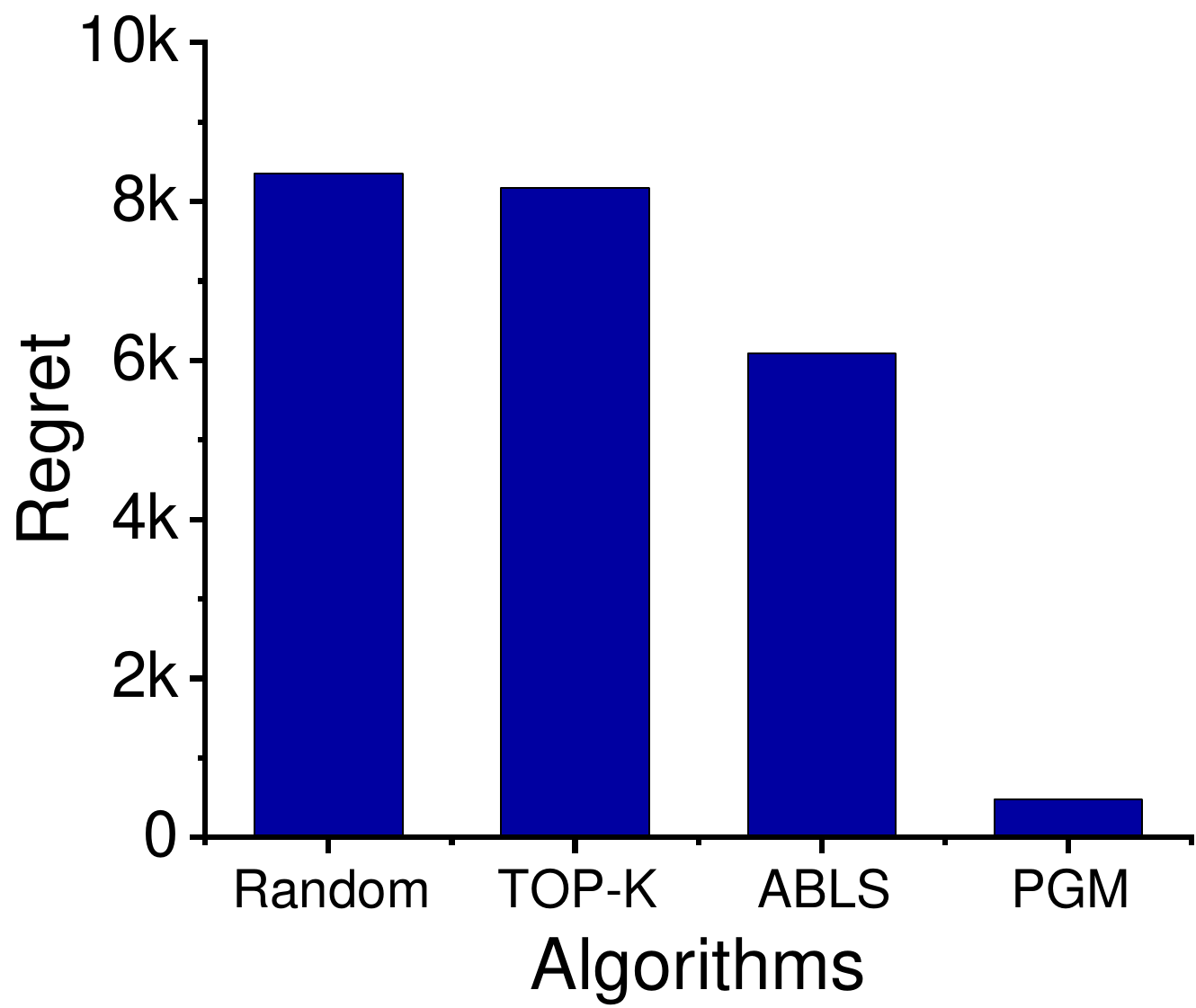} &
        \includegraphics[width=0.185\linewidth]{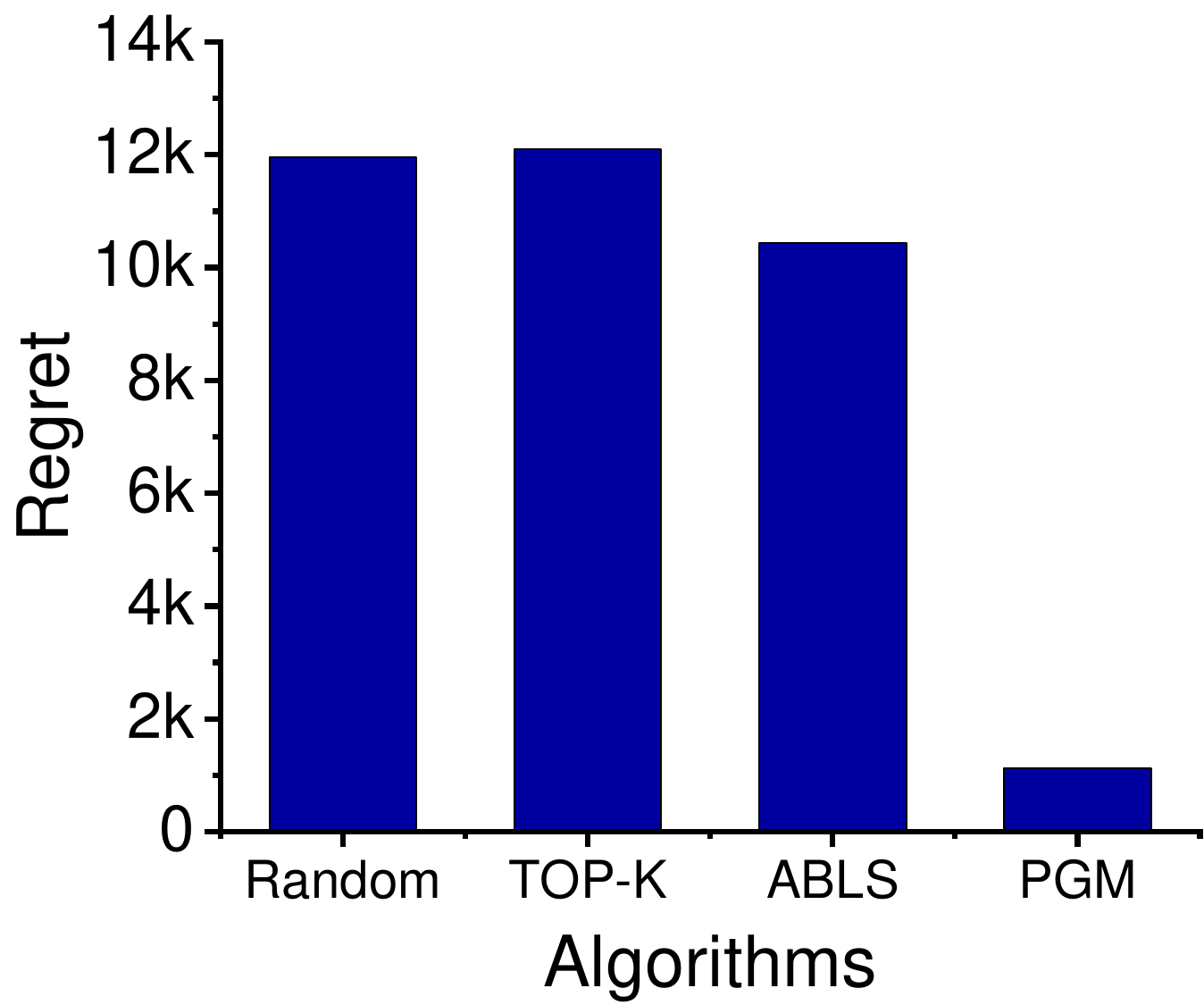} &
        \includegraphics[width=0.185\linewidth]{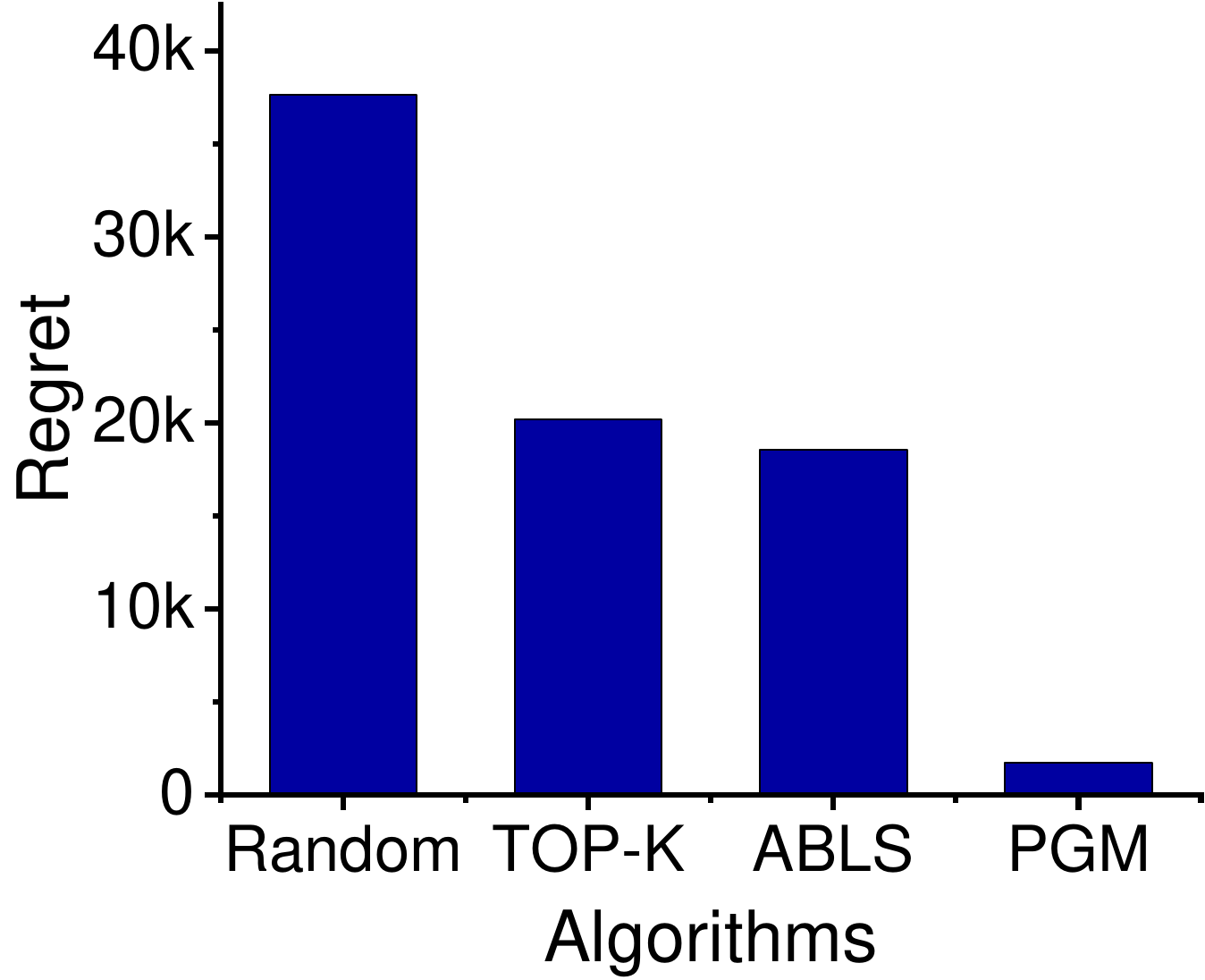} &
        \includegraphics[width=0.185\linewidth]{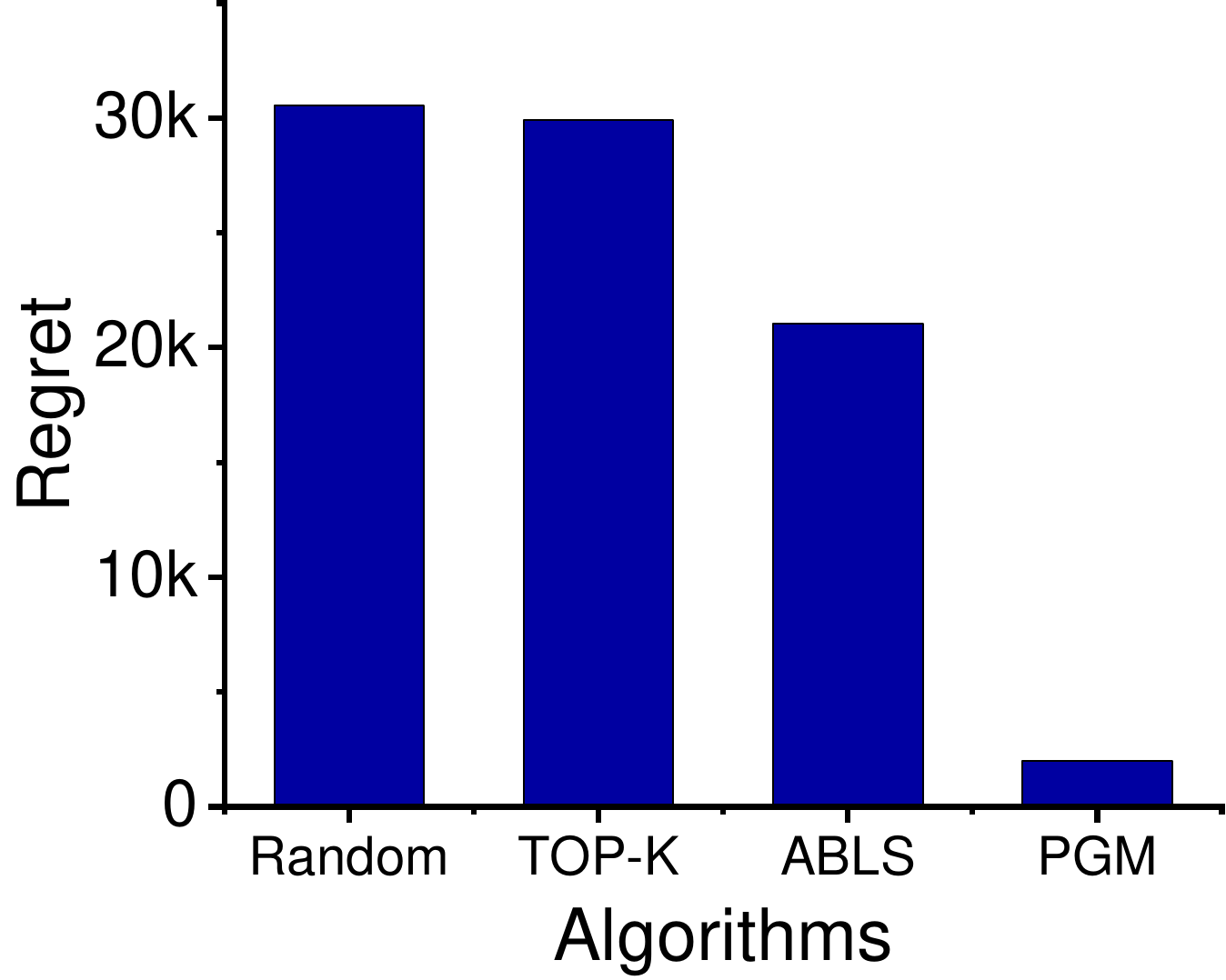} &
        \includegraphics[width=0.185\linewidth]{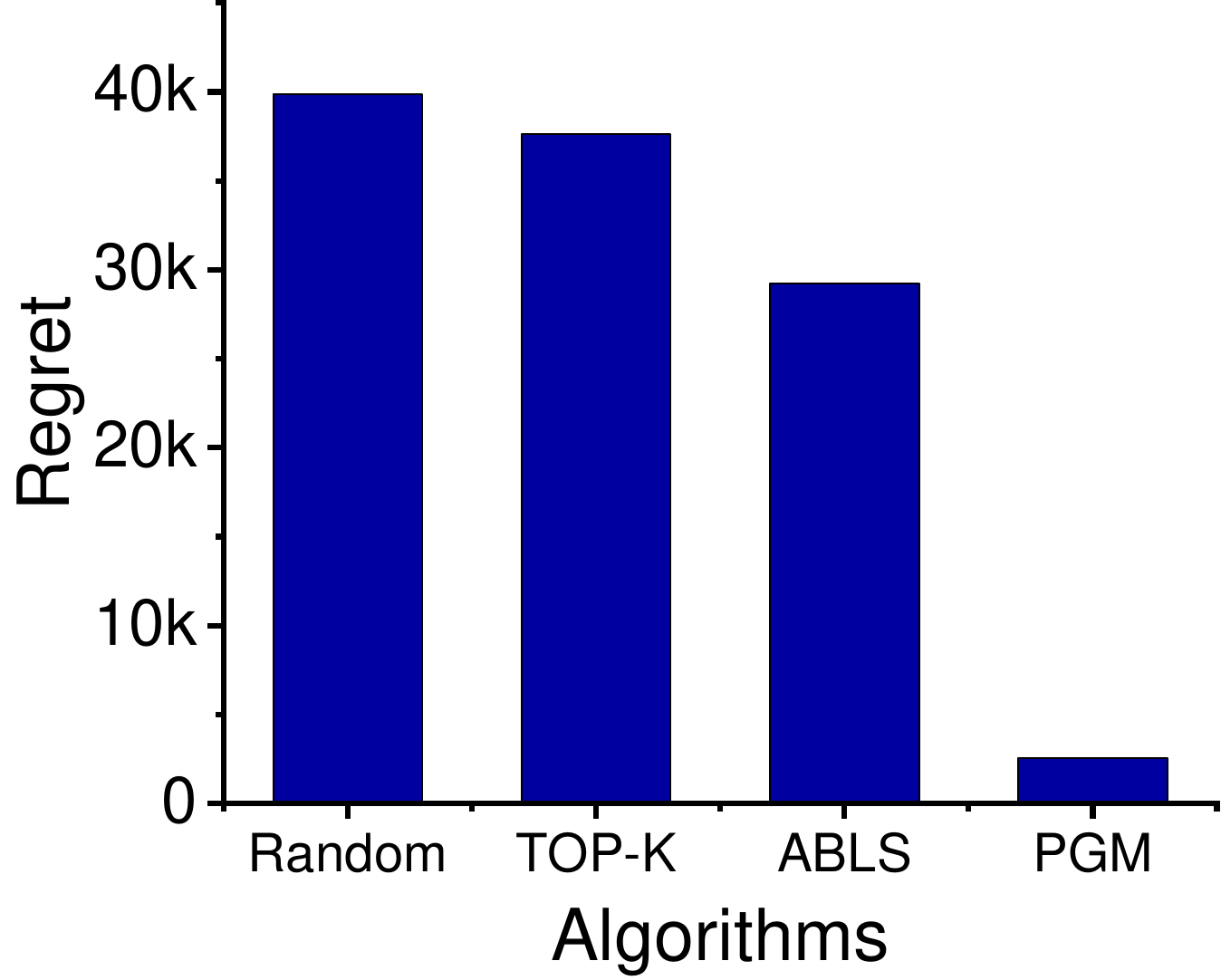} \\
        {\tiny (a) $\alpha = 40 \%$} &
        {\tiny (b) $\alpha = 60 \%$} &
        {\tiny (c) $\alpha = 80 \%$} &
        {\tiny (d) $\alpha = 100 \%$} &
        {\tiny (e) $\alpha = 120 \%$} \\
        \includegraphics[width=0.185\linewidth]{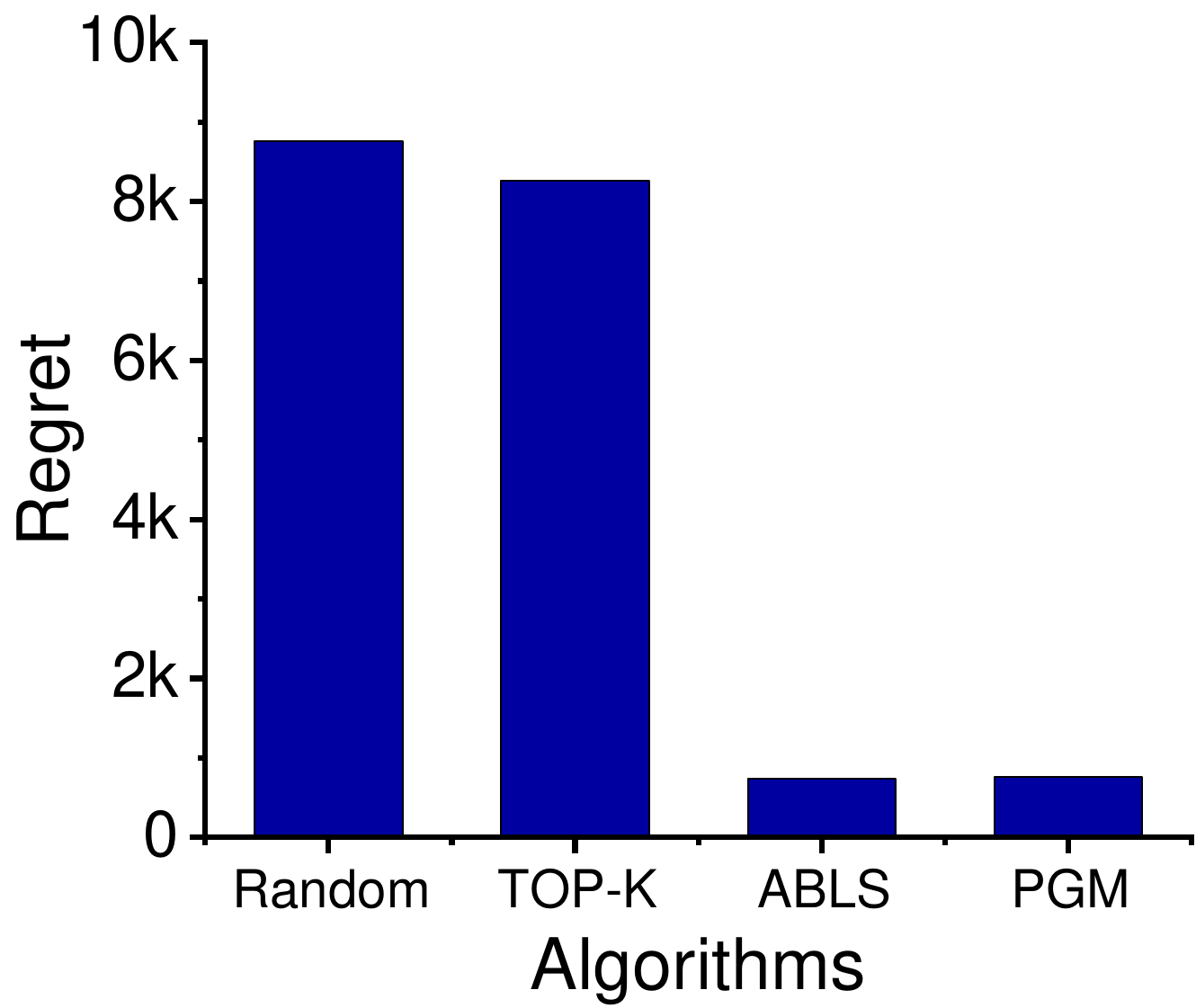} &
        \includegraphics[width=0.185\linewidth]{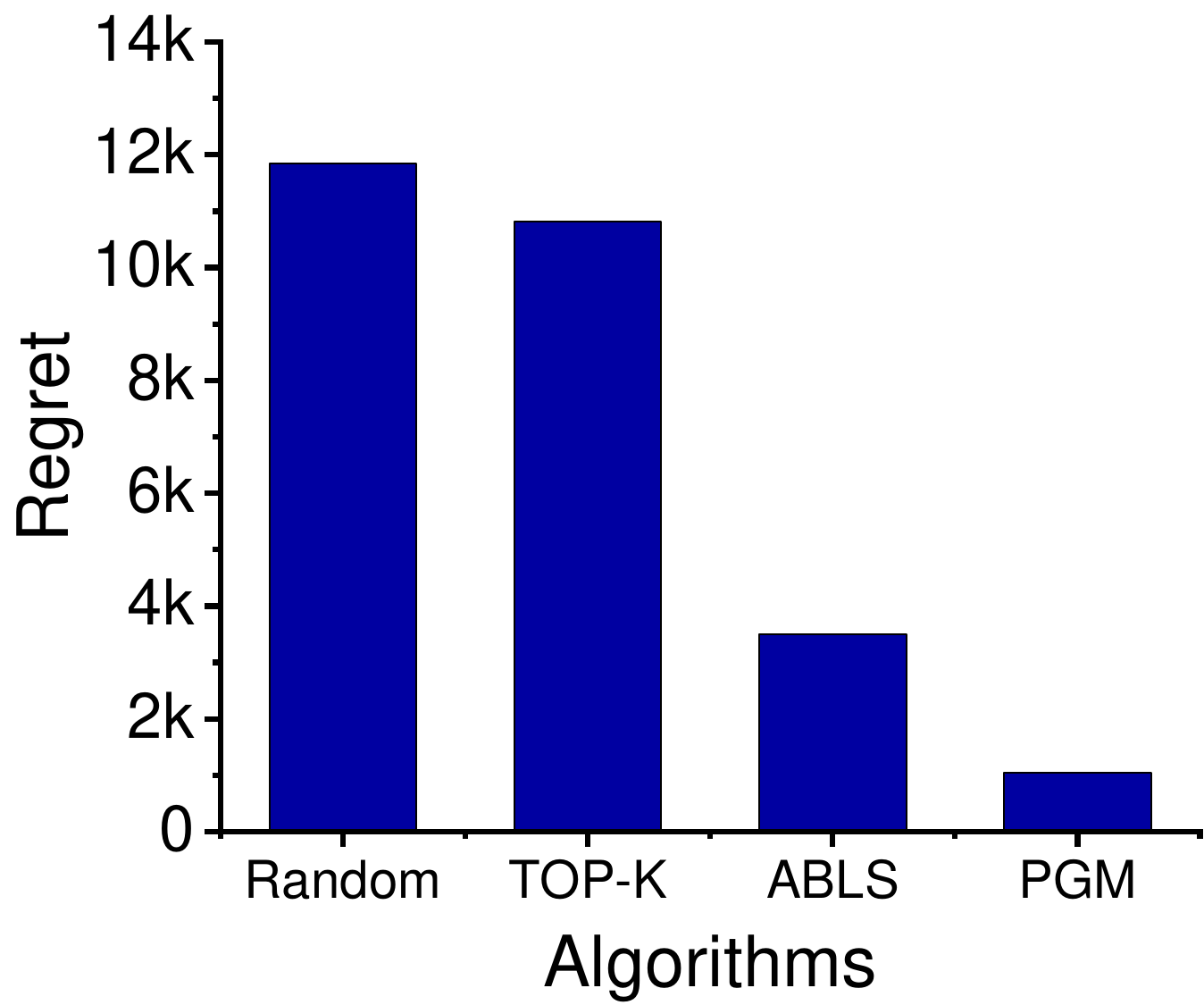} &
        \includegraphics[width=0.185\linewidth]{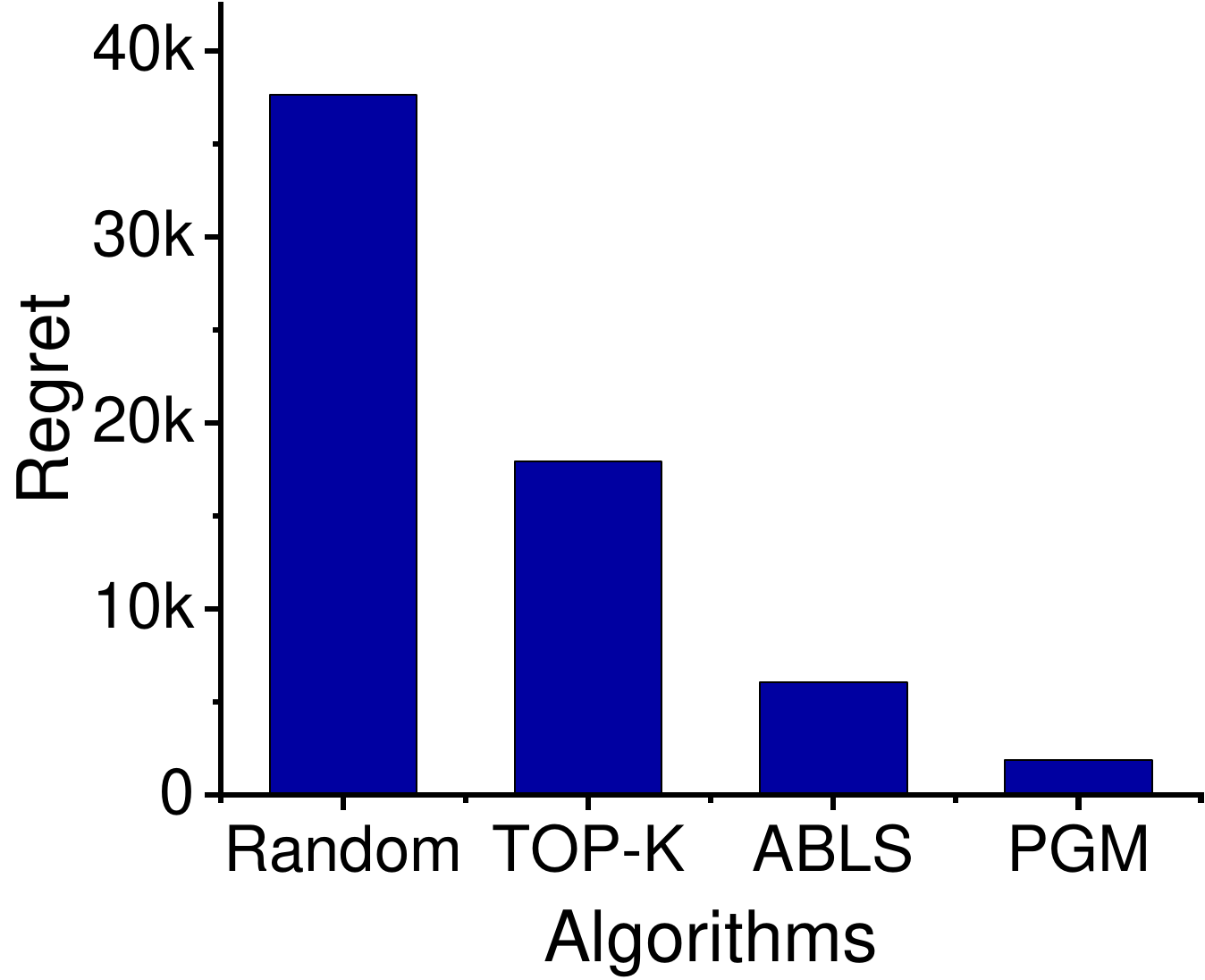} &
        \includegraphics[width=0.185\linewidth]{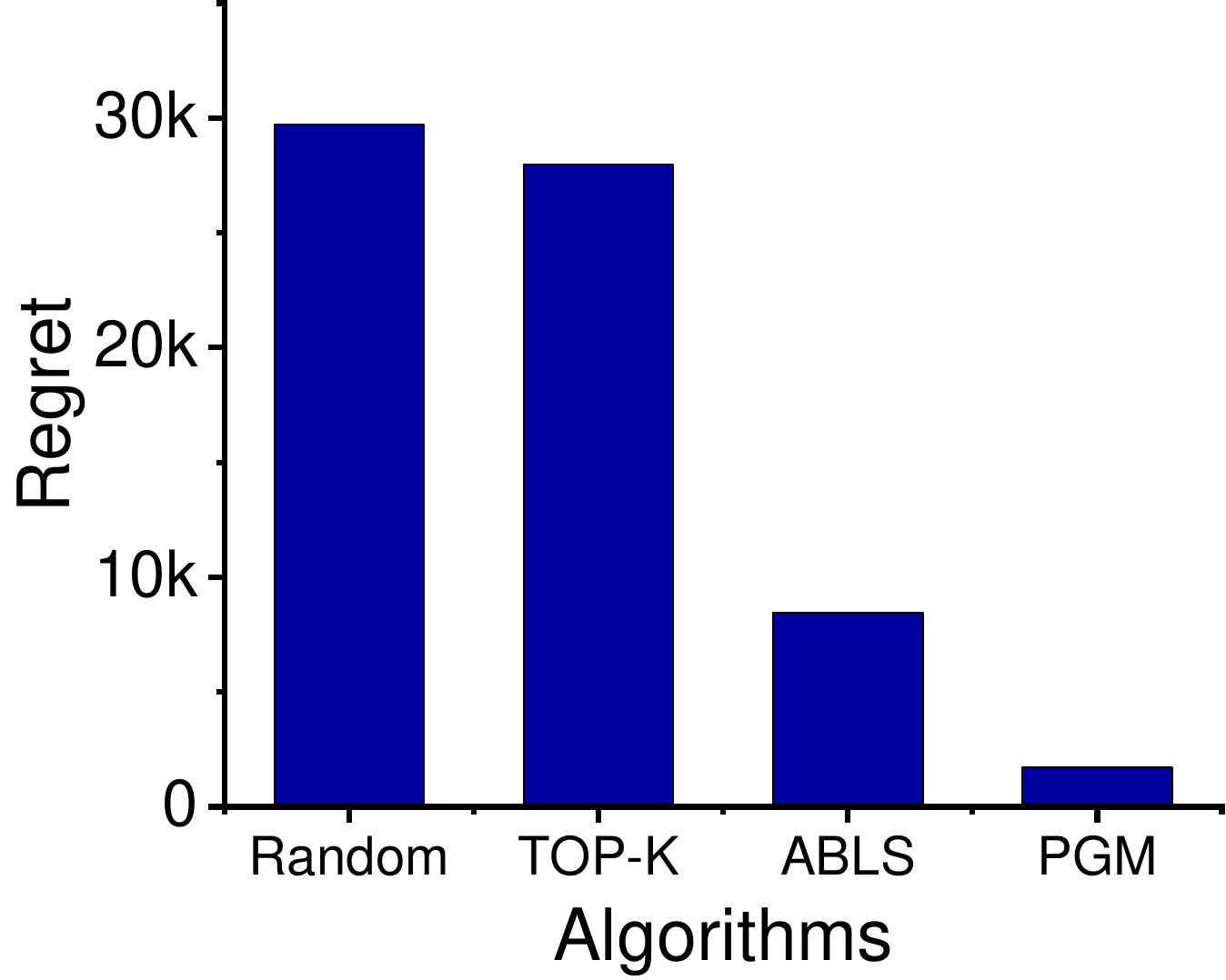} &
        \includegraphics[width=0.185\linewidth]{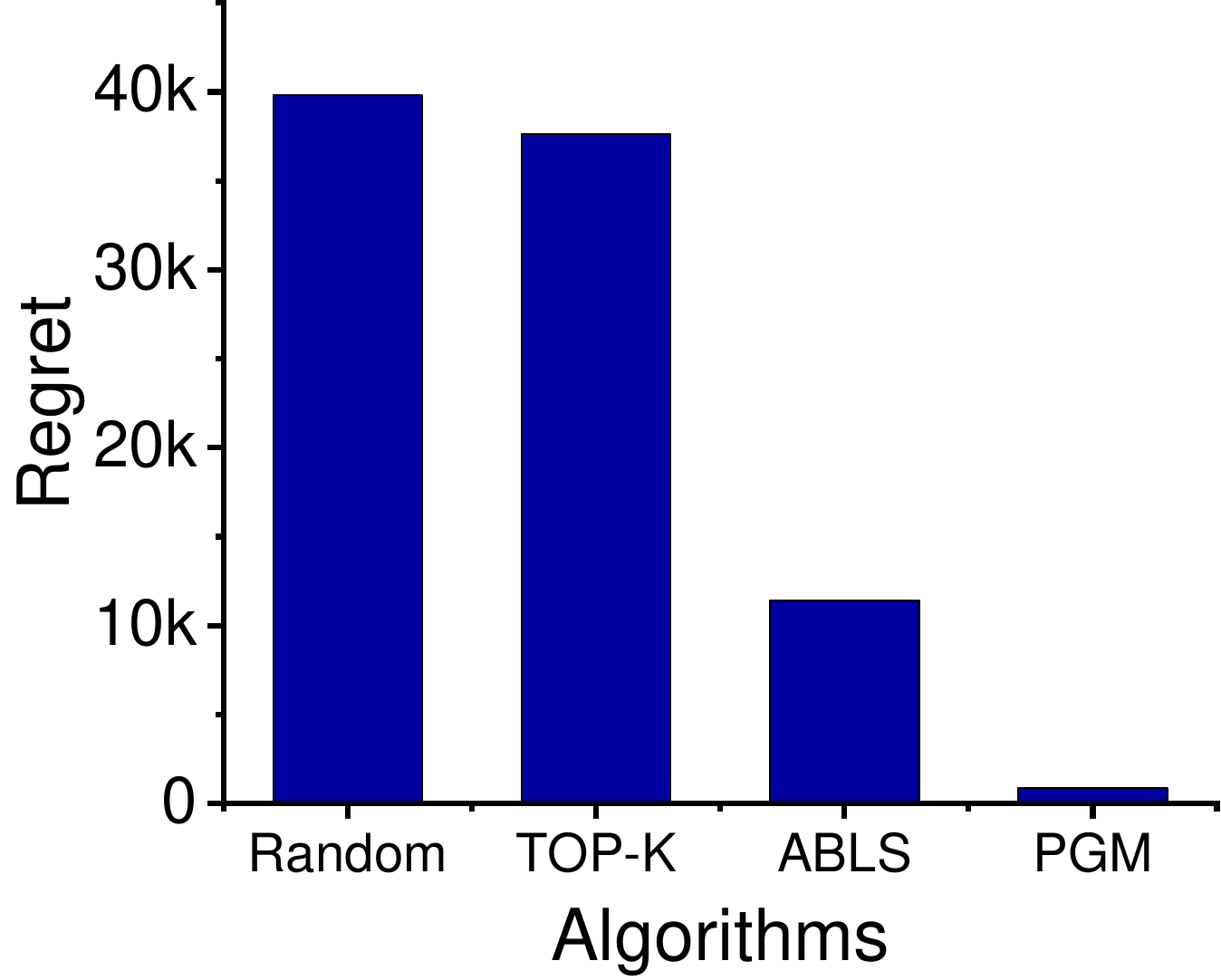} \\
        
        {\tiny (f) $\alpha = 40 \%$} &
        {\tiny (g) $\alpha = 60 \%$} &
        {\tiny (h) $\alpha = 80 \%$} &
        {\tiny (i) $\alpha = 100 \%$} &
        {\tiny (j) $\alpha = 120 \%$} \\
        \includegraphics[width=0.185\linewidth]{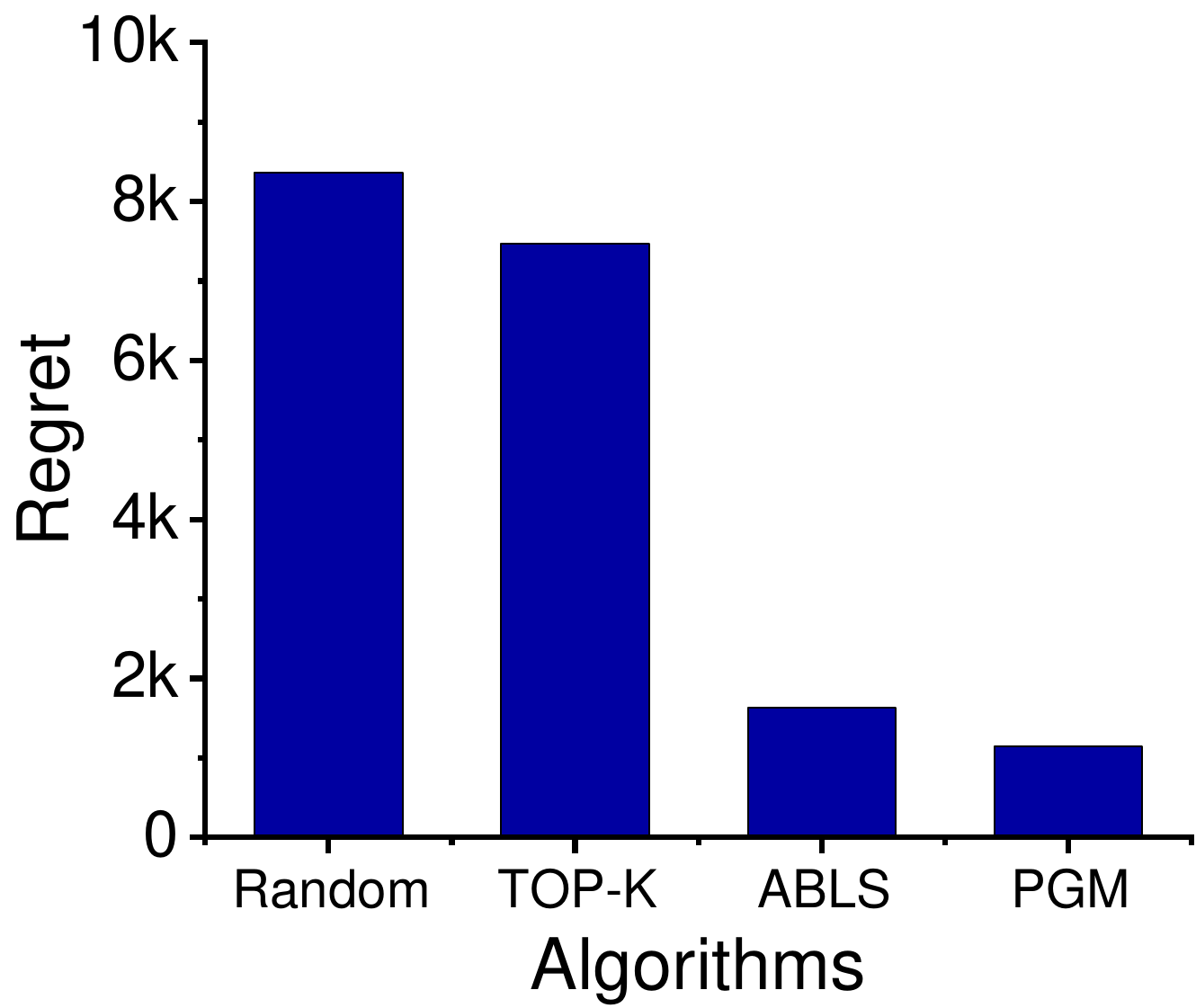} &
        \includegraphics[width=0.185\linewidth]{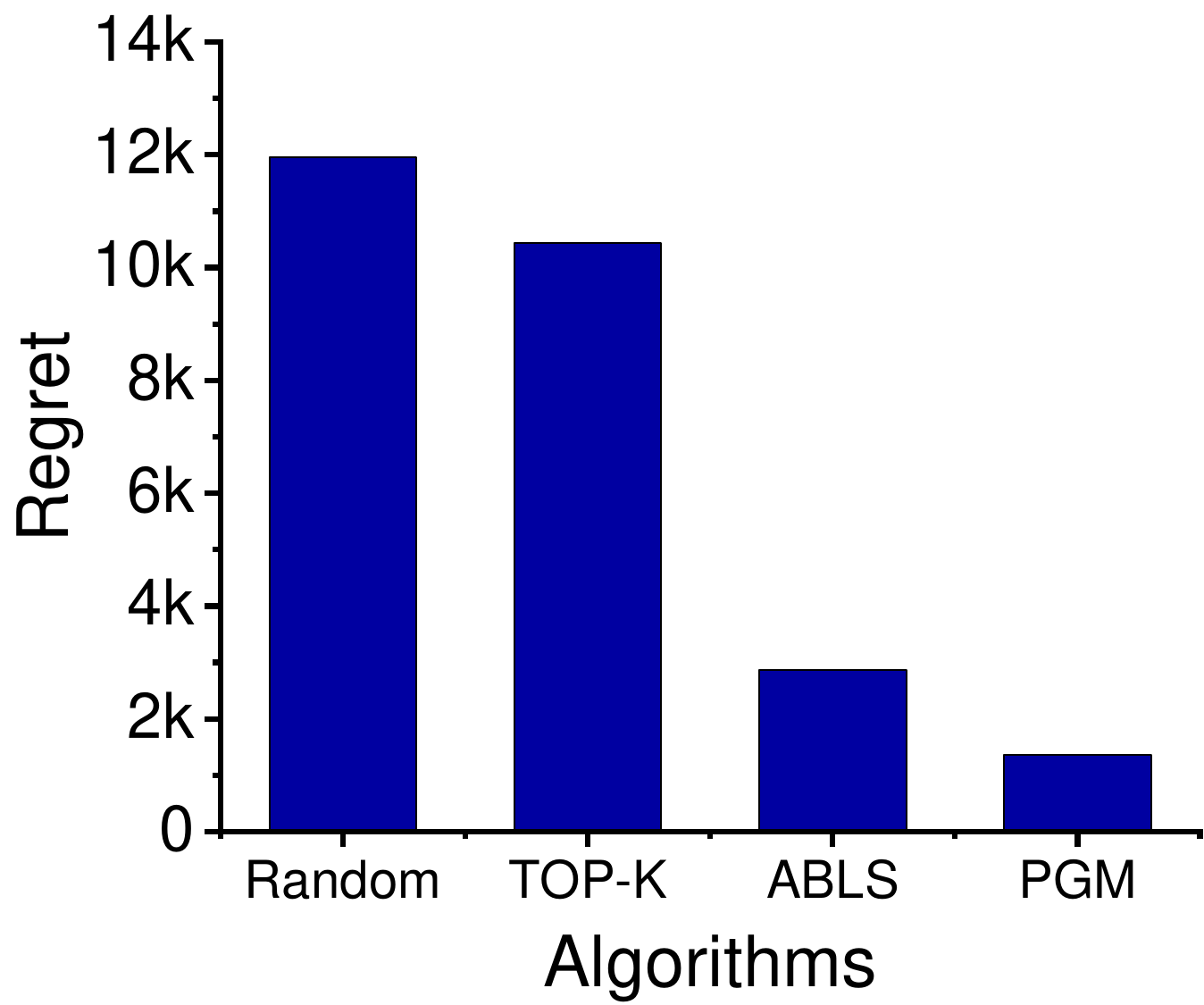} &
        \includegraphics[width=0.185\linewidth]{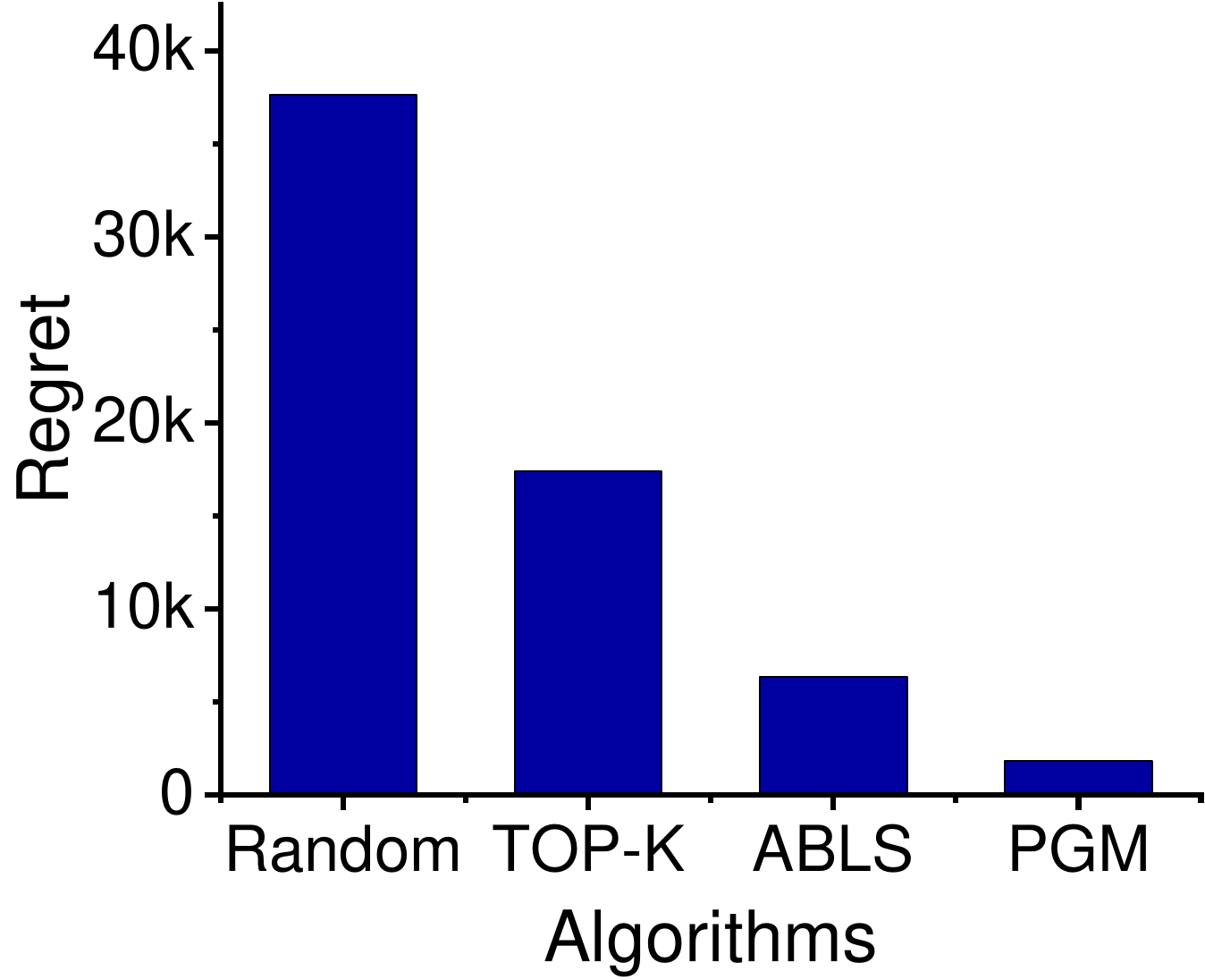} &
        \includegraphics[width=0.185\linewidth]{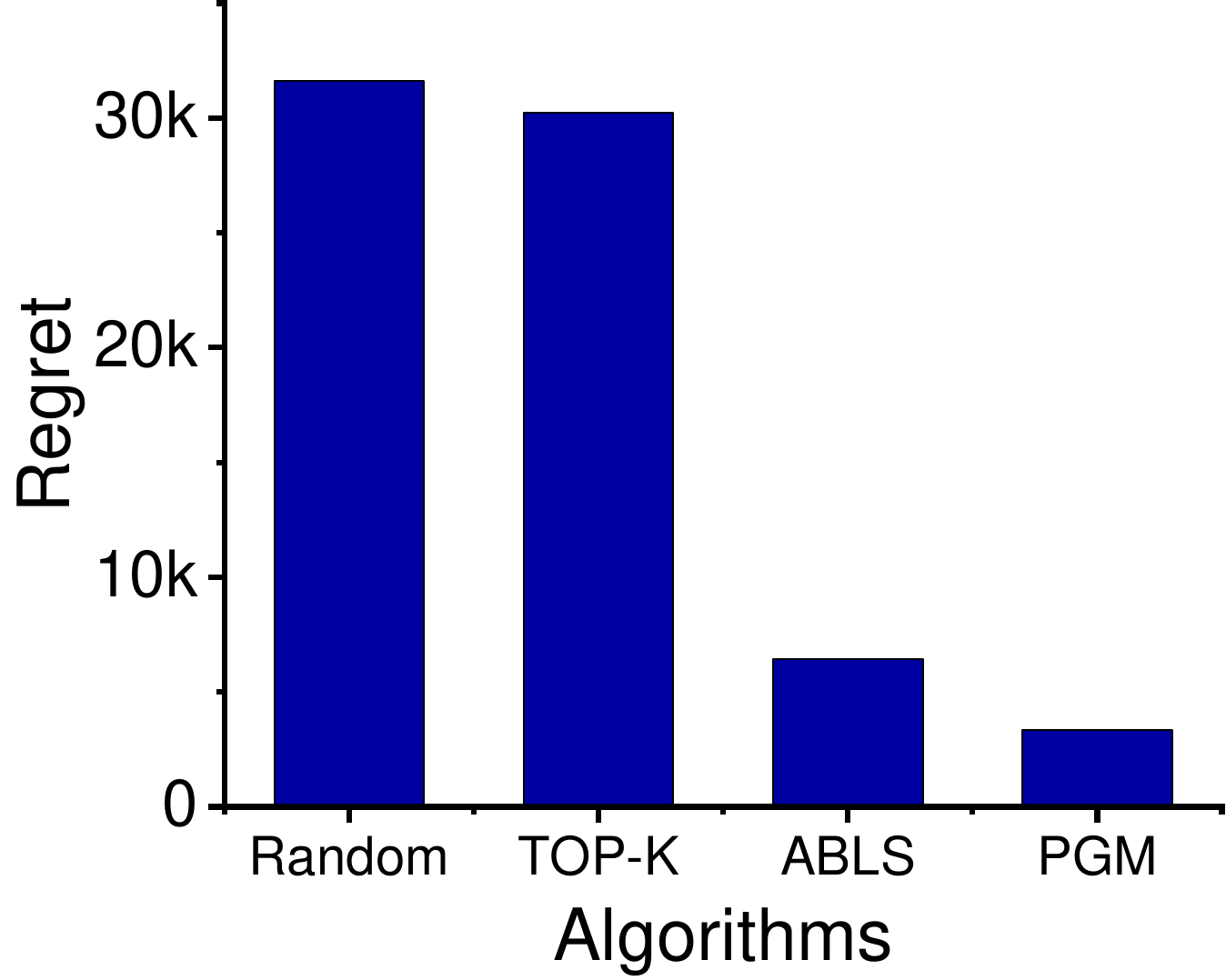} &
        \includegraphics[width=0.185\linewidth]{Regret/Weighted/Regret2W.pdf} \\
        {\tiny (k) $\alpha = 40 \%$} &
        {\tiny $(\ell)$ $\alpha = 60 \%$} &
        {\tiny (m) $\alpha = 80 \%$} &
        {\tiny (n) $\alpha = 100 \%$} &
        {\tiny (o) $\alpha = 120 \%$} \\
    \end{tabular}
    \caption{Regret on varying $\alpha$, when $\lambda = 5\%, \mathcal{|A|} = 20$ for Uniform (a-e), Trivalency (f-j), Weighted Cascade (k-o)Settings}
    \label{Fig:Regret-Uniform}
\end{figure*}

\begin{figure*}[ht]
    \centering
    \setlength{\tabcolsep}{2pt} 
    \renewcommand{\arraystretch}{0.9} 
    \begin{tabular}{ccccc}          
        \includegraphics[width=0.185\linewidth]{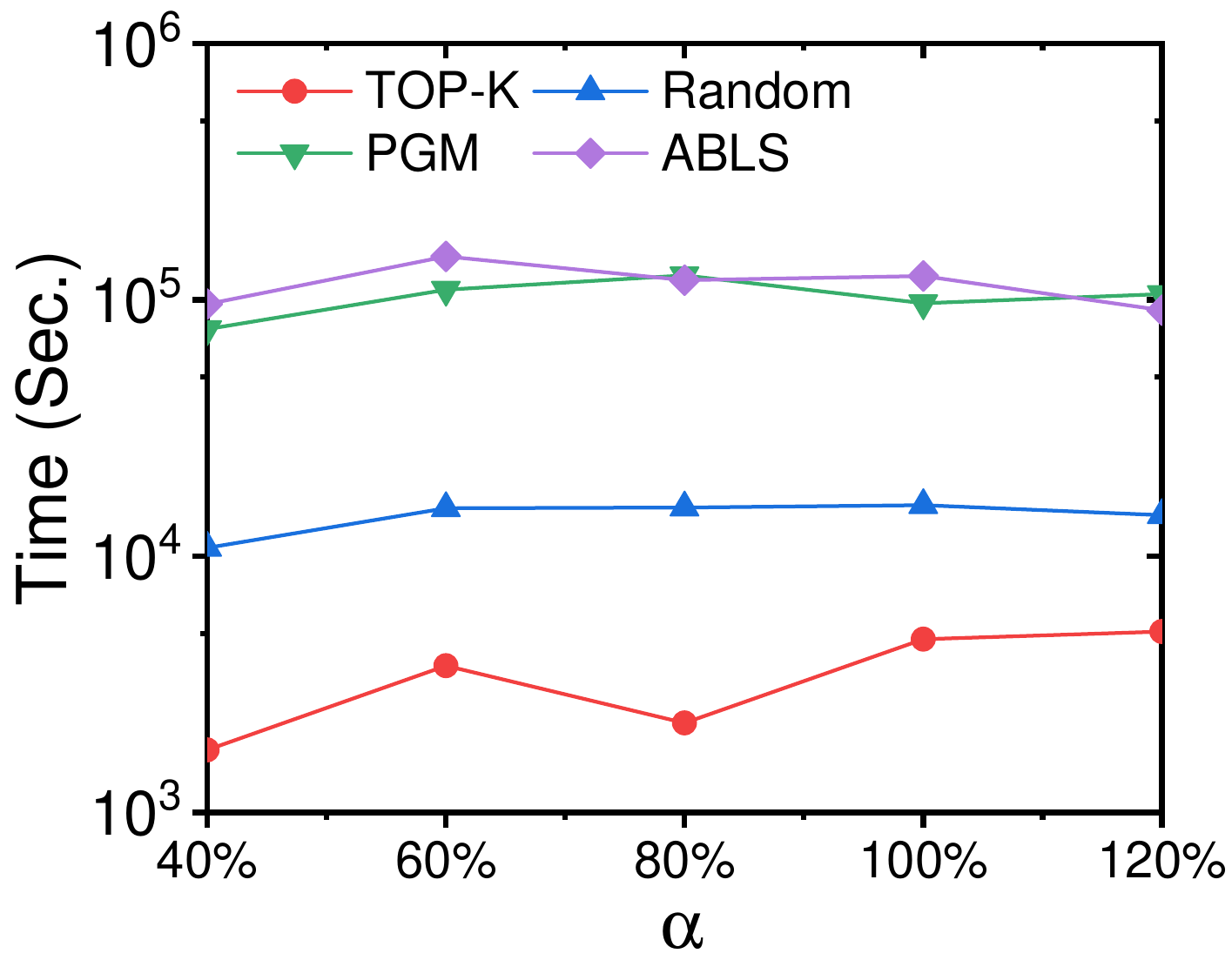} &
        \includegraphics[width=0.185\linewidth]{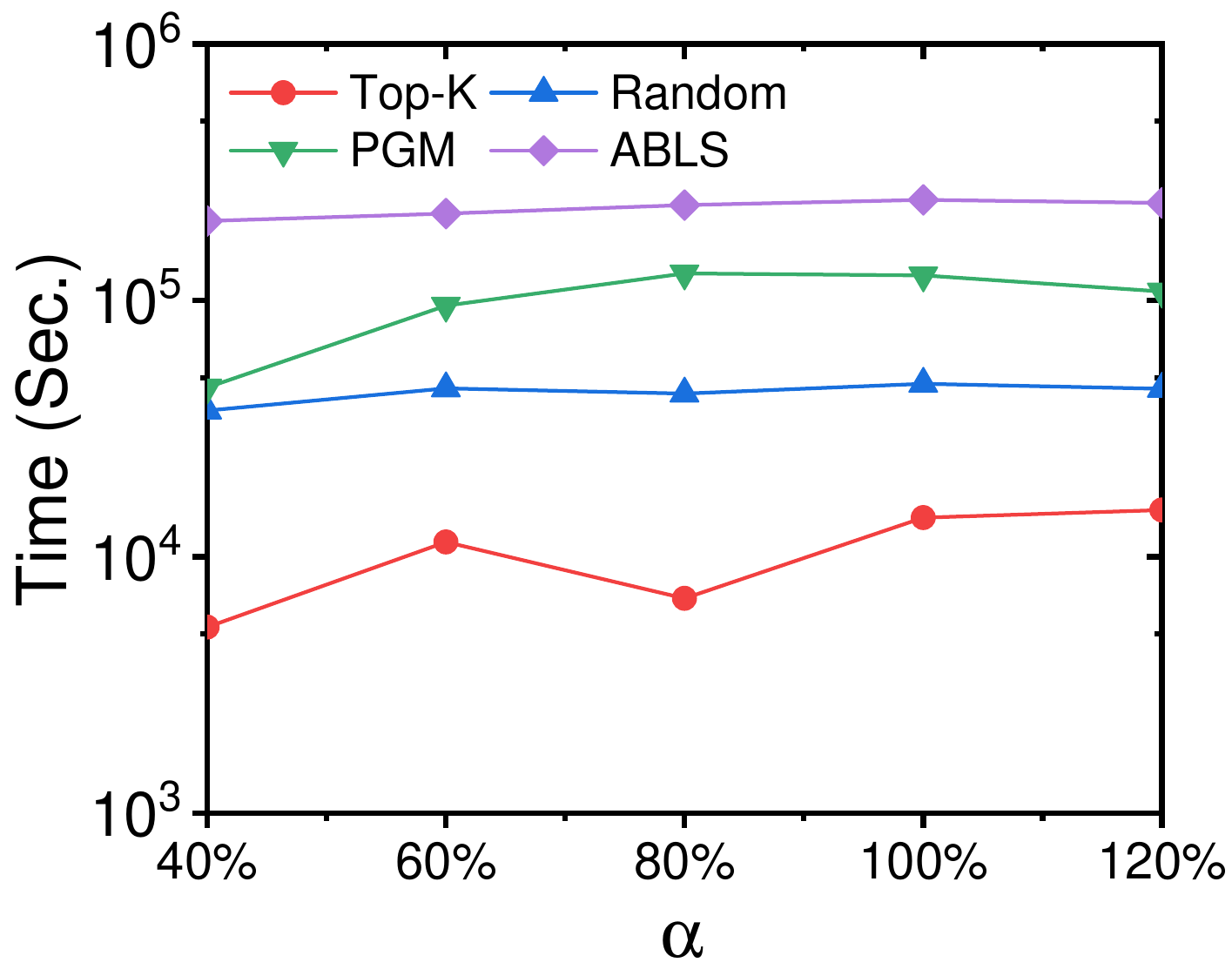} &
        \includegraphics[width=0.185\linewidth]{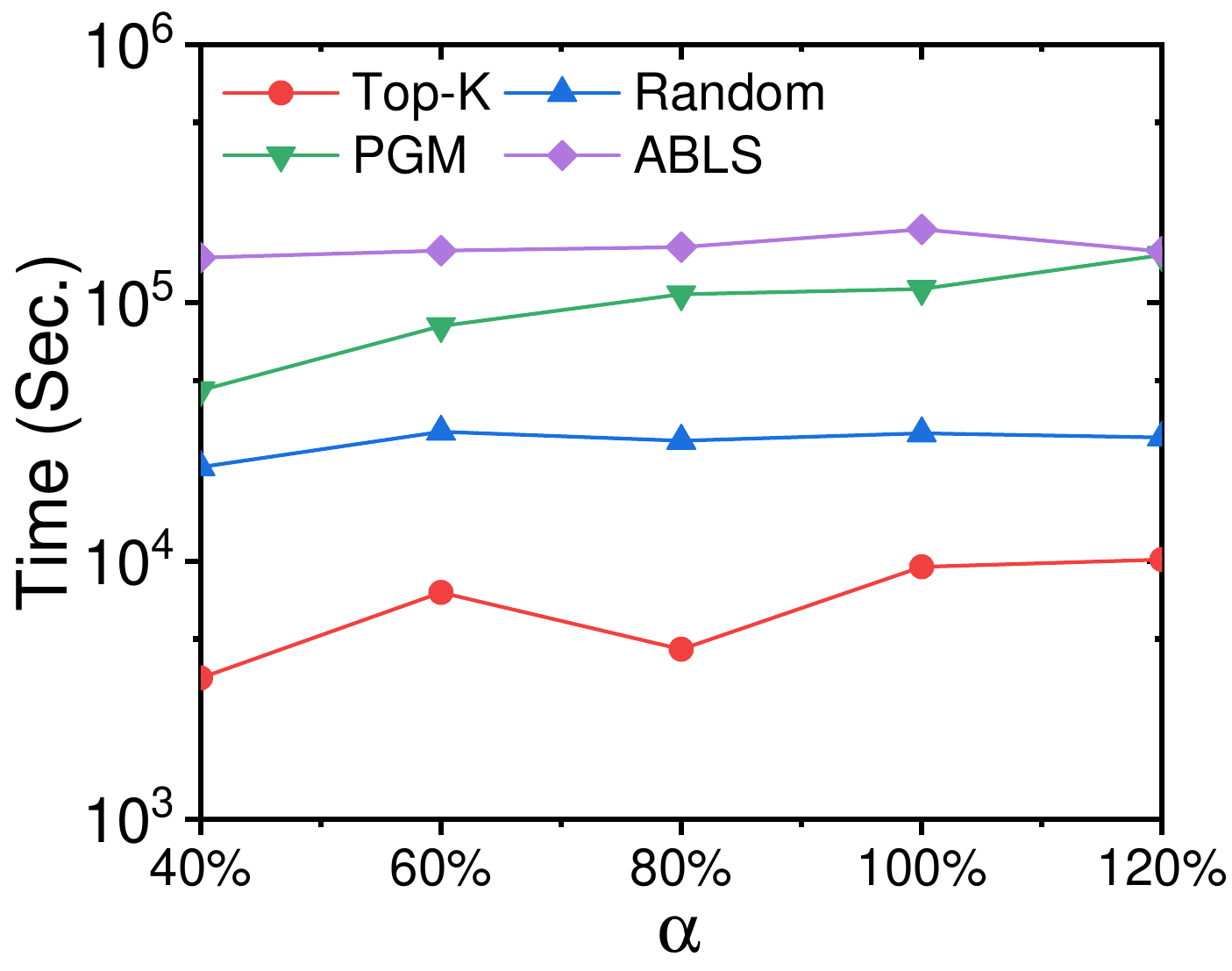} &
        \includegraphics[width=0.185\linewidth]{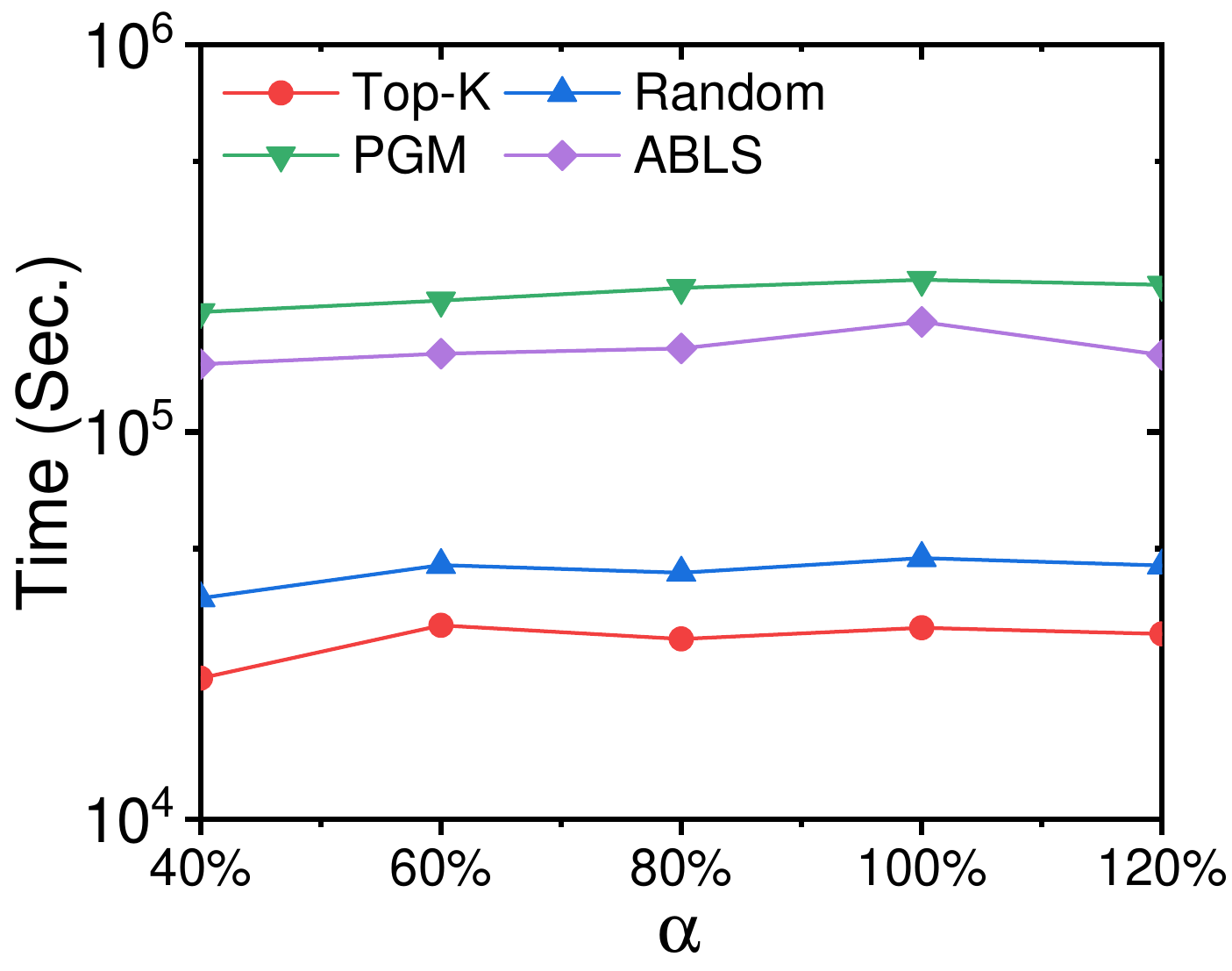} & 
        \includegraphics[width=0.185\linewidth]{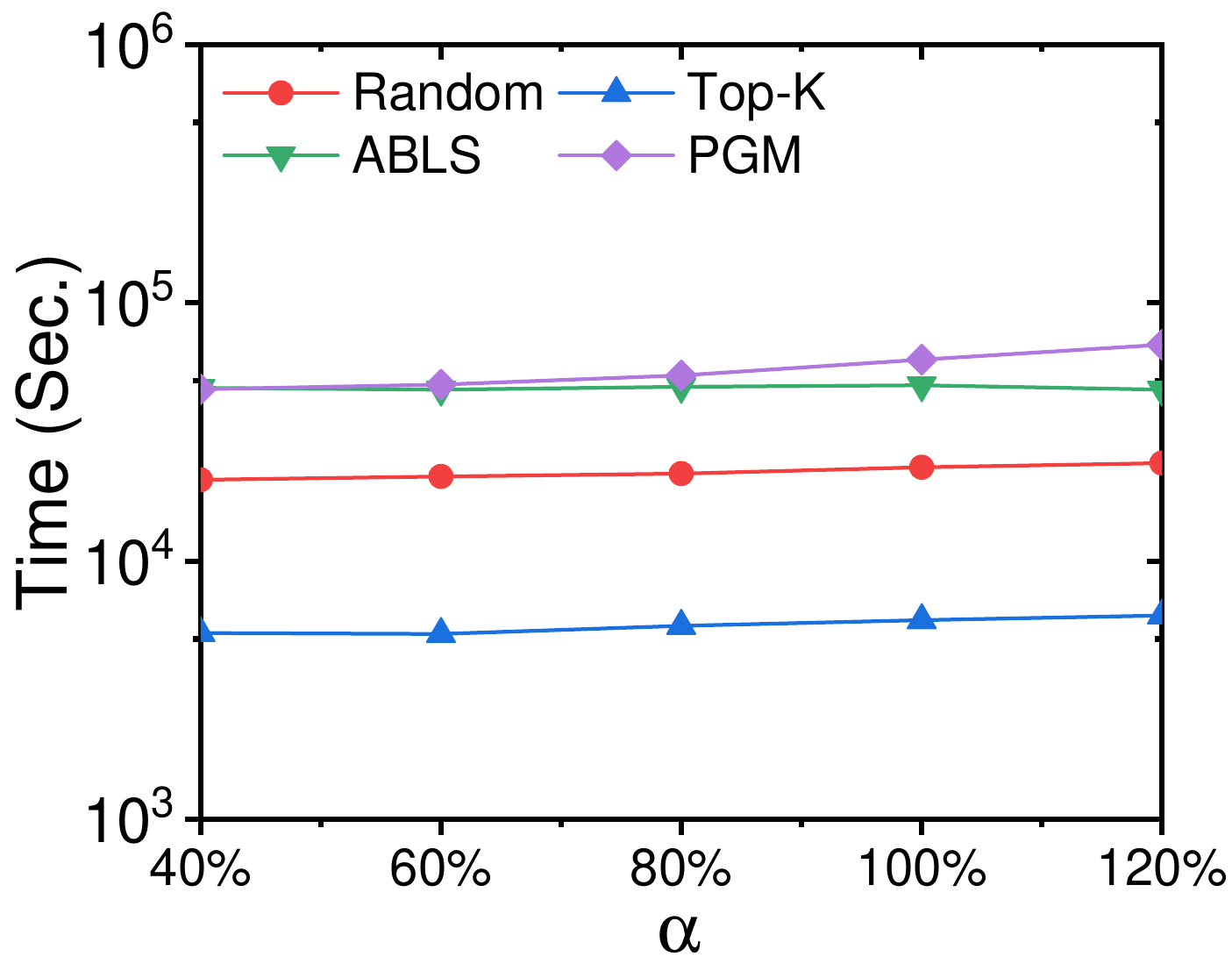} \\

        {\tiny (a) $\lambda = 5\%, \mathcal{|A|} = 20$} &
        {\tiny (b) $\lambda = 5\%, \mathcal{|A|} = 20$} &
        {\tiny (c) $\lambda = 5\%, \mathcal{|A|} = 20$} &
        {\tiny (d) $\lambda = 1\%, \mathcal{|A|} = 100$} &
        {\tiny (e) $\lambda = 20\%, \mathcal{|A|} = 5$} \\
        \includegraphics[width=0.185\linewidth]{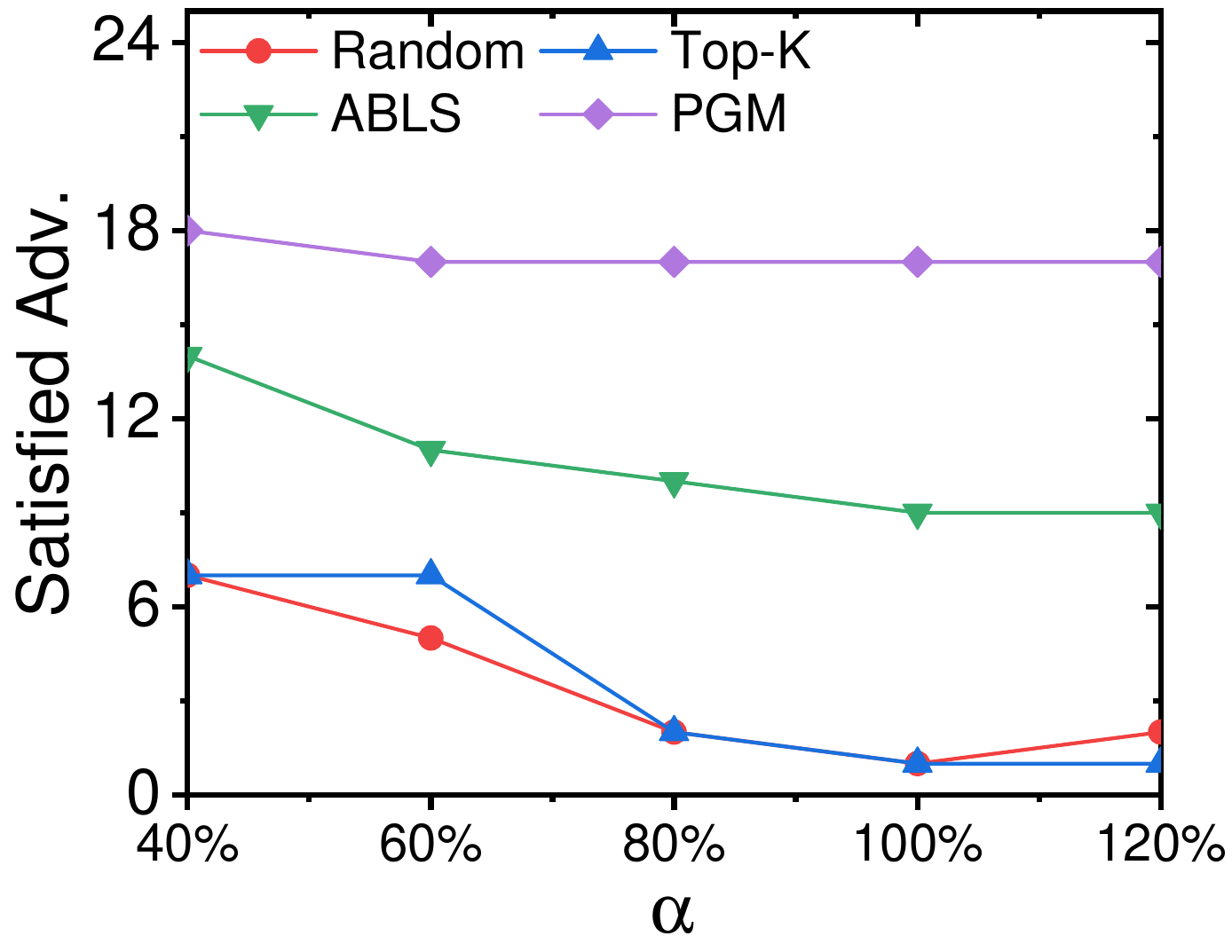} & 
        \includegraphics[width=0.185\linewidth]{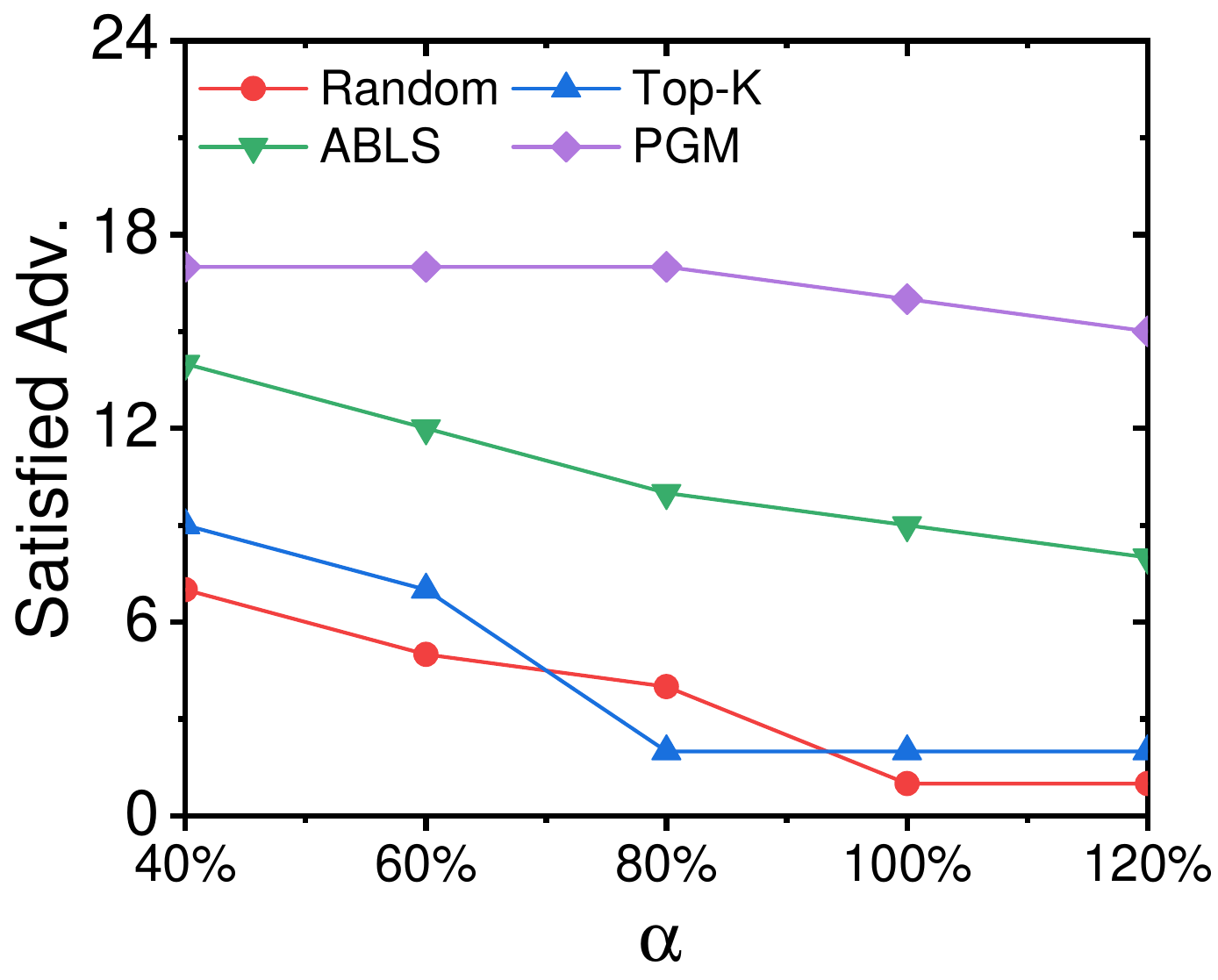} &
        \includegraphics[width=0.185\linewidth]{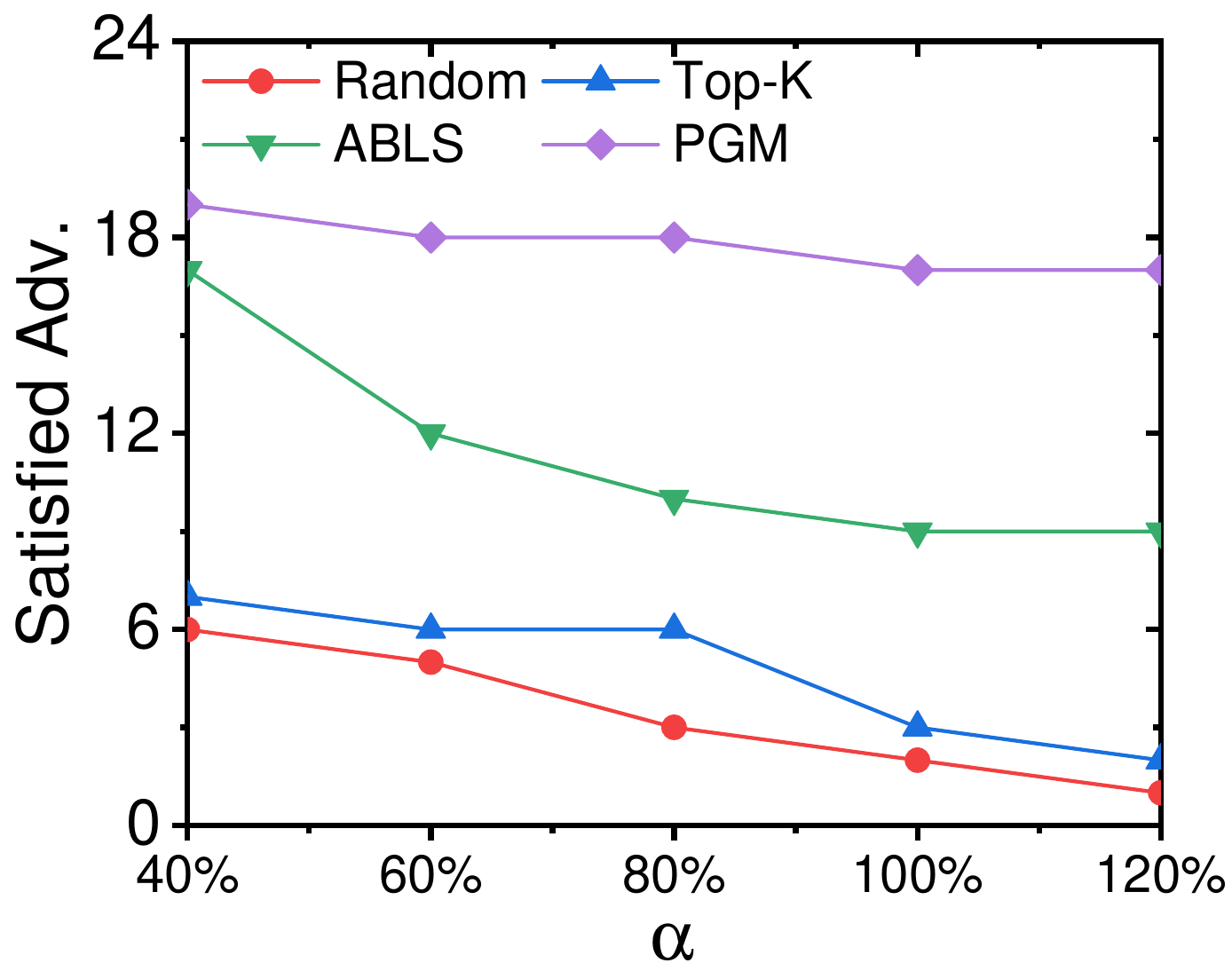} &
        \includegraphics[width=0.185\linewidth]{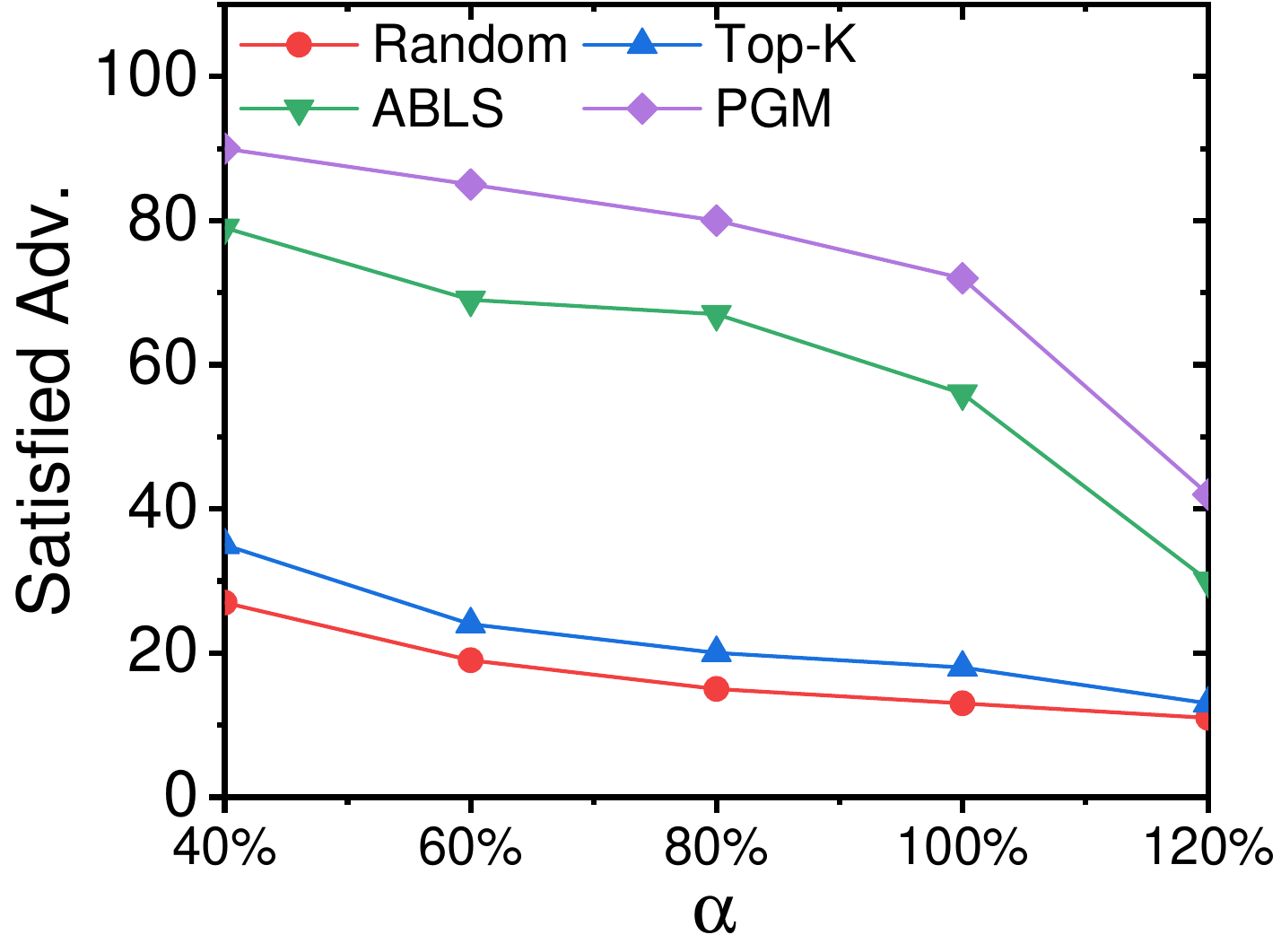}&
        \includegraphics[width=0.185\linewidth]{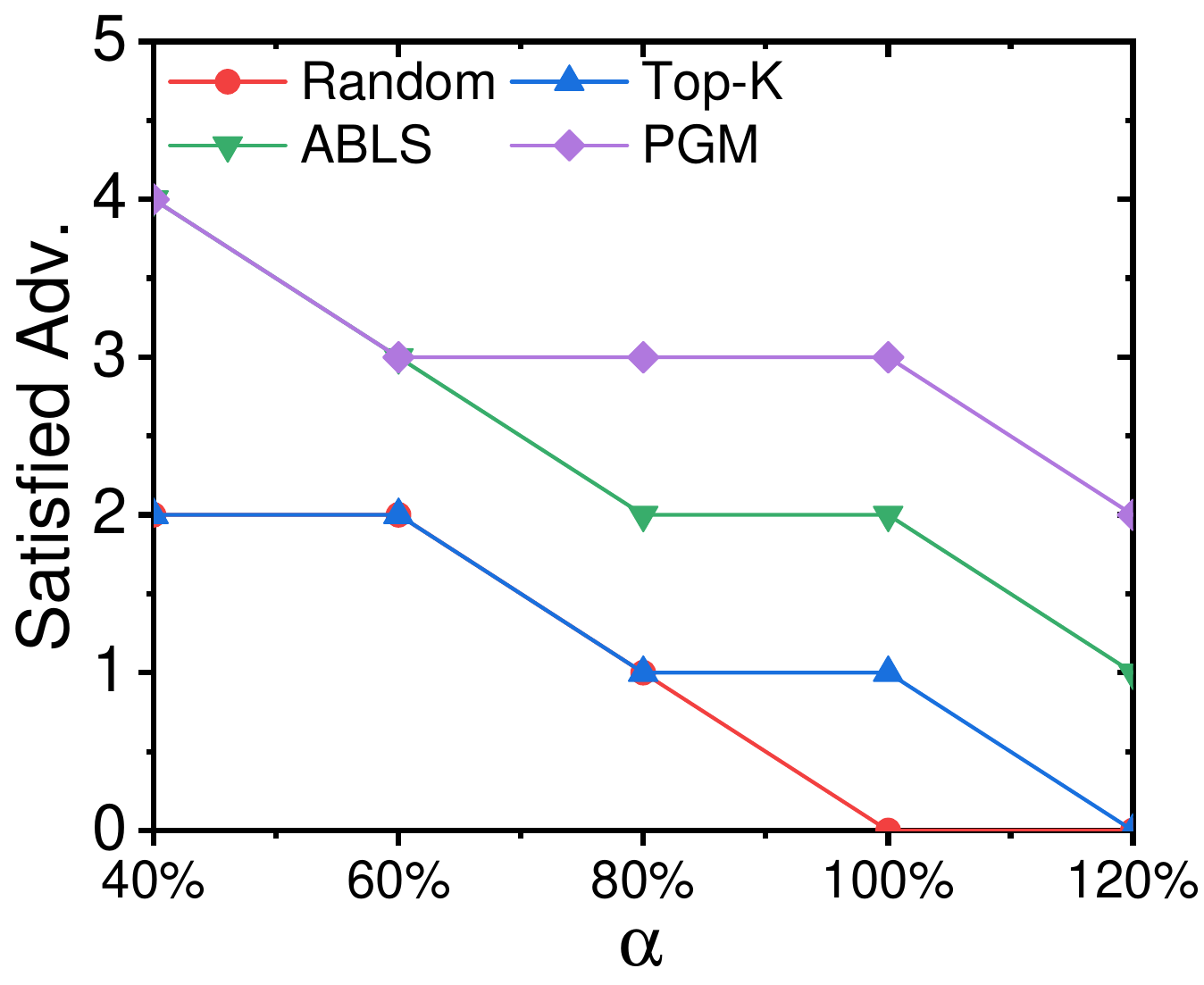}\\

        {\tiny (f) $\lambda = 5\%, \mathcal{|A|} = 20$} &
        {\tiny (g) $\lambda = 5\%, \mathcal{|A|} = 20$} &
        {\tiny (h) $\lambda = 5\%, \mathcal{|A|} = 20$} &
        {\tiny (i) $\lambda = 1\%, \mathcal{|A|} = 100$} &
        {\tiny (j) $\lambda = 20\%, \mathcal{|A|} = 5$} \\
    \end{tabular}
    \caption{Time on varying $\alpha$, $(a)$ Uniform, $(b)$ Weighted, $(c,d,e)$ Trivalency probability Setting and Satisfied Adv. on Varying $\alpha$,$(f)$ Uniform, $(g)$ Weighted, $(i,j,k)$ Trivalency probability Setting}
    \label{Fig:Regret2}
\end{figure*}
\begin{figure*}[ht]
    \centering
    \setlength{\tabcolsep}{2pt} 
    \renewcommand{\arraystretch}{0.9} 
    \begin{tabular}{ccccc}     
        \includegraphics[width=0.185\linewidth]{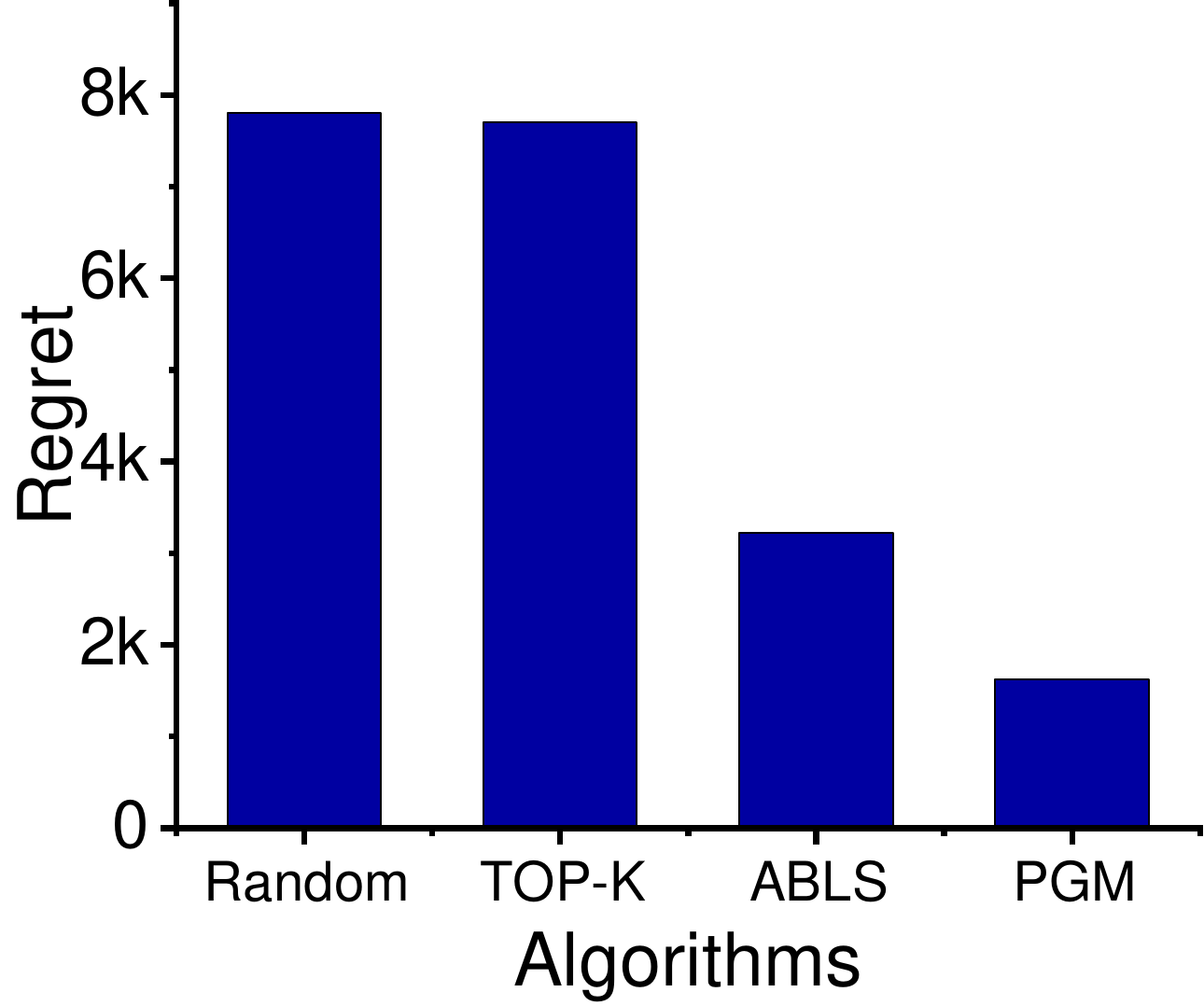} &
        \includegraphics[width=0.185\linewidth]{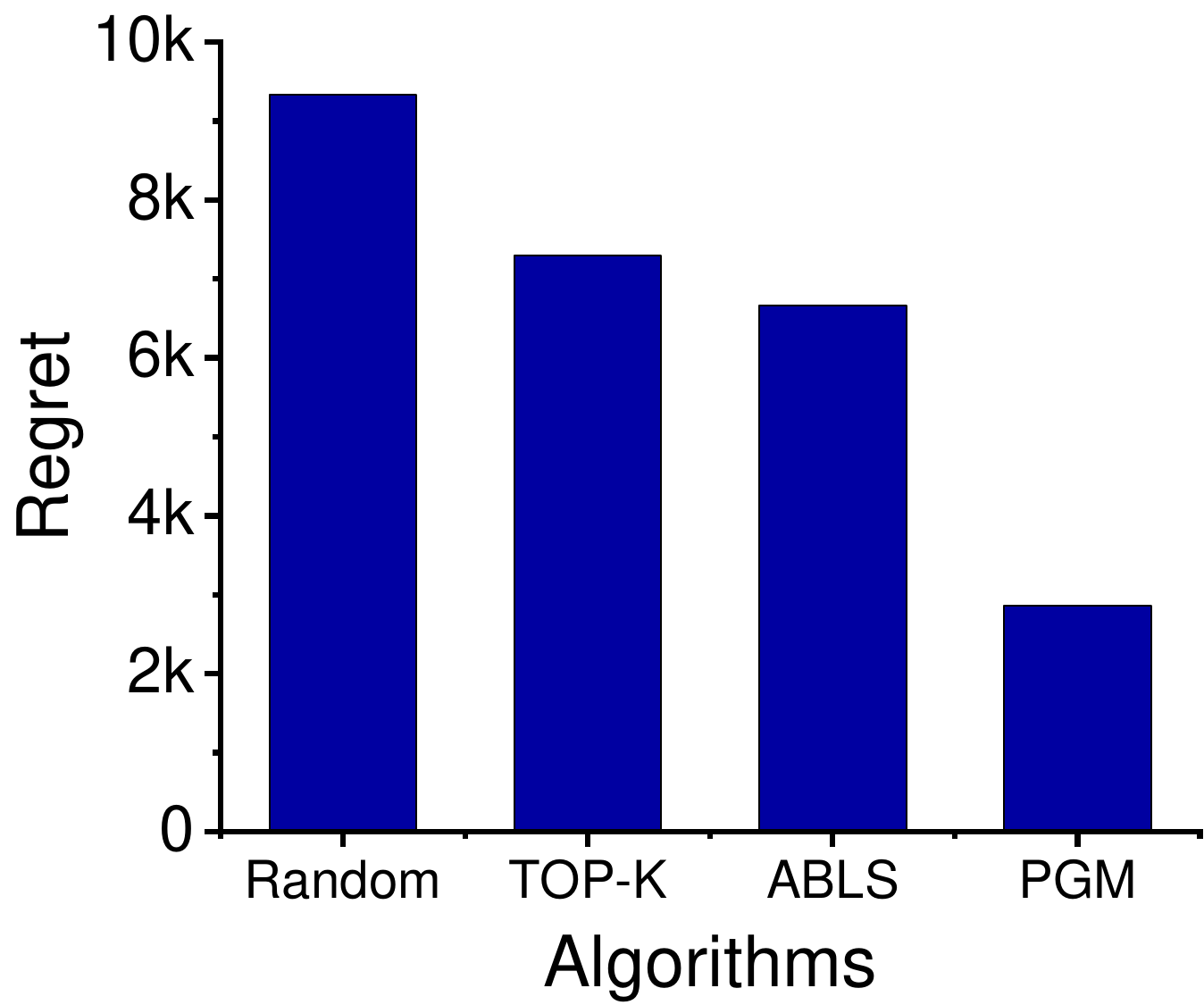} &
        \includegraphics[width=0.185\linewidth]{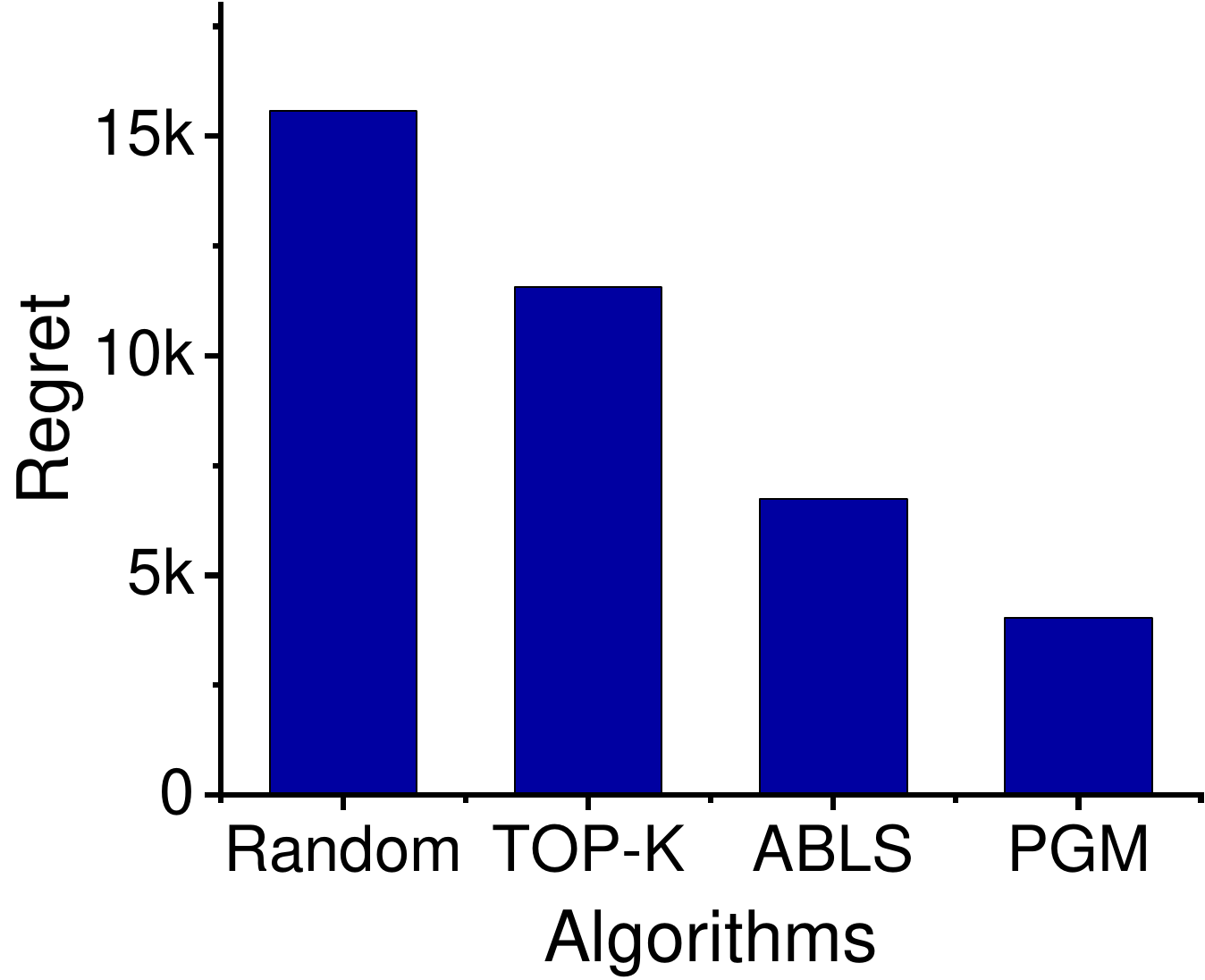} &
        \includegraphics[width=0.185\linewidth]{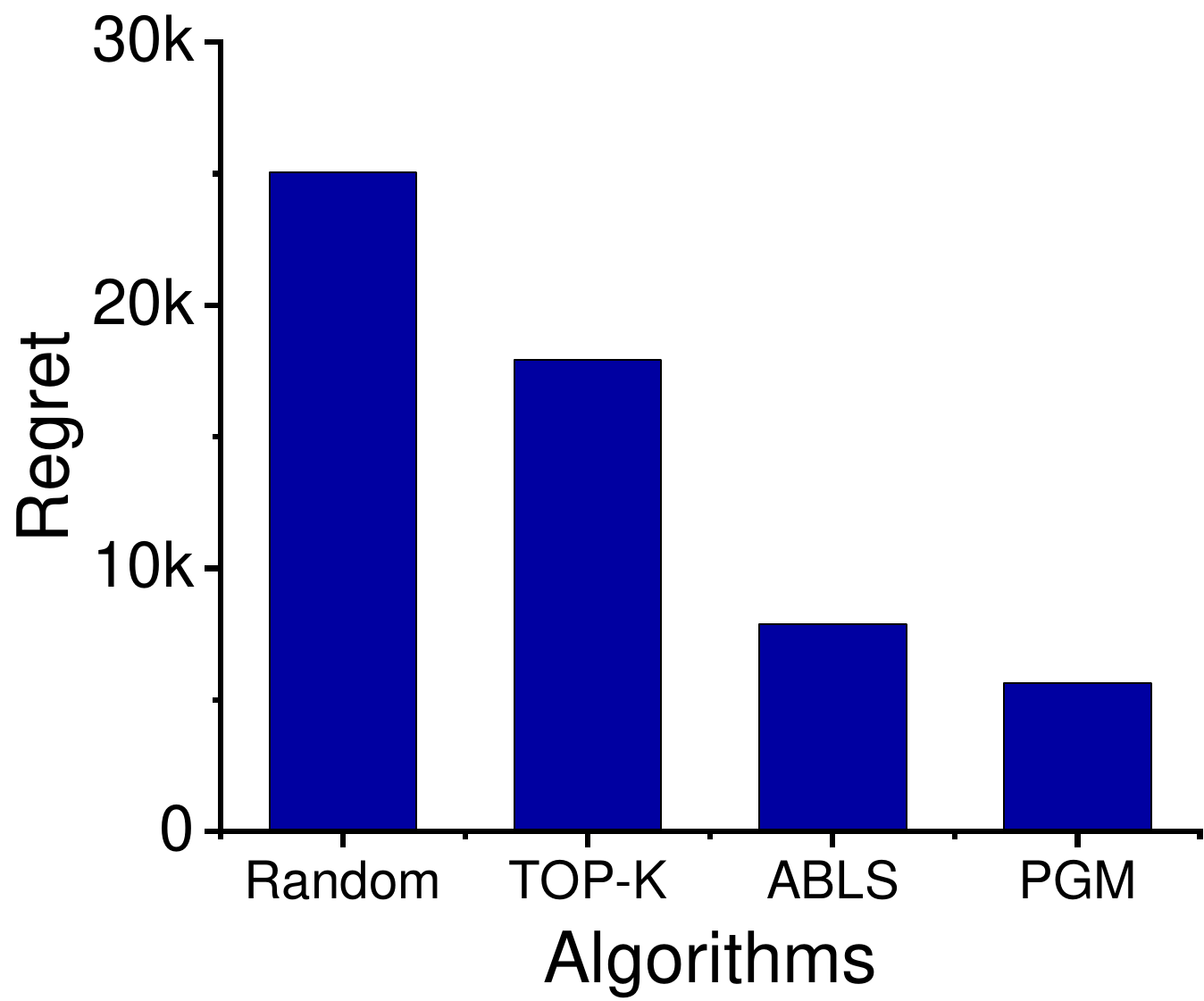} &
        \includegraphics[width=0.185\linewidth]{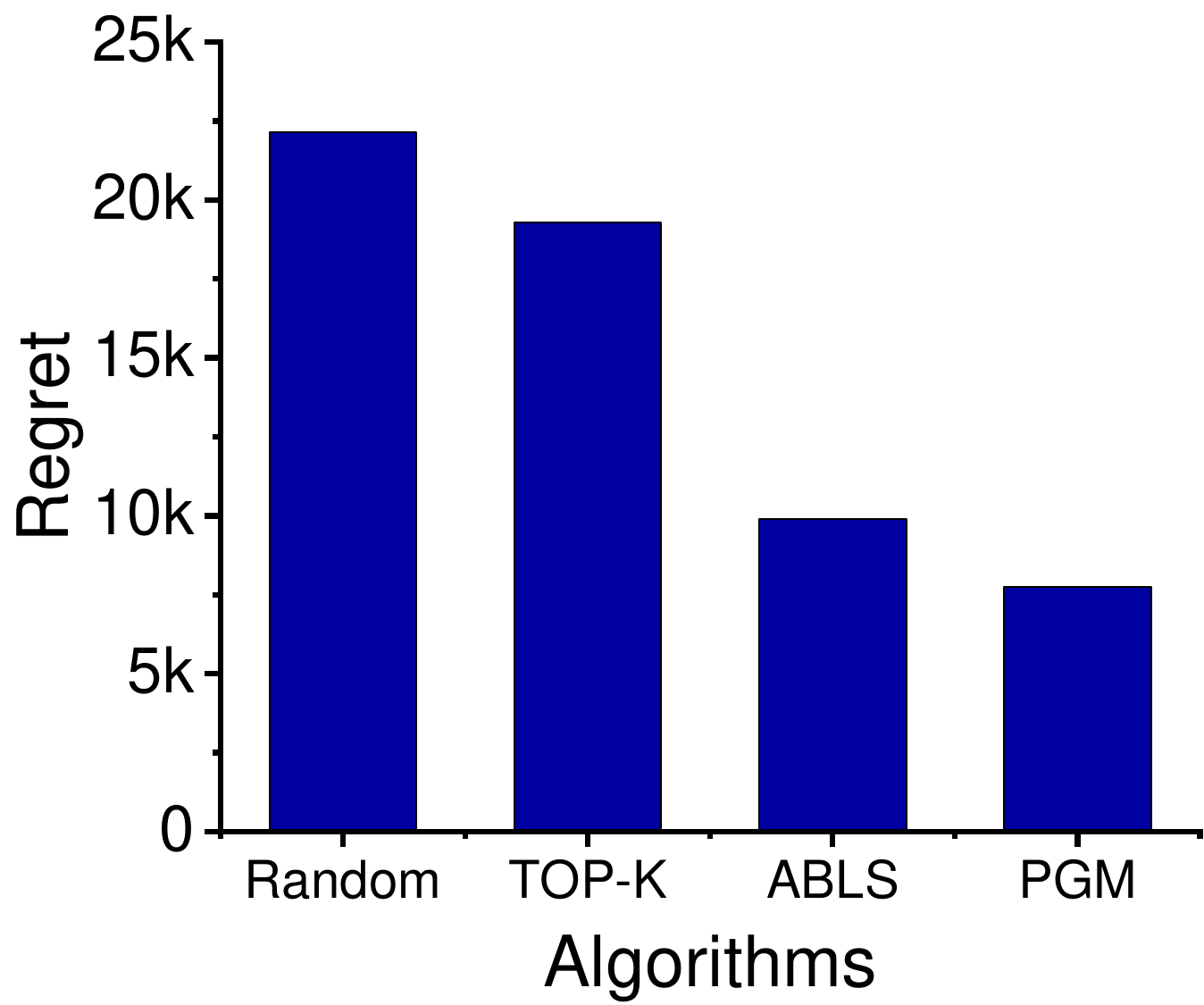} \\
        {\tiny (a) $\alpha = 40 \%$} &
        {\tiny (b) $\alpha = 60 \%$} &
        {\tiny (c) $\alpha = 80 \%$} &
        {\tiny (d) $\alpha = 100 \%$} &
        {\tiny (e) $\alpha = 120 \%$} \\
        \includegraphics[width=0.185\linewidth]{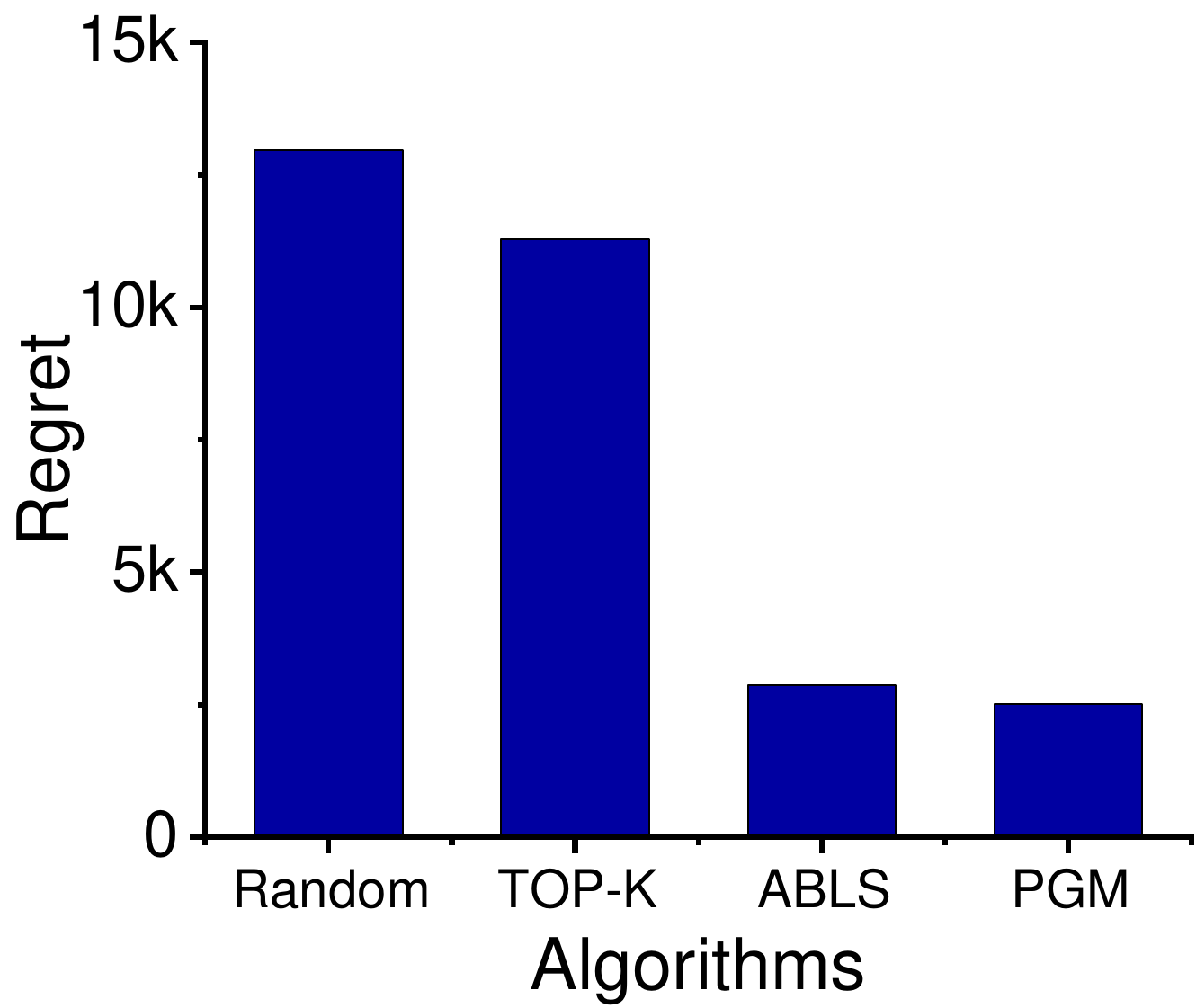} &
        \includegraphics[width=0.185\linewidth]{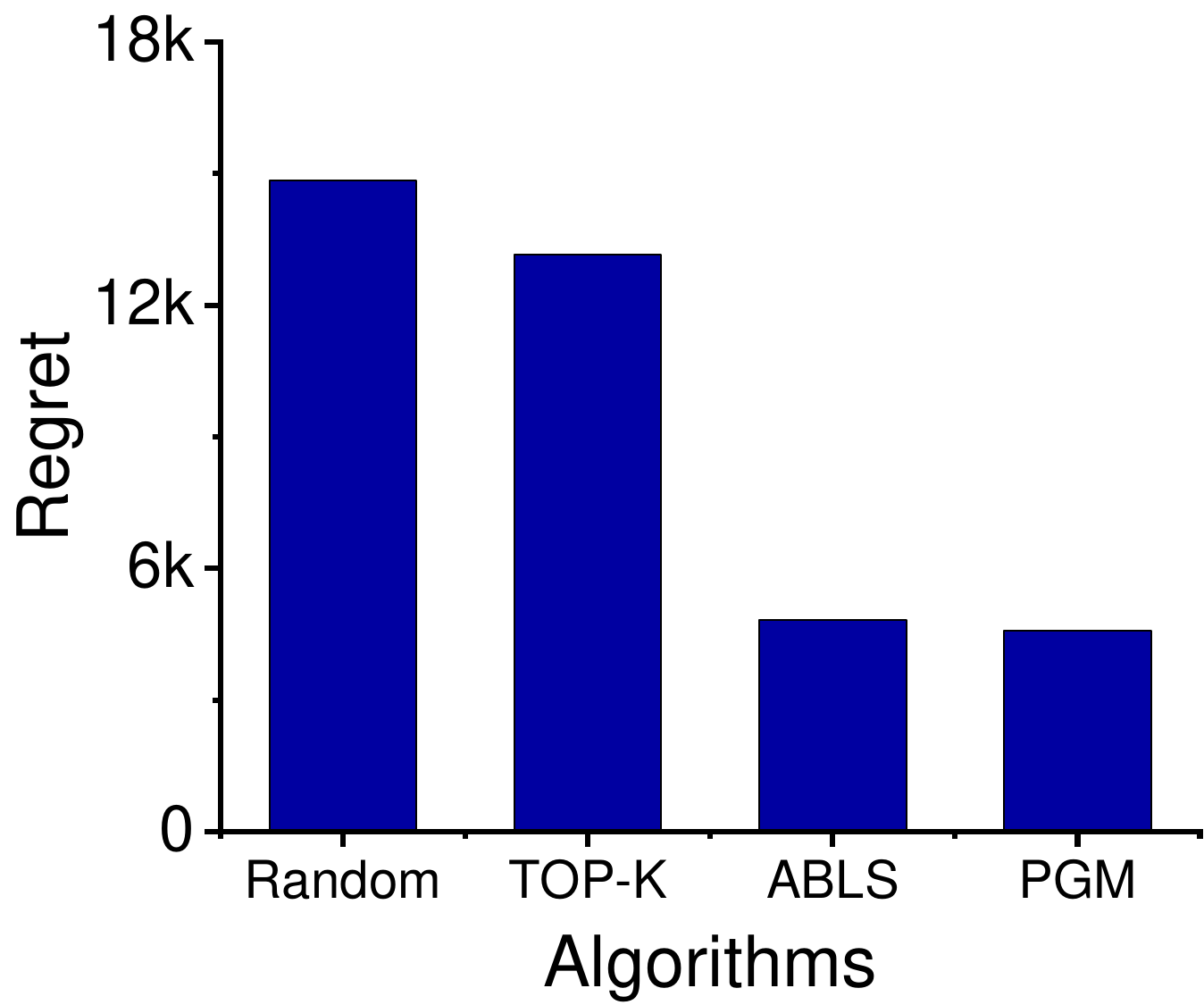} &
        \includegraphics[width=0.185\linewidth]{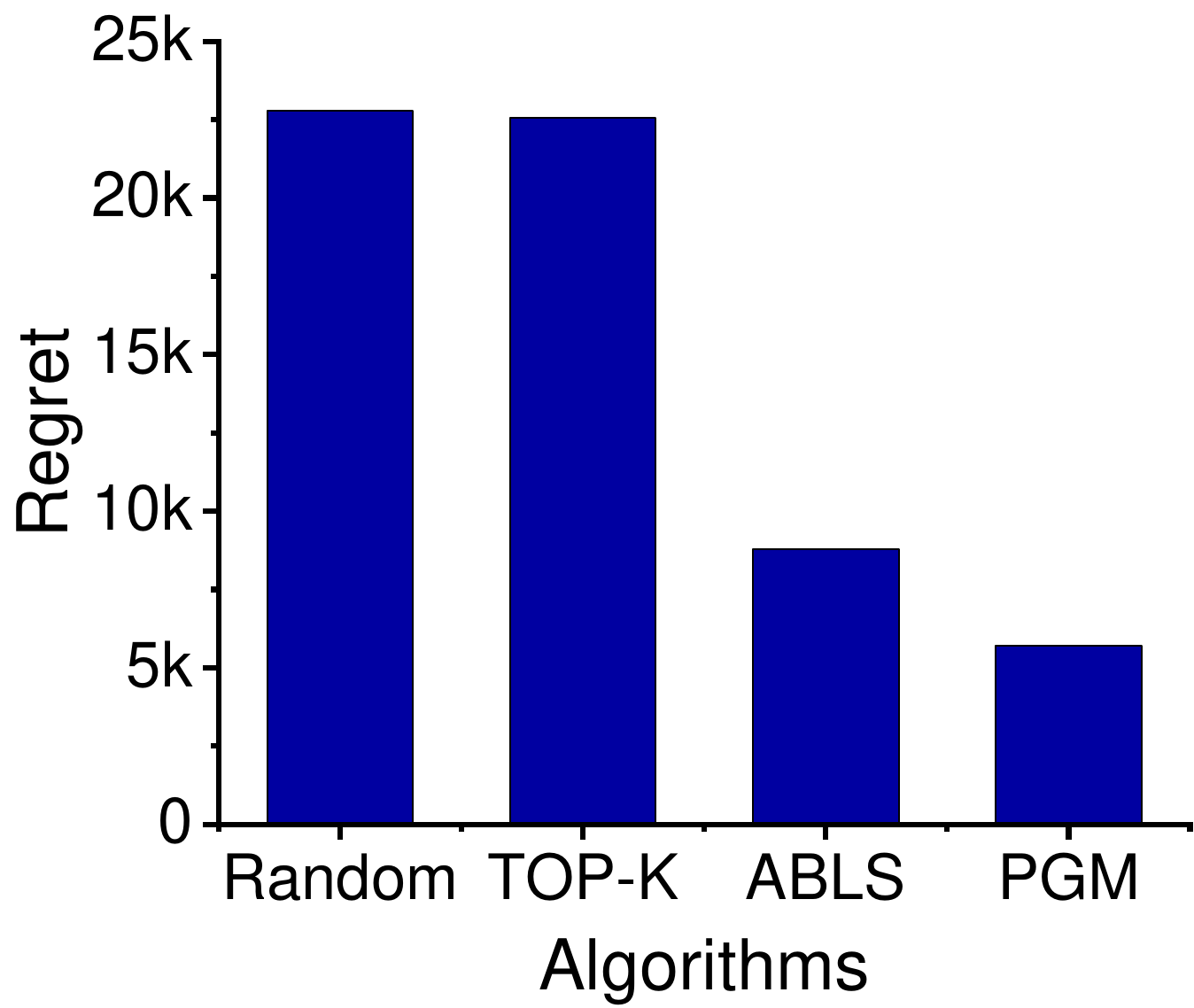} &
        \includegraphics[width=0.185\linewidth]{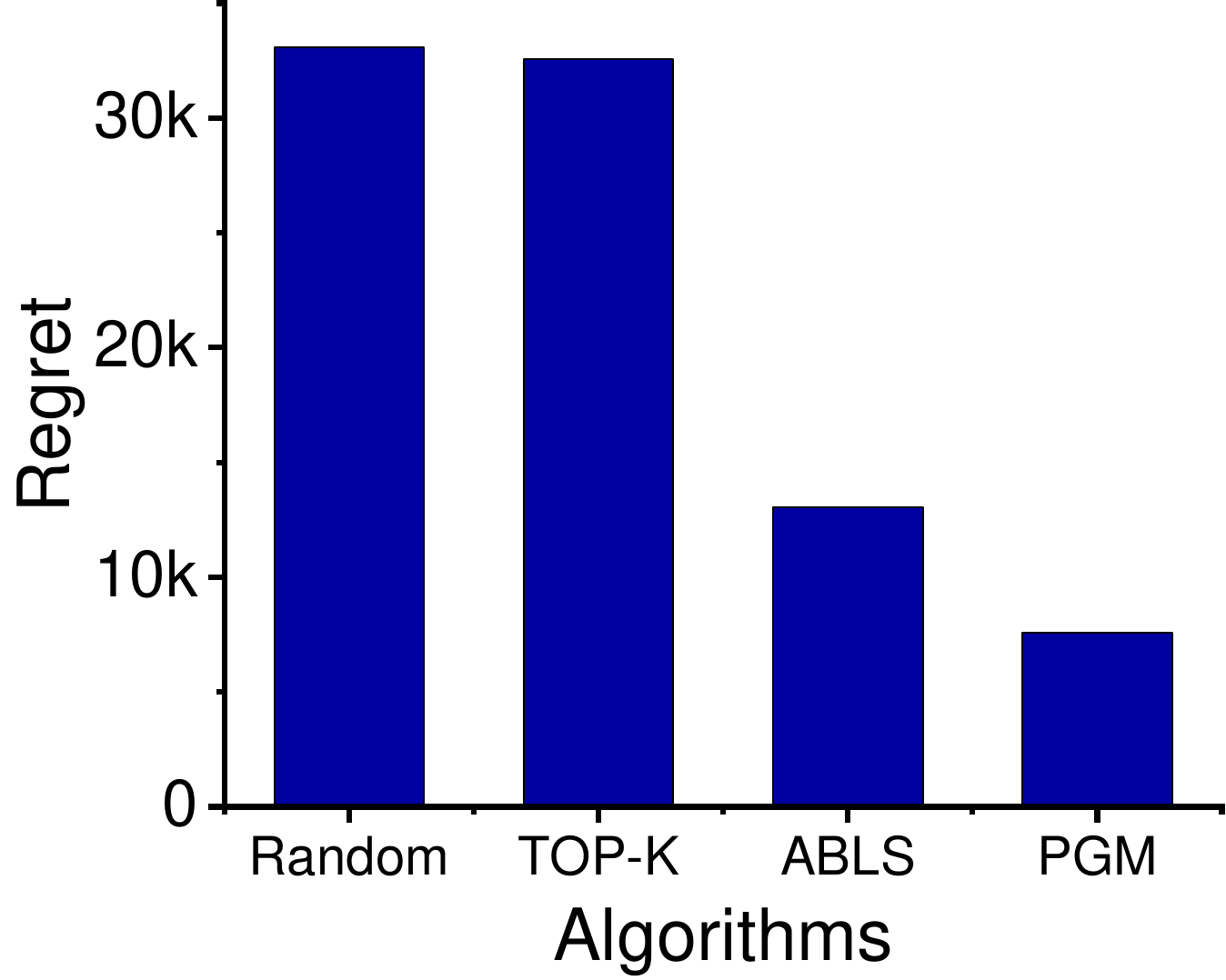} &
        \includegraphics[width=0.185\linewidth]{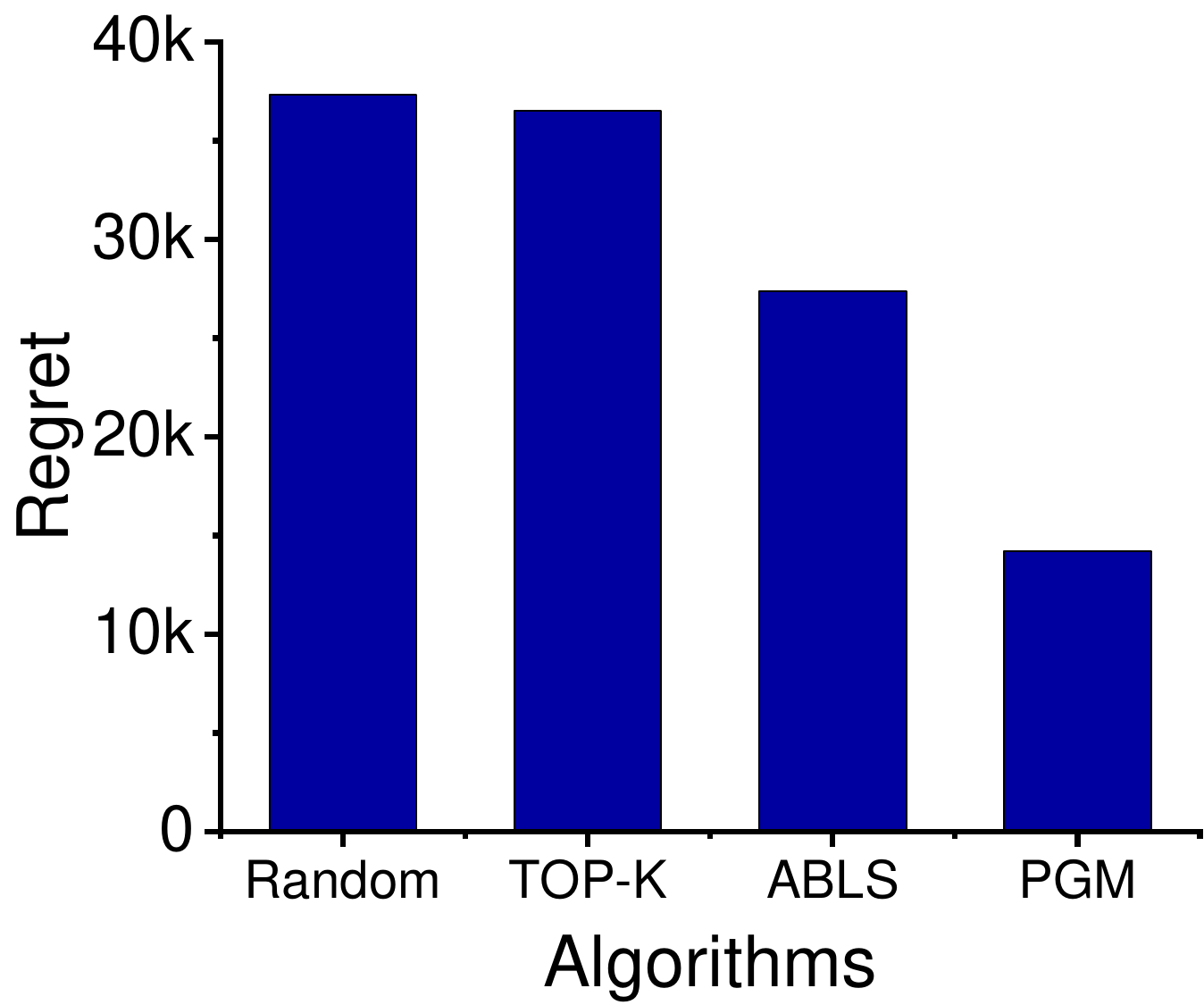} \\
        {\tiny (f) $\alpha = 40 \%$} &
        {\tiny (g) $\alpha = 60 \%$} &
        {\tiny (h) $\alpha = 80 \%$} &
        {\tiny (i) $\alpha = 100 \%$} &
        {\tiny (j) $\alpha = 120 \%$} \\
    \end{tabular}
    \caption{Regret on varying $\alpha$, when $\lambda = 1\%, \mathcal{|A|} = 100$$(a-e)$,   when $\lambda = 20\%, \mathcal{|A|} = 5$ $(f-k)$, for Trivalency Settings}
    \label{Fig:Regret3}
\end{figure*}
\subsection{Key Parameters.}
In this study, the values of the following parameters need to be fixed, and we describe them briefly. All key parameters and their corresponding values are summarized in Table \ref{Table-2}.

\begin{table}[h!]
\caption{\label{Table-2} Key Parameters}
\vspace{-0.15 in}
\begin{center}
    \begin{tabular}{ | p{2cm}| p{5.5cm}|}
    \hline
    Parameter & Values  \\ \hline
    $\alpha$ & $40\%, 60\%, 80\%, \textbf{100\%}, 120\%$   \\ \hline
    $\lambda$ & $1\%, 2\%, \textbf{5\%}, 10\%, 20\%$  \\ \hline
    $\gamma$ & $0, 0.25, \textbf{0.5}, 0.75, 1$  \\ \hline
    $\delta$ & $0, 0.25, \textbf{0.5}, 0.75, 1$  \\ \hline
    $\epsilon$ & $0, 0.01, \textbf{0.05}, 0.075, 0.1$  \\ \hline
    $\rho$ & $0, 0.2, \textbf{0.5}, 0.7, 1$  \\ \hline
    $\pi$ & $25m,50m,\textbf{100m},125m,150m$  \\ \hline
    \end{tabular}
\end{center}
\end{table}

\paragraph{\textbf{Demand-Supply Ratio $(\alpha)$.}}
It denotes the ratio of global influence demand over the influence supply by the influence provider, i.e., $\alpha = I^{d}/I^{s}$, where  $I^{d} = \sum_{a \in \mathcal{A}} I^a$ is the global influence demand, $I^{s} =  \sum_{b \in \mathcal{BS}} \mathcal{I}(\{b\}) + \sum_{p \in \mathcal{P}} \mathcal{I}^{G}(\{p\})$ is influence supply from influence provider.

\paragraph{\textbf{Average-Individual Demand Ratio $(\lambda)$.}}
It is the percentage of average influence demand over the influence provider supply i.e., $\lambda = I^{\bar{d}}/I^{s}$, where $I^{\bar{d}} = I^{d}/|\mathcal{A}|$ is the average individual demand of the advertiser.

\paragraph{\textbf{Advertiser’s Demand $(I)$.}}
Once the average individual demand is fixed, we can calculate the advertiser demand $I^{a} = \lfloor \omega \cdot  I^{s} \cdot \lambda \rfloor$, where $\omega$ is a random choice factor between $0.8$ and $1.2$ to generate different demand for the advertisers.

\paragraph{\textbf{Advertiser’s Payment $(\mathcal{L})$.}}
Following widely adapted settings in regret maximization \cite{ali2024minimizing,ali2024toward,zhang2021minimizing} and economic studies \cite{aslay2017revenue,aslay2015viral,banerjee2019maximizing}, we set each advertiser's payment to be proportional to their influence demand, i.e.,  $\mathcal{L}_{i} = \lfloor \beta \cdot I^{a} \rfloor$, where $\beta$ is the random factor chosen between $0.9$ to $1.1$.

\paragraph{\textbf{Unsatisfied Penalty Ratio $(\gamma)$.}}
Recall Definition \ref{Def:Combine_Regret_Model}, where $\gamma \in [0, 1]$ denotes the proportion of the payment penalty incurred when the advertiser's demand is not met.

\paragraph{\textbf{Cardinality Penalty Ratio $(\delta)$.}}
This ratio $\delta$ defines the penalty imposed for allocating the maximum number of slots or seed nodes, where $\delta \in [0, 1]$.

\paragraph{\textbf{Interaction Parameter $(\rho)$}.}
The $\rho$ controls the strength of interaction between billboard slots and social media seed nodes. We vary the value of $\rho$ between $0$ and $1$.

\paragraph{\textbf{Regret Tolerance Parameter $(\epsilon)$.}}
We vary the $\epsilon$ value from $0$ to $0.1$ to simulate the trade-off between regret and runtime.

\paragraph{\textbf{Distance Parameter $(\pi)$.}}
The distance parameter $\pi$ decides the influence of a billboard slot. We vary the value of $\pi$ between $25$ meters $150$ meters.

\subsection{Baseline Methodologies.} \label{Sec:Baseline}
\paragraph{\textbf{Random Allocation (RA)}}
The billboard slots and social media seeds are selected randomly. It picks nodes randomly without considering any influence maximization criteria and stops when the budget is used up.
\paragraph{\textbf{Top-k Allocation.}} In this approach, most influential billboard slots and seed nodes are selected till their respective demand and budget constraints are satisfied.

\subsection{\textbf{Goals of Our Experimentation.}}
The following research question is addressed in our study.
\begin{itemize}
    \item \textbf{RQ1.} What will happen if the influence provider's influence supply is below, near, or exceeds the global demand of the advertisers in a multi-mode advertisement setting? 
    \item \textbf{RQ2.} Which kind of advertisers, i.e., fewer advertisers with high influence demand or a large number of advertisers with small influence demand, are more beneficial? 
    \item \textbf{RQ3.} How does the interaction effect affect the selection of slots and reduce total regret? Also how the parameters like $\epsilon$, $\delta$, T, $\gamma$ impacts on minimizing regret.
\end{itemize}

\subsection{Results and Discussions.} \label{Sec:RD} 
In this work, we consider four cases to show the effectiveness of our proposed approach. 
\paragraph{\textbf{Effectiveness Study.}} The efficiency of the proposed solutions is described in four different cases.
\paragraph{Case 1:~ Low $\alpha$, Low $\lambda$.}
Corresponding to case $1$, we have $\alpha \leq 80\%$ and $\lambda \leq 2\%$. This represents the situation where global demand and individual influence demand are low, i.e., influence providers have more influence to supply than the demand from the advertiser. That is why the number of satisfied advertisers is higher, except for `Top-k' and `Random' approaches. We have three main observations. First, with the increase of $\alpha$, the individual demand of the advertisers increases, and the number of satisfied advertisers decreases. Second, the `PGM' and `ABLS' outperform `Random' and `Top-k', i.e., reduce the regret better because they are able to satisfy advertisers with fewer slots or seeds. Third, as there is a large number of advertisers with small individual demand, each advertiser gets more influence than required, which leads some of the advertisers to become dissatisfied.
\paragraph{Case 2:~ Low $\alpha$, High $\lambda$.} Corresponding to case $2$ we have $\alpha \leq 80\%$ and $\lambda \geq 5\%$. This refers to the situation where global demand is still lower than the supply. However, individual demand is much higher. We have two observations. First, when the global demand becomes low but individual demand is high, the excessive influence supply to advertisers drops. This happens because with the increase of $\beta$ value, the number of advertisers is less, and individual influence demand is higher. That influence demand is closer to the influence provided by slots or seeds, and consequently, the extra influence supply to the advertisers decreases. Second, with the higher individual influence demand from advertisers, the influence provider could deploy more slots or seeds using `PGM' and `ABLS' to minimize the regret, as shown in Figure \ref{Fig:Regret-Uniform} for three different probability settings: uniform, trivalency, and weighted cascade. Among them, trivalency reduces the regret better than the other probability settings because trivalency models simulate realistic variations in influence strength between different individuals as shown in Figure \ref{Fig:Regret-Uniform}$(a,b,c,f,g,h,k,\ell,m)$ and Figure \ref{Fig:Regret3}(i,j).
\paragraph{Case 3:~ High $\alpha$, Low $\lambda$.}
Corresponding to case $3$, we have $\alpha \geq 100\%$ and $\lambda \leq 2\%$. With the increase of the $\alpha$ value, the global influence demand is high and individual demand is low. We have two main observations. First, as global influence demand is high, none of the algorithms can satisfy all the advertisers, and dissatisfaction increases. Second, when $\alpha \geq 100\%$, the influence supply is equal to or less than the demand from the advertisers. As we cannot control the excessive influence supply, the regret increases in this case. This behavior can be observed in Figure \ref{Fig:Regret3}(d,e). 

\paragraph{Case 4:~ High $\alpha$, High $\lambda$.}
Corresponding to case $4$, we have $\alpha \geq 100\%$ and $\lambda \geq 5\%$ defines a situation where both global and individual influence demand is high, i.e., demand comes from a few advertisers with high influence demand. We have two main observations.
First, large $\alpha$ and $\lambda$ value leads to higher regret for each advertiser. Hence, all algorithms suffer from higher regret.
second, when $\lambda$  value increases from $5\%$ to $20\%$ all algorithms suffer from higher regret as shown in Figure \ref{Fig:Regret-Uniform}(d,e,i,j,n,o) and Figure \ref{Fig:Regret3}(j,k). So, a smaller number of advertisers with higher influence demand is not very beneficial for the influence provider. Instead, a large number of advertisers with small individual demand is more beneficial for the influence provider.

\paragraph{\textbf{Efficiency Study.}} Efficiency is important as every day thousands of advertisers come to the influence provider with the required influenced demand. Hence, we have conducted an efficiency evaluation under different cases of global and individual influence demand. we have some observations. First, with the increase of $\alpha$, each proposed and baseline method requires extra search time to deploy slots or seeds to the advertisers. We have presented the time requirements for `PGM', `ABLS', `Top-k', and `Random' approaches for $\lambda = 5\%, \mathcal{|A|} = 20$ in the Figure \ref{Fig:Regret2}(a,b,c) for uniform, weighted, and trivalency probability settings, respectively. We observe that with the increase of the number of advertisers, the computational time increases, and this change happens drastically for the `PGM' approach as shown in Figure \ref{Fig:Regret2}(a,b,c). When $|\mathcal{A}| = 20$ the `ABLS' takes more run time compared to `PGM' and baselines. However, when $|\mathcal{A}| = 100$, `PGM' takes the highest computational time as shown in Figure \ref{Fig:Regret2}(c,d).
\paragraph{\textbf{Scalability Test.}}
The show the scalability of our proposed solution approaches `PGM' and `ABLS' we experimented with two extreme cases when  $\lambda = 1\%, \mathcal{|A|} = 100$$(a, b, c, d, e)$ and when $\lambda = 20\%, \mathcal{|A|} = 5$ $(f,g,h,j,k)$, for trivalency probability settings. In both cases, `PGM' and `ABLS' outperform the baseline in terms of minimizing regret as shown in Figure \ref{Fig:Regret3}(a-j). One point needs to be highlighted: with the increase of the demand-supply ratio $\alpha$, the number of satisfied advertisers decreases. Hence, the regret increases as shown in Figure \ref{Fig:Regret2}(f-j).

\paragraph{\textbf{Other Parameter Study.}}
We analyze the impact of key parameters: (a) Increasing $\epsilon$ reduces runtime but increases regret due to stricter selection; (b) Higher $\delta$ penalizes large allocations, leading to compact but less effective solutions; (c) Larger $\gamma$ emphasizes demand satisfaction, lowering regret at higher resource use; (d) More iterations $T$ in PGM improve regret marginally after a point, with increased computation cost. (e) With increasing $\rho$ value, the strength of interaction between slots and seed nodes also increases. (f) With an increase of $\pi$, the influence of the slots increases as one slot can influence a larger number of trajectories. We set $\alpha = 100\%, \lambda = 5\%$, $\gamma = 0.5$, $\delta = 0.5$, $\epsilon = 0.05$ $\rho = 0.5$ and $\pi = 100$ meter as the default setting for our experiment.

\section{Conclusion and Future Directions} \label{Sec:CFD}
 In this work, we studied the problem of regret minimization in multi-mode advertising, where billboard slots and social media seeds are allocated jointly to multiple advertisers. We introduce a novel \textit{Interaction Effect} formulation and propose a new regret model that captures both individual and combined effects of two different advertising media and their interaction. We proved that the problem is NP-hard and inapproximable, and we introduced two solution approaches: the Projected Subgradient Method (PGM) and the Approximate Bisubmodular Local Search (ABLS). Experiments on large-scale real-world datasets demonstrated that our methods consistently reduce total regret compared to existing baselines across various demand settings and influence probability models. Our flexible framework can be adapted to other resource allocation problems beyond advertising. The current methods (PGM and ABLS) work efficiently on large datasets; however, this work could be extended for streaming and online settings, where advertisers arrive dynamically. Another potential direction could be the development of parallel and distributed optimization frameworks for scalability in real-world advertising platforms.

\bibliographystyle{ACM-Reference-Format}
\bibliography{sample-bibliography} 
\end{document}